\documentclass[acmsmall,screen]{acmart}

\setcopyright{cc}
\setcctype{by-sa}
\acmPrice{}
\acmDOI{10.1145/3563344}
\acmYear{2022}
\copyrightyear{2022}
\acmSubmissionID{oopslab22main-p535-p}
\acmJournal{PACMPL}
\acmVolume{6}
\acmNumber{OOPSLA2}
\acmArticle{181}
\acmMonth{10}

\usepackage{booktabs}   %% For formal tables:
\usepackage{subcaption} %% For complex figures with subfigures/subcaptions
\usepackage{mathtools}
\usepackage{thm-restate}
\usepackage[noend]{algpseudocode}
\usepackage{algorithm}
\usepackage{forest}
\usepackage[finalizecache=false,frozencache=true,cachedir=minted-cache]{minted}
\usepackage{multicol}
\usepackage{multirow}
\usepackage{mathtools}
\usepackage{xspace}
\usepackage{xcolor, soul}
\usepackage{wrapfig}
\usepackage{tikz}
\usepackage{pgf}
\usetikzlibrary{calc,fit}
\usepackage{cleveref}
\usepackage{enumitem }
\setlist[itemize]{leftmargin=*}

\usepackage{etoolbox}
\newtoggle{SUPPLEMENTAL}
\toggletrue{SUPPLEMENTAL}

\crefname{section}{\S}{\S}
\Crefname{section}{\S}{\S}
\crefformat{section}{\S#2#1#3}

\definecolor{darkgreen}{RGB}{0,128,0}
\definecolor{gray}{RGB}{200,200,200}

\newcommand{\SYSTEM}{\textsc{Plinko}\xspace}
\newcommand{\Q}[1]{(\textit{\textbf{Q#1}})}

\DeclarePairedDelimiter{\abs}{\lvert}{\rvert}
\newcommand{\DeclareAutoPairedDelimiter}[3]{%
	\expandafter\DeclarePairedDelimiter\csname Auto\string#1\endcsname{#2}{#3}%
	\begingroup\edef\x{\endgroup
		\noexpand\DeclareRobustCommand{\noexpand#1}{%
			\expandafter\noexpand\csname Auto\string#1\endcsname*}}%
	\x}
\DeclareAutoPairedDelimiter{\deno}{\llbracket}{\rrbracket}

\newcommand*\Let[2]{\State #1 $\gets$ #2}

\newcommand*\dom[1]{\mathrm{dom}\left( #1 \right)}
\newcommand*\E[1]{\mathbb{E}\left[ #1 \right]}

\algrenewcommand\algorithmicindent{0.75em}%

\algnewcommand\algorithmicswitch{\textbf{switch}}
\algnewcommand\algorithmiccase{\textbf{case}}
\algnewcommand\algorithmicassert{\texttt{assert}}
\algnewcommand\Assert[1]{\State \algorithmicassert(#1)}%

\algdef{SE}[SWITCH]{Switch}{EndSwitch}[1]{\algorithmicswitch\ #1\ \algorithmicdo}{\algorithmicend\ \algorithmicswitch}%
\algdef{Se}[CASE]{Case}{EndCase}[1]{\algorithmiccase\ #1}%
\algtext*{EndSwitch}%
\algtext*{EndCase}%

\makeatletter
\tikzset{%
	remember picture with id/.style={%
		remember picture,
		overlay,
		save picture id=#1,
	},
	save picture id/.code={%
		\edef\pgf@temp{#1}%
		\immediate\write\pgfutil@auxout{%
			\noexpand\savepointas{\pgf@temp}{\pgfpictureid}}%
	},
	if picture id/.code args={#1#2#3}{%
		\@ifundefined{save@pt@#1}{%
			\pgfkeysalso{#3}%
		}{
			\pgfkeysalso{#2}%
		}
	}
}

\def\savepointas#1#2{%
	\expandafter\gdef\csname save@pt@#1\endcsname{#2}%
}

\def\tmk@labeldef#1,#2\@nil{%
	\def\tmk@label{#1}%
	\def\tmk@def{#2}%
}

\tikzdeclarecoordinatesystem{pic}{%
	\pgfutil@in@,{#1}%
	\ifpgfutil@in@%
	\tmk@labeldef#1\@nil
	\else
	\tmk@labeldef#1,(0pt,0pt)\@nil
	\fi
	\@ifundefined{save@pt@\tmk@label}{%
		\tikz@scan@one@point\pgfutil@firstofone\tmk@def
	}{%
		\pgfsys@getposition{\csname save@pt@\tmk@label\endcsname}\save@orig@pic%
		\pgfsys@getposition{\pgfpictureid}\save@this@pic%
		\pgf@process{\pgfpointorigin\save@this@pic}%
		\pgf@xa=\pgf@x
		\pgf@ya=\pgf@y
		\pgf@process{\pgfpointorigin\save@orig@pic}%
		\advance\pgf@x by -\pgf@xa
		\advance\pgf@y by -\pgf@ya
	}%
}

\makeatother
\newlength\AlgIndent
\newcounter{mymark}
\newcommand\ColorLine[1]{%
	\stepcounter{mymark}%
	\tikz[remember picture with id=mark-\themymark,overlay] {;}%
	\begin{tikzpicture}[remember picture,overlay]%
		\filldraw[#1]%
		let \p1=(pic cs:mark-\themymark), 
		\p2=(current page.east)  in 
    ([xshift=-0.4em,yshift=-1.0ex]0,\y1)  rectangle ++([xshift=-7.7cm]\x2,\baselineskip);
	\end{tikzpicture}%
}%

\makeatletter
\algnewcommand\cRequire[1]{\item[\ColorLine{#1}\algorithmicrequire]}%
\algnewcommand\cEnsure[1]{\item[\ColorLine{#1}\algorithmicensure]}%
\algnewcommand\cState[1]{\State\ColorLine{#1}}%
\algnewcommand\cStatex[1]{\Statex\ColorLine{#1}}%
\algnewcommand\cComment[1]{\Comment\ColorLine{#1}}%

\algdef{SE}[WHILE]{cWhile}{ENDWHILE}%
[2]{\ColorLine{#2}\algorithmicwhile\ #1\ \algorithmicdo}%
{\algorithmicend\ \algorithmicwhile}%
\algdef{SE}[FOR]{cFor}{ENDFOR}%
[2]{\ColorLine{#2}\algorithmicfor\ #1\ \algorithmicdo}%
{\algorithmicend\ \algorithmicfor}%
\algdef{S}[FOR]{cForAll}%
[2]{\ColorLine{#2}\algorithmicforall\ #1\ \algorithmicdo}%
\algdef{SE}[LOOP]{cLoop}{ENDLOOP}%
[1]{\ColorLine{#1}\algorithmicloop}%
{\algorithmicend\ \algorithmicloop}%
\algdef{SE}[REPEAT]{cRepeat}{UNTIL}%
[1]{\ColorLine{#1}\algorithmicrepeat}%
[1]{\algorithmicuntil\ #1}%
\algdef{SE}[IF]{cIf}{cEndIf}%
[2]{\ColorLine{#2}\algorithmicif\ #1\ \algorithmicthen}%
[1]{\ColorLine{#1}\algorithmicend\ \algorithmicif}%
\algdef{C}[IF]{IF}{cElsif}%
[2]{\ColorLine{#2}\algorithmicelse\ \algorithmicif\ #1\ \algorithmicthen}%
\algdef{Ce}[ELSE]{IF}{cElse}{ENDIF}%
[1]{\ColorLine{#1}\algorithmicelse}%
\algdef{SE}[SWITCH]{cSwitch}{cEndSwitch}%
[2]{\ColorLine{#2}\algorithmicswitch\ #1\ \algorithmicdo}%
{\algorithmicend\ \algorithmicswitch}%
\algdef{SE}[CASE]{cCase}{cEndCase}%
[2]{\ColorLine{#2}\algorithmiccase\ #1}%
[1]{\ColorLine{#1}\algorithmicend\ \algorithmiccase}%
\makeatother

\renewcommand{\fboxsep}{1.25pt}

\begin{document}

%% Title information
\title{Symbolic Execution for Randomized Programs}         %% [Short Title] is optional;
%% when present, will be used in
%% header instead of Full Title.
% \titlenote{with title note}             %% \titlenote is optional;
%% can be repeated if necessary;
%% contents suppressed with 'anonymous'
% \subtitle{Subtitle}                     %% \subtitle is optional
% \subtitlenote{with subtitle note}       %% \subtitlenote is optional;
%% can be repeated if necessary;
%% contents suppressed with 'anonymous'

%% Author information
%% Contents and number of authors suppressed with 'anonymous'.
%% Each author should be introduced by \author, followed by
%% \authornote (optional), \orcid (optional), \affiliation, and
%% \email.
%% An author may have multiple affiliations and/or emails; repeat the
%% appropriate command.
%% Many elements are not rendered, but should be provided for metadata
%% extraction tools.

%% Author with single affiliation.
% \author{First Last}
% \authornote{with author1 note}          %% \authornote is optional;
%                                         %% can be repeated if necessary
% \orcid{nnnn-nnnn-nnnn-nnnn}             %% \orcid is optional
% \affiliation{
% \position{Position1}
% \department{Department1}              %% \department is recommended
% \institution{Institution1}            %% \institution is required
% \streetaddress{Street1 Address1}
% \city{City1}
% \state{State1}
% \postcode{Post-Code1}
% \country{Country1}                    %% \country is recommended
% }
%   \email{first1.last1@inst1.edu}          %% \email is recommended

\author{Zachary Susag}
% \authornote{with author1 note}          %% \authornote is optional;
%% can be repeated if necessary
\orcid{0000-0002-3981-9529}             %% \orcid is optional
\affiliation{
  % \position{Position1}
  % \department{Department1}              %% \department is recommended
  \institution{Cornell University}            %% \institution is required
  % \streetaddress{Street1 Address1}
  % \city{City1}
  % \state{State1}
  % \postcode{Post-Code1}
  \country{USA}                    %% \country is recommended
}
\email{zjs@cs.cornell.edu}          %% \email is recommended

\author{Sumit Lahiri}
% \authornote{with author1 note}          %% \authornote is optional;
%% can be repeated if necessary
\orcid{0000-0002-6867-9035}             %% \orcid is optional
\affiliation{
  % \position{Position1}
  % \department{Department1}              %% \department is recommended
  \institution{IIT Kanpur}            %% \institution is required
  % \streetaddress{Street1 Address1}
  % \city{City1}
  % \state{State1}
  % \postcode{Post-Code1}
  \country{India}                    %% \country is recommended
}
\email{sumitl@cse.iitk.ac.in}          %% \email is recommended

\author{Justin Hsu}
% \authornote{with author1 note}          %% \authornote is optional;
%% can be repeated if necessary
\orcid{0000-0002-8953-7060}             %% \orcid is optional
\affiliation{
  % \position{Position1}
  % \department{Department1}              %% \department is recommended
  \institution{Cornell University}            %% \institution is required
  % \streetaddress{Street1 Address1}
  % \city{City1}
  % \state{State1}
  % \postcode{Post-Code1}
  \country{USA}                    %% \country is recommended
}
\email{justin@cs.cornell.edu}          %% \email is recommended

\author{Subhajit Roy}
% \authornote{with author1 note}          %% \authornote is optional;
%% can be repeated if necessary
\orcid{0000-0002-3394-023X}             %% \orcid is optional
\affiliation{
  % \position{Position1}
  % \department{Department1}              %% \department is recommended
  \institution{IIT Kanpur}            %% \institution is required
  % \streetaddress{Street1 Address1}
  % \city{City1}
  % \state{State1}
  % \postcode{Post-Code1}
  \country{India}                    %% \country is recommended
}
\email{subhajit@cse.iitk.ac.in}          %% \email is recommended

%% Abstract
%% Note: \begin{abstract}...\end{abstract} environment must come
 %% before \maketitle command
\begin{abstract}
  We propose a symbolic execution method for programs that can draw random
  samples. In contrast to existing work, our method can verify randomized
  programs with unknown inputs and can prove probabilistic properties that
  universally quantify over all possible inputs. Our technique augments standard
  symbolic execution with a new class of \emph{probabilistic symbolic
  variables}, which represent the results of random draws, and computes symbolic
  expressions representing the probability of taking individual paths. We
  implement our method on top of the \textsc{KLEE} symbolic execution engine
  alongside multiple optimizations and use it
  to prove properties about probabilities and expected values for a range of
  challenging case studies written in C++, including Freivalds' algorithm,
  randomized quicksort, and a randomized property-testing algorithm for
  monotonicity. We evaluate our method against \textsc{Psi}, an exact
  probabilistic symbolic inference engine, and \textsc{Storm}, a probabilistic
  model checker, and show that our method significantly outperforms both tools.
\end{abstract}

%% 2012 ACM Computing Classification System (CSS) concepts
%% Generate at 'http://dl.acm.org/ccs/ccs.cfm'.


%% Keywords
%% comma separated list
  %% \keywords are mandatory in final camera-ready submission

%% \maketitle
%% Note: \maketitle command must come after title commands, author
%% commands, abstract environment, Computing Classification System
%% environment and commands, and keywords command.
\maketitle

\section{Introduction}
\label{sec:intro}

Symbolic execution (SE)~\citep{king1976} is a highly successful method to
automatically find bugs in programs. In a nutshell, SE iteratively explores the
space of possible program paths while treating program states as \emph{symbolic
functions} of program inputs (symbolic parameters). Whenever a path to an error
state is discovered, SE searches for a concrete setting of the symbolic
parameters which realizes this path, thereby triggering the error. This
basic bug-finding strategy has proved to be a powerful technique, supporting
some of the most effective tools for analyzing large
codebases~\citep{bessey_2010}.

Researchers have adapted SE to many
domains~(e.g.,~\citep{geldenhuys_2012, filieri_2013, borges_2014, p4wn_2021,
farina_2019, raimondas_2010, raimondas_2011, adam_2012}). In this work, we
consider how to perform SE for \emph{randomized} programs, which can draw random
samples from built-in distributions. These programs play a critical role in many
important applications today, from machine learning to security and privacy.
However, randomized programs, like all programs, are susceptible to
bugs~\citep{axprof_2019}.  Correctness properties are difficult to formally
verify; existing methods typically require substantial manual effort and target
narrow properties. Moreover, randomized programs are difficult to test: the
desired behavior is not deterministic, and it may not be possible to detect that
a program is producing an incorrect distribution of outputs by running the program.

\paragraph*{Challenges and prior work.} 
Accordingly, randomized programs are an attractive target for formal
verification methods. However, there are several challenges in developing a SE
procedure for randomized programs:
\begin{itemize}
\item \textbf{Quantitative properties.} Standard SE aims to check whether a bad
  program state is reachable. In
  randomized programs, we are often more interested in whether a bad state is
  reached \emph{too often}, or whether a good state is reached \emph{often
  enough}. Accordingly, SE for randomized programs must be able to analyze the
  \emph{quantitative} probability of reaching program states.
\item \textbf{Computing branch probabilities.} A concrete execution of a
  non-probabilistic program follows one branch at each branching instruction,
  since the branch condition is either true or false in any program state. In
  probabilistic programs, the program can be thought of as transforming a
  \emph{distribution} of program states. Thus, a branch condition has some
  probability of being true and some probability of being false in every probabilistic
  program state. A SE procedure for probabilistic programs should
  be able to compute probabilities of branches and combine them
  to reason about probabilities of whole paths.
\item \textbf{Handling unknown inputs.} Like standard programs, randomized
  programs usually accept input parameters. These unknown values are qualitatively
  different from the unknown results of random sampling commands: there is no
  probability assigned to different settings of the inputs, and the target
  correctness property universally quantifies over all possible inputs.

  Prior work has considered SE for probabilistic
  programs~\citep{geldenhuys_2012,sampson_2014}, but unlike traditional SE, most
  existing methods do not treat the program inputs as \textit{unknown} values.
  Instead, existing methods typically consider probabilistic programs without
  inputs~\citep{sampson_2014} or assume that the inputs are
  drawn from a known probability distribution~\citep{geldenhuys_2012}. This simplification
  allows existing SE methods to compute probabilities exactly using model
  counting or approximately using statistical sampling. However, it
  significantly limits the kinds of programs and properties that can be
  analyzed.
\end{itemize}

\paragraph*{Our work.}
We propose a symbolic execution method for randomized programs with unknown
inputs. These programs can be thought of as producing a family of output
distributions, one for each concrete input, with inputs ranging over a large,
possibly infinite set. We consider two kinds of properties:
\begin{itemize}
\item \textbf{Probability bounds.} For all inputs, the \emph{probability of
  reaching a program state} is at most/at least/exactly equal to some quantity,
  which may depend on the inputs. These properties are the probabilistic
  analog of the reachability properties considered by traditional SE.
\item \textbf{Expectation bounds.} For all inputs, the \emph{expected value of a
  program expression} in the output is at most/at least/exactly equal to some
  quantity, which may depend on the inputs. These properties are useful for
  bounding expected resource usage or running time.
\end{itemize}
There are two key technical ingredients in our approach:
\begin{itemize}
\item \textbf{Distinguishing between regular and probabilistic symbolic variables.}
  Our method models sampling statements by introducing a new kind of \emph{probabilistic}
  symbolic variable. While regular symbolic variables represent unknown inputs,
  probabilistic symbolic variables represent the result from sampling a known
  distribution.
\item \textbf{Computing symbolic branch probabilities.} Branch probabilities may
  depend on unknown program inputs. Accordingly, we cannot use approaches like
  model counting to compute the branch probabilities since the probabilities
  are not constants. Instead, our method computes the branch probability expressions
  symbolically.
\end{itemize}

\paragraph*{Contributions and outline.}
After illustrating our method on an example (\cref{sec:overview}), we present
our main contributions:
\begin{itemize}
\item A symbolic execution method for probabilistic programs with unknown input
  parameters and a formal proof of soundness for our method (\cref{sec:pse}).
\item A collection of case studies from the randomized algorithms literature
  (\cref{sec:case_studies}), including bounding the soundness probability of
  Freivalds' algorithm~\citep{freivalds1977}, the expected number of comparisons
  for randomized quicksort, and the correctness probability of a randomized
  algorithm for checking monotonicity~\citep{goldreich_2017}. All examples have
  unknown input parameters.
\item An optimized implementation of our approach called \SYSTEM (\cref{sec:implementation}).
  \SYSTEM is built on top of the \textsc{KLEE} symbolic execution engine~\citep{cadar2008}, 
  allowing \SYSTEM to perform symbolic execution on real implementations of randomized programs 
  written in any language that can generate LLVM bytecode, while also faithfully modeling program 
  states on real hardware (e.g., with overflowing arithmetic, arrays, etc.).
\item A thorough experimental evaluation of our tool (\cref{sec:evaluation}). We show
  that our method significantly outperforms two well-developed tools for
  analyzing probabilistic programs:
  \textsc{Psi}~\citep{psi_2016,lambda_psi_2020}, a leading tool for exact
  probabilistic symbolic inference, and \textsc{Storm}, a probabilistic model
  checker~\citep{hensel_2022}. We also evaluate the factors affecting
  the performance of \SYSTEM and the effectiveness of our optimizations.
\end{itemize}

\noindent{}We conclude by discussing other potential optimizations (\cref{sec:filtering}),
related work (\cref{sec:related}), and future directions (\cref{sec:conclusion}).
Source code for \SYSTEM and our case studies is available on Zenodo~\citep{plinkoArtifact2022}.
\nottoggle{SUPPLEMENTAL}{The full version of this paper is available on \texttt{arXiv}\citep{susag2022arXiv}, which contains full proofs and additional details.}{}

\section{Overview}
\label{sec:overview}

The Monty Hall problem~\citep{selvin1975} is a classic probability puzzle based on the American television show \textit{Let's Make a Deal}.
The problem itself is simple:
\begingroup
\addtolength\leftmargini{-0.21in}
\begin{quote}
	You are a contestant on a game show. 
  Behind one of three doors there is a car, and behind the others, goats.
	You pick a door, and the host, who knows what is behind each of the doors, opens a different door, behind which is a goat.
	The host then offers you the choice to switch to the remaining closed door.
	Should you?
\end{quote}
\endgroup
\noindent{}While it may seem counterintuitive, you should always switch doors: regardless of your original door choice, you will win the car two-thirds of the time if you do switch and only a third of the time if you do not.
We model the Monty Hall problem as the probabilistic program in \Cref{fig:montyhall}. The variable $\mathtt{choice} \in [1,3]$ represents the door that is originally chosen by the contestant, while the boolean variable $\mathtt{door\_switch}$ represents whether the contestant chooses to switch doors.
Regardless of the value of \texttt{choice}, if $\mathtt{door\_switch}$ is set to \texttt{true} then \texttt{monty\_hall} should return \texttt{true} (i.e., the contestant wins) two-thirds of the time.

Our goal is to automate this reasoning.
More generally, the problem we aim to solve is: given a probabilistic program with discrete sampling statements and a target probability bound, how do we verify that the program satisfies the bound?
While the property for the Monty Hall problem is possible to verify with existing methods due to its small input space (i.e., it is feasible to try all possible inputs and verify the probability bound on each output distribution), more realistic probabilistic programs often have an enormous space of inputs (e.g., randomized quicksort takes an array of integers as input).
We solve this problem by developing a symbolic execution method.
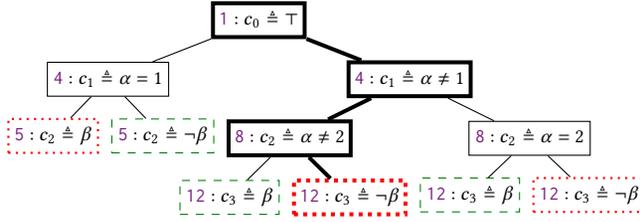
\begin{figure}[t]
	\centering
	{\scriptsize
		\begin{forest}baseline,for tree=draw,
			[{$\mathtt{\ref{line:monty_beg}}: c_0\triangleq\top$},align=center, base=bottom, line width=1.5pt
			[{$\mathtt{\ref{line:monty_choice_!=_car}}: c_1 \triangleq \alpha = 1$},align=center, base=bottom
			[{$\mathtt{\ref{line:monty_not_switch}}: c_2 \triangleq \beta$},align=center, base=bottom, node options={dotted,thick}, draw=red] %loss
			[{$\mathtt{\ref{line:monty_not_switch}}: c_2 \triangleq \neg\beta$},align=center, base=bottom, node options={dashed}, draw=darkgreen]] %win
			[{$\mathtt{\ref{line:monty_choice_!=_car}}: c_1 \triangleq \alpha \neq 1$},align=center, base=bottom, line width=1.5pt, edge={line width=1.5pt}
			[{$\mathtt{\ref{line:monty_choice_!=_2}}: c_2 \triangleq \alpha \neq 2$},align=center, base=bottom, line width=1.5pt, edge={line width=1.5pt}
			[{$\mathtt{\ref{line:monty_switch}}: c_3 \triangleq \beta$},align=center, base=bottom, node options={dashed}, draw=darkgreen ] % win
			[{$\mathtt{\ref{line:monty_switch}}: c_3 \triangleq \neg\beta$},align=center, base=bottom, edge={line width=1.5pt}, node options={dotted,thick}, draw=red, line width=1.5pt]] % loss
			[{$\mathtt{\ref{line:monty_choice_!=_2}}: c_2 \triangleq \alpha = 2$},align=center, base=bottom
			[{$\mathtt{\ref{line:monty_switch}}: c_3 \triangleq \beta$},align=center, base=bottom, node options={dashed}, draw=darkgreen ] % win
			[{$\mathtt{\ref{line:monty_switch}}: c_3 \triangleq \neg\beta$},align=center, base=bottom, node options={dotted,thick}, draw=red ]]]] % loss
		\end{forest}
	}
	\caption{Symbolic execution tree for the Monty Hall Problem if the car is behind door 1.}
	\label{fig:montyhall_tree}
  \vspace{-2em}
\end{figure}
In symbolic execution, program inputs are replaced by \textit{symbolic} variables that can take on any value.
The program is then ``run'' on these symbolic variables.
When a branch is encountered (e.g., an \texttt{if} condition), execution\begin{wrapfigure}[26]{r}{0.5\textwidth}
  \centering
  \inputminted[numbersep=3pt,xleftmargin=10pt,tabsize=2,fontsize=\footnotesize,linenos,escapeinside=||]{c}{montyhall.c}
  \vspace{-1em}
	\caption{C code for the Monty Hall problem.}
	\label{fig:montyhall}
\end{wrapfigure}proceeds along both branches, and the branch's guard is recorded in that path's \textit{path condition}, denoted by $\varphi$.
This yields a \textit{symbolic execution tree} that describes all possible paths through the program.

To illustrate this idea, suppose we were to use traditional symbolic execution to analyze the program in \Cref{fig:montyhall} but fixed the car behind door 1.
The two inputs, \texttt{choice} and \texttt{door\_switch}, would become the symbolic variables $\alpha$ and $\beta$ respectively.
The execution tree is shown in \Cref{fig:montyhall_tree}.
(Symmetric trees can be made for the cases when the car is behind doors 2 and 3.)
Each node contains the corresponding line number and branch guard.
Leaves surrounded by a green dashed line (\tikz{\draw[dashed, line width=1pt, color=darkgreen] (0,0) -- (0.35,0) }) denote paths where the program returns \texttt{true} (i.e., the contestant wins the car), and leaves surrounded by a red dotted line (\tikz{\draw[color=red, dotted, thick, line width=1pt] (0,0) -- (0.35,0) }) denote paths where the program returns \texttt{false} (i.e., the contestant wins a goat).
A path condition $\varphi$ is a conjunction of all the $c_i$ between the root of the tree and a leaf.

For example, consider the bolded path in~\Cref{fig:montyhall_tree}.
Execution begins on line~\ref{line:monty_beg} with an empty path condition: $\varphi = \top$.
When the first \texttt{if} condition is reached on line~\ref{line:monty_choice_!=_car}, \texttt{if (choice == car\_door)}, the path follows the false branch. 
We then conjoin the corresponding symbolic expression to $\varphi$, i.e., $\varphi = \alpha \neq 1$.
Execution along this path then jumps to line~\ref{line:monty_choice_!=_1}.
However, the guard  \texttt{choice != 1 \&\& car\_door != 1} is always false as \texttt{car\_door == 1}.
Since there is no ambiguity about which branch to follow, $\varphi$ remains unchanged.
The next branch is on line~\ref{line:monty_choice_!=_2}, which is equivalent to checking whether \texttt{car\_door != 2}. 
The path takes the true branch, so we update $\varphi$ with the new guard: $\varphi = (\alpha \neq 1) \wedge (\alpha \neq 2)$.
The last branch is on line~\ref{line:monty_switch}, where the program checks whether the contestant chose to switch doors.
This guard ultimately determines whether the contestant wins the car along this path as we previously determined that the contestant did not originally pick the winning door.
Here, switching wins the car while not switching only wins a goat.
The entire path condition $\varphi$ can then be written as $\varphi = (\alpha \neq 1) \wedge (\alpha \neq 2) \wedge \neg\beta$.

Given a setting of the input parameters, this analysis can tell us whether the contestant will win the car when the car is behind door 1; however, it cannot tell us how likely winning the car is.
In principle, if we combined the tree shown in~\Cref{fig:montyhall_tree} with the symmetric trees for the cases where the car is behind doors 2 and 3, and if we knew how likely it was to take one branch over another, we could extend this reasoning to determine how often we would reach a winning state.

To carry out this kind of quantitative reasoning within symbolic execution, \textbf{our core idea is to represent probabilistic samples as a new class of symbolic variables.}
We refer to these variables as \textit{probabilistic} symbolic variables in contrast to regular symbolic variables, which we instead call \emph{universal} symbolic variables.
Just as in standard symbolic execution, universal symbolic variables represent unknown program inputs.
Since these program inputs are universally quantified instead of being drawn from some assumed distribution, the name ``universal'' symbolic variable is appropriate.
In contrast, probabilistic symbolic variables represent the results of random draws.
During symbolic execution, instead of actually drawing a random sample, we represent the would-be sample as a probabilistic symbolic variable and record the source distribution.
Our method uses this information to compute the (symbolic) probabilities of taking branches.

To demonstrate, let \texttt{car\_door} be represented by the probabilistic symbolic variable, $\delta$.
Instead of branching solely on the universal symbolic variables $\alpha$ and $\beta$, we additionally branch on the probabilistic symbolic variable $\delta$.
This execution tree is presented in~\Cref{fig:montyhall_tree_prob}.
We elide branches where the current path condition shows that only one choice is feasible.

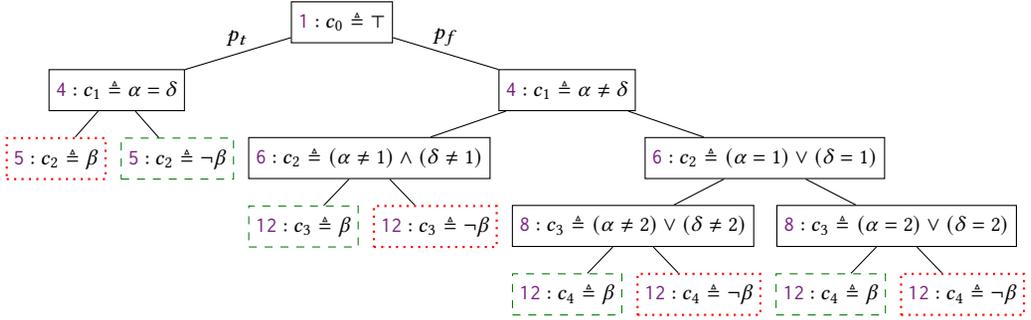
\begin{figure*}
	\centering
	{\footnotesize
		\begin{forest}baseline,for tree=draw,
			[{$\mathtt{\ref{line:monty_beg}}: c_0\triangleq\top$},align=center, base=bottom,
			[{$\mathtt{\ref{line:monty_choice_!=_car}}: c_1 \triangleq \alpha = \delta$},align=center, base=bottom, edge label={node [midway,above] {$p_t$} }
      [{$\mathtt{\ref{line:monty_not_switch}}: c_2 \triangleq \beta$},align=center, base=bottom, node options={dotted,thick}, draw=red ] %loss
      [{$\mathtt{\ref{line:monty_not_switch}}: c_2 \triangleq \neg\beta$},align=center, base=bottom, node options={dashed}, draw=darkgreen ]] %win
			[{$\mathtt{\ref{line:monty_choice_!=_car}}: c_1 \triangleq \alpha \neq \delta$},align=center, base=bottom, edge label={node [midway,above] {$p_f$} }
			[{$\mathtt{\ref{line:monty_choice_!=_1}}: c_2 \triangleq (\alpha \neq 1) \wedge (\delta \neq 1)$},align=center, base=bottom,
			[{$\mathtt{\ref{line:monty_switch}}: c_3 \triangleq \beta$},align=center, base=bottom, node options={dashed}, draw=darkgreen ] %win
      [{$\mathtt{\ref{line:monty_switch}}: c_3 \triangleq \neg\beta$},align=center, base=bottom, node options={dotted,thick}, draw=red ]] %loss
			[{$\mathtt{\ref{line:monty_choice_!=_1}}: c_2 \triangleq (\alpha = 1) \vee (\delta = 1)$},align=center, base=bottom
			[{$\mathtt{\ref{line:monty_choice_!=_2}}: c_3 \triangleq (\alpha \neq 2) \vee (\delta \neq 2)$},align=center, base=bottom
			[{$\mathtt{\ref{line:monty_switch}}: c_4 \triangleq \beta$},align=center, base=bottom, node options={dashed}, draw=darkgreen ] %win
			[{$\mathtt{\ref{line:monty_switch}}: c_4 \triangleq \neg\beta$},align=center, base=bottom, node options={dotted,thick}, draw=red ]] %loss
			[{$\mathtt{\ref{line:monty_choice_!=_2}}: c_3 \triangleq (\alpha = 2) \vee (\delta = 2)$},align=center, base=bottom
			[{$\mathtt{\ref{line:monty_switch}}: c_4 \triangleq \beta$},align=center, base=bottom, node options={dashed}, draw=darkgreen ]%win
			[{$\mathtt{\ref{line:monty_switch}}: c_4 \triangleq \neg\beta$},align=center, base=bottom, node options={dotted,thick}, draw=red ]]]]] %loss
		\end{forest}
	}
	\caption{Symbolic execution tree for the Monty Hall Problem with probabilistic symbolic variables. $p_t$ and $p_f$ denote the probability expressions for taking the true and false branches of the \texttt{if} condition found on line \ref{line:monty_choice_!=_car} of \Cref{fig:montyhall}. While every branch has an associated probability expression, we omit displaying them for presentation purposes.}
	\label{fig:montyhall_tree_prob}
\vspace*{-1em}
\end{figure*}

Since probabilistic symbolic variables are drawn from a distribution, we can compute the \textit{probability} of taking a certain branch by simply counting how many values from the domain of the distribution satisfy the guard condition.
For example, to compute the probability of taking the true branch in line~\ref{line:monty_choice_!=_car} (denoted by $p_t$ in~\Cref{fig:montyhall_tree_prob}), it suffices to count how many settings of $\delta$ satisfy $\alpha = \delta$ and divide by the total number of possible assignments to $\delta$ (i.e., 3) since \texttt{car\_door} is \textit{uniformly} randomly chosen to be 1, 2, or 3.
However, since $\alpha$ is a symbolic variable, we can only obtain a probability \textit{expression} over universal symbolic variables.
Let $[\cdot]$ denote Iverson brackets, where $[Q]=1$ if formula $Q$ is true and 0 otherwise.
Then, $p_t \triangleq ([\alpha = 1] + [\alpha = 2] + [\alpha = 3])/3 = 1/3$ as $\delta \in \{1,2,3\}$ and each setting is equally likely.
Similarly, $p_f \triangleq ([\alpha \neq 1] + [\alpha \neq 2] + [\alpha \neq 3])/3 = 2/3$.
We stress that in general these probabilities can be \emph{symbolic expressions} as they may depend on the universal symbolic variables.
This feature is one of the key advantages of our technique over existing methods for probabilistic symbolic execution: it makes reasoning about the probabilities much more complex, but it also allows us to prove properties for programs with truly unknown inputs.
In ~\Cref{sec:pse}, we explain how our technique can automatically generate these probability expressions and, by using standard probability rules, combine them to construct probability expressions for entire paths.

Returning now to our motivating question: if the contestant switches doors, do their chances of winning the car exceed $1/3$?
Using the tree in~\Cref{fig:montyhall_tree_prob}, we can limit our focus to just those paths that lead to a win.
Then, to find the probability of winning the car given the contestant switches doors, we filter out those paths where $\neg\beta$ is true and sum together the probabilities of the remaining paths.
Regardless of what $\alpha$ is, the resulting expression produces a probability of $2/3$.

\section{Probabilistic Symbolic Execution}
\label{sec:pse}

Now that we have described how our approach works at a high level, we can flesh out the details.
We begin with preliminary details~(\Cref{sec:pse_prelims}) and an overview of standard symbolic execution for non-probabilistic programs~(\Cref{sec:symex}).
Then, we present our symbolic execution approach for probabilistic programs.
Our approach has two steps: first, we augment standard symbolic execution to track probabilistic information as paths are explored~(\Cref{sec:pse_intro}).
Then, we convert the program paths into a logical query that encodes a probabilistic property~(\Cref{sec:query_gen}).
Finally, we formalize our approach mathematically and prove its soundness~(\Cref{sec:formalization}).

\subsection{Preliminaries}
\label{sec:pse_prelims}

Our algorithm works on programs written in \textbf{pWhile}, a core imperative probabilistic programming language.
The statements of \textbf{pWhile} are described by the following grammar:
\[
  S := \mathtt{x} \gets e \mid \mathtt{x} \sim d \mid S_1 ; S_2 \mid \mathbf{if}~c~\mathbf{then~goto}~T \mid \mathbf{halt} 
\]
Intuitively, the assignment statement $\mathtt{x} \leftarrow e$ assigns the
result of evaluating the expression $e$ to the program variable \texttt{x},
while the sampling statement $\mathtt{x} \sim d$ draws a random sample from a
distribution $d$ and assigns the result to the program variable \texttt{x}.
Here, $d$ is a \textit{discrete} distribution expression defining which distribution the sample should be drawn from; these distributions are primitives, corresponding to mathematically standard distributions (e.g., uniform or coin flip distributions).
We interpret distributions as functions from values to the range $[0,1]$ (i.e., the probability of the value occurring in the distribution).
For example, $\mathsf{UniformInt}(\mathtt{1},\mathtt{6})$ is a uniform distribution that selects a random integer between $1$ and $6$ (inclusive).
For presentation purposes, we limit our focus to uniform integer distributions, but our method can support other finite, discrete distributions, such as the Bernoulli or uniform joint distributions.

Control flow is implemented by $S_1 ; S_2$, which sequences two statements $S_1$ and $S_2$, and $\mathrel{\mathbf{if}} c \mathrel{\mathbf{then~goto}} T$, which jumps execution to the statement referenced by the address/label $T$ if the guard $c$ holds, and otherwise falls through to the next instruction. We assume that high-level constructs like loops and regular conditionals are compiled down to conditional branches.
Finally, $\mathbf{halt}$ marks the end of execution.

Our probabilistic symbolic execution algorithm will aim to answer the following question: What is the maximum (or minimum) probability that a program, \textsf{Prog}, terminates in a state where a predicate $\psi$ holds?
Formally, we are interested in computing
\[
  \displaystyle\max_{\vec{x}}~\{ \Pr_\mu[\psi] \mid \mu = \mathsf{Prog}(\vec{x})\}
\]
where $\mu$ is the distribution on outputs obtained by running $\mathsf{Prog}$ on inputs $\vec{x}$.
Since this quantity is difficult to compute in general, we will also be interested in proving bounds:
\[
  \displaystyle\max_{\vec{x}}~\{ \Pr_\mu[\psi] \mid \mu = \mathsf{Prog}(\vec{x})\}
  \bowtie f(\vec{x})
\]
where $\bowtie~\in\{ \leq, \geq, =, \ldots \}$ is a comparison operator and $f$ is a given function of the program inputs.

\subsection{Symbolic Execution}
\label{sec:symex}
\setlength{\textfloatsep}{4pt}
{\footnotesize
\begin{algorithm}[htp]
	\caption{Probabilistic Symbolic Execution Algorithm}
	\label{alg:symb_ex}
	\begin{algorithmic}[1]
		\Procedure{SymbEx} {\textsf{Prog} : \textit{Program}, $\vec{x}$ : \textit{Program Inputs}, \colorbox[RGB]{200, 200, 200}{$\psi$ : \textit{Predicate}}}
		\State{$E_s$} $\leftarrow$ {[]}, $\sigma_{\text{init}} \gets $ \textsc{BindToSymbolic}($\vec{x}$), \colorbox[RGB]{200, 200, 200}{$Enc_\psi \gets 0.0, Enc_\mathtt{v} \gets 0.0$} \label{line:initialization} %\algorithmiccomment{Initialization}
		\Let{$I_{0}$}{\textsc{getStartInstruction}(\textsf{Prog})} \label{line:first_inst}
		\Let{$S_{0}$}{($I_{0}$, $\top$, $\sigma_{\text{init}}$, \colorbox[RGB]{200, 200, 200}{$\emptyset$, 1.0})} \label{line:init_state} %\algorithmiccomment{Empty Initial State}
		\State {$E_s$.\textsc{Push}({$S_{0}$})} %\algorithmiccomment{Start with $S_{0}$ in Execution Stack}
		\While{$E_s \neq \emptyset$} \label{line:execution_loop}
		\Let{($I_{c}, \varphi, \sigma, $\colorbox[RGB]{200, 200, 200}{$P$, $p$})}{$E_s$.\textsc{Pop}( )} \label{line:select_state}
		\If{\textsc{UNSAT}($\varphi$)}
		\State \textbf{continue}
		\EndIf
		\Switch{\textsc{instType}($I_{c}$)} \label{line:inst_type}
		\Case{\fbox{$\mathtt{x} \gets e$}} \label{line:pse_assignment}\label{line:symbex_assign} %\algorithmiccomment{Assignment Instruction} 
		\Let{$I_1$}{\textsc{getNextInstruction}($I_{c}$)}
		\Let{$\sigma[\mathtt{x}]$}{$\sigma[e]$} \label{line:beg_raw_assign}
		\Let{$S_{1}$}{($I_{1}$, $\varphi$, $\sigma$, \colorbox[RGB]{200, 200, 200}{$P$, $p$})} \label{line:end_raw_assign}
		\State {$E_s$.\textsc{Push}({$S_{1}$})}
		\EndCase
		\cCase{\fbox{$\mathtt{x} \sim d$}}{gray}\label{line:pse_sampling}% \algorithmiccomment \colorbox[RGB]{200, 200, 200}{Sampling Instruction} 
		\cState{gray} $I_1$ $\gets$ \textsc{getNextInstruction}($I_{c}$)
		\cState{gray} $\delta \gets \text{Generate a fresh probabilistic symbolic variable}$ \label{line:fresh_psv}
		\cState{gray} {$\sigma_1,~P_1 \gets (\sigma[\mathtt{x} \mapsto \delta], P[\delta \mapsto d])$} \label{line:pse_sym_sample}
		\cState{gray} $S_1 \gets (I_{c}, \varphi, \sigma_1, P_1, p)$ \label{line:end_raw_sample}
		\cState{gray} $E_s$.\textsc{Push}({$S_1$})
    \cEndCase{gray} \label{line:end_pse_sampling}
		\Case{\fbox{\textbf{if} $c$ \textbf{then goto} $T$}} \label{line:symbex_branch} %\algorithmiccomment{Branch Instruction}
		\Let{$c_{sym}$}{$\sigma[c]$} \label{line:guard_convert}
		\State \colorbox[RGB]{200, 200, 200}{$p_t,~p_f \leftarrow~\mathbf{PBranch}~(c_{sym},~\varphi,~P)$} \label{line:pse_sym_branch}
		\Let{$I_{1}$}{\textsc{getInstruction}($T$)} 
		\Let{$I_{2}$} {\textsc{getNextInstruction}($I_c$)}	
		\Let{$S_{t}$}{($I_{1}$, $\varphi \wedge c_{sym}$, $\sigma$,\colorbox[RGB]{200, 200, 200}{$P$, $p \cdot p_t$})} \label{line:symbex_true_state} %\algorithmiccomment \colorbox[RGB]{200, 200, 200}{Multiply with current \textit{path} probability}
		\Let{$S_{f}$}{($I_{2}$, $\varphi \wedge \neg c_{sym}$, $\sigma$,\colorbox[RGB]{200, 200, 200}{$P$, $p \cdot p_f$})} \label{line:symbex_false_state}
		\State {$E_s$.\textsc{Push}({$S_{f}$})}
		\State {$E_s$.\textsc{Push}({$S_{t}$})}% \algorithmiccomment{Start with True State}
		\EndCase \label{line:end_symbex_branch}
		\Case{\fbox{\textbf{halt}}} %\algorithmiccomment{Terminate Instruction}
		\cState{gray} $\psi_{sym}\gets\sigma[\psi]$ \label{line:beg_enc}
		\cIf {\textsc{SAT}($\psi_{sym}$)}{gray}
    \cState{gray} $p_t,p_f \gets \textbf{PBranch}(\psi_{sym}, \varphi, P)$
		\cState{gray} $Enc_\psi \gets Enc_\psi + p$	\label{line:encoding_t}	 
		\cEndIf{gray}
		\cState{gray} $Enc_\mathtt{v} \gets Enc_\mathtt{v} + p \cdot \sigma[\mathtt{v}]$ \label{line:encoding_a}
		\EndCase
		\EndSwitch \label{line:end_execution_loop}
		\EndWhile 
		\cState{gray} $r \gets\textsc{Solve}~(Enc_\psi,~Enc_\mathtt{v})$ \label{line:solve_encoding}
		\cState{gray}\Return{$r$}
		\EndProcedure
	\end{algorithmic}
\end{algorithm}
}

\textit{Symbolic execution}~\citep{king1976} is a program analysis technique that iteratively explores the set of paths through a given program.
Since program paths may depend on potentially unknown input variables, symbolic execution treats each input as a \emph{symbolic} variable.
All program operations are redefined to manipulate these symbolic variables.
We present a high-level description of symbolic execution in \Cref{alg:symb_ex}.
The reader should ignore the portions in \colorbox{gray}{shaded boxes} for now; these are our extensions to handle probabilistic programs.

Symbolic execution tracks program state using the following data structures: the next instruction $I_c$; a path condition $\varphi$, which is a conjunctive formula consisting of the symbolic branch constraints that are true for a particular path; and an expression map $\sigma$, which maps program variables to symbolic expressions over constants and symbolic variables.
To begin, \Cref{alg:symb_ex} initializes the execution stack $E_s$ with an empty list and the expression map $\sigma$ with fresh symbolic variables for each \textit{program input} in the input vector $\vec{x}$ by calling \textsc{BindToSymbolic}, as shown on \cref{line:initialization}.
$I_0$ is then initialized with the first instruction of the program on \cref{line:first_inst}, and the initial state tuple $S_0$ is created on \cref{line:init_state} and appended to $E_s$.

The core of \Cref{alg:symb_ex} iterates through all reachable states until $E_s$ is exhausted (\crefrange{line:execution_loop}{line:end_execution_loop}).
For each iteration, we pop a state from $E_s$ (\cref{line:select_state}), and, if the selected path condition $\varphi$ is feasible, the target instruction is analyzed based on the instruction's form: \cref{line:pse_assignment} for assignment statements, and \cref{line:symbex_branch} for branch statements.

For an assignment statement $\mathtt{x} \leftarrow e$ where $e$ is a program expression, we first convert $e$ into a \textit{symbolic} expression using the expression map $\sigma$.
We use the notation $\sigma[e]$ to denote the symbolic expression $e_{sym}$ constructed by replacing each program variable in $e$ with its corresponding symbolic expression in $\sigma$.
For a branch statement \textbf{if} $c$ \textbf{then goto} $T$, \Cref{alg:symb_ex} \textit{forks} by constructing the corresponding path conditions for both branches, pushing the symbolic branch guard $c_{sym}$ in the positive and negative form (for the true and false branches, respectively) to $E_s$.
Finally, symbolic execution along a path terminates at a \textbf{halt} instruction, and information about the final symbolic state is collected for subsequent analysis.
\subsection{Probabilistic Symbolic Execution}
\label{sec:pse_intro}

To support random sampling instructions, our algorithm distinguishes between two categories of symbolic variables:
\begin{itemize}
\item \textit{Universal symbolic variables}: These are identical to those in traditional symbolic execution and are used to model unknown program inputs.
\item \textit{Probabilistic symbolic variables}: These model a random sample from a known distribution.
  Fresh probabilistic symbolic variables are created for each sampling statement to represent the result of a random draw.
\end{itemize}
Symbolic execution becomes more challenging when both universal and probabilistic variables are present as branches are taken with some quantitative probability dependent on the results of the random samples and the universal symbolic variables.
To track information about probabilistic states, we use the following data structures:
\begin{itemize}
\item \textit{Distribution map} ($P$): Analogous to $\sigma$, this map from probabilistic symbolic variables to distribution expressions tracks the distribution from which a probabilistic symbolic variable was originally sampled from.
\item \textit{Path probability} ($p$): For each path, we maintain a path probability expression $p$ over the universal symbolic variables.
\end{itemize}

We can now define how symbolic execution handles sampling statements and how to compute the probability of taking either side of a branch statement. 

\begin{paragraph}{Sampling (\Crefrange{line:pse_sampling}{line:end_pse_sampling})}
	On encountering a sampling statement $\mathtt{x} \sim d$ (\cref{line:pse_sampling} of \Cref{alg:symb_ex}), we generate a fresh probabilistic symbolic variable $\delta$ to model the random value sampled from the distribution described by the distribution expression $d$.
  We update the expression map $\sigma$ to map the program variable \texttt{x} to $\delta$ and record $\delta$'s distribution in the distribution map $P$.
\end{paragraph}

\begin{algorithm}[t]
	\caption{Probabilistic Branch Algorithm}
	\label{alg:branch}
	\begin{algorithmic}[1]
		\Function{PBranch}{$c_{sym}, \varphi, P$}
		\Let{$(\delta_1,\ldots,\delta_n)$}{$\dom{P}$}
		\Let{$(d_1,\ldots,d_n)$}{$(P[\delta_1],\ldots,P[\delta_n])$} %\Comment{Retrieve all distribution expressions}
    \Let{$\mathcal{D}$}{$\dom{d_1} \times \cdots \times \dom{d_n}$}\label{line:D}
		\Let{$p_c$}{$\frac{\displaystyle\smashoperator[r]{\sum_{(v_1,\ldots,v_n) \in \mathcal{D}}} [(c_{sym} \wedge \varphi)\{\delta_1 \mapsto v_1,\ldots,\delta_n \mapsto v_n\}]}{\displaystyle\smashoperator[r]{\sum_{(v_1,\ldots,v_n) \in \mathcal{D}}}[\varphi\{\delta_1 \mapsto v_1,\ldots,\delta_n \mapsto v_n\}]}$}\label{line:branch_prob_compute}
    \Let{$p_c'$}{$1-p_c$}
		\State\Return{$(p_c, p_c')$}
		\EndFunction
	\end{algorithmic}
\end{algorithm}
\begin{paragraph}{Branches (\Crefrange{line:symbex_branch}{line:end_symbex_branch})}
	Similar to traditional symbolic execution, upon reaching a branch statement our algorithm pushes the symbolic states corresponding to the true and false branches to $E_s$.
  Then, \textsc{PBranch} (\Cref{alg:branch}) computes probability \textit{expressions} for the true and false branches.
  \textsc{PBranch} takes the following as inputs: an expression $c_{sym}$, which is the symbolic equivalent of the program guard expression $c$; the current path condition $\varphi$; and the current distribution map $P$. 
  It returns two probability expressions $(p_c,p_c')$ representing the probabilities of taking the true and false branches respectively.
  For simplicity, we present this subroutine in the simplified setting where all sampling instructions are from discrete uniform distributions; handling general weighted distributions is not much more complicated and is supported in our implementation.
	
	To gain intuition for how we compute path probabilities, consider the case where the guard condition contains \textit{only} universal symbolic variables.
  Let $c_{sym} = \sigma[c]$ be the equivalent symbolic expression for the guard $c$ and assume that $c_{sym}$ does not mention any probabilistic symbolic variables.
  In this case, the branch that is taken is entirely determined by the universal symbolic variables. 
	For a fixed setting of the universal symbolic variables, one branch must have a probability of 1 and the other 0.
  Put another way, either $c_{sym}$ or $\neg c_{sym}$ is satisfiable but never both.
	We use Iverson brackets to represent the symbolic probability of taking the true and false branches, namely $[c_{sym}]$ and $[\neg c_{sym}]$ respectively.

	For probabilistic branches---guards which contain probabilistic symbolic variables---computing the branch probability is more involved.
  We give an intuitive justification here and defer the proof of correctness to \Cref{sec:formalization}.
  When~\Cref{alg:symb_ex} encounters a branch statement, recall that $\varphi$ contains the necessary constraints on the symbolic variables that must hold in order to reach the branch statement in question.
  Therefore, we can view the probability of taking either branch as a \textit{conditional probability}: the probability that $c_{sym}$ holds assuming that $\varphi$ is satisfied.
	Formally, we want to compute $\Pr[c_{sym} \mid \varphi] = {\Pr[c_{sym}\wedge \varphi]}/{\Pr[\varphi]}$.

  \Cref{alg:branch} builds $p_c$, a formula for $\Pr[c_{sym} \mid \varphi]$.
	Note that each probabilistic symbolic variable $\delta$ is mapped to precisely one distribution, $d$.
	Assuming there are $n$ probabilistic symbolic variables $\delta_1,\ldots,\delta_n$ and $n$ distributions $d_1,\ldots,d_n$, the set of all possible values for $\delta_1,\ldots,\delta_n$ is $\mathcal{D}$ (\cref{line:D}).
  Using our simplifying assumption that all probabilistic symbolic variables are drawn from uniform distributions, each assignment $(v_1, \dots, v_n)$ for $(\delta_1, \dots, \delta_n)$ has equal probability.
  Thus, instead of computing the conditional probability as a ratio of probabilities, we calculate the ratio of the number of assignments from $\mathcal{D}$ which satisfy $c_{sym} \wedge \varphi$ and $\varphi$ as shown on \cref{line:branch_prob_compute} of \Cref{alg:branch}.
	We stress that $p_c$ is a \textit{symbolic expression} possibly containing universal symbolic variables.
	Finally, we use the fact that the sum of the conditional probabilities of the branch outcomes is 1 to get the probability of the false branch.

  We can now compute the probability of an entire path $\varphi \wedge c_{sym}$ (or $\varphi \wedge \neg c_{sym}$). 
  Since $p_t$ is the \textit{conditional} probability of $c_{sym}$ being true (\cref{line:pse_sym_branch}) and $p$ is the probability of $\varphi$, we can use the standard rule $\Pr[A \wedge B] = \Pr[A \mid B] \Pr[B]$, where $A$ and $B$ are any two events, to conclude that $\Pr[\varphi \wedge c_{sym}] = p \cdot p_t$.
  The same can be done for $\varphi \wedge \neg c_{sym}$ as shown on \cref{line:symbex_true_state,line:symbex_false_state}.
\end{paragraph}
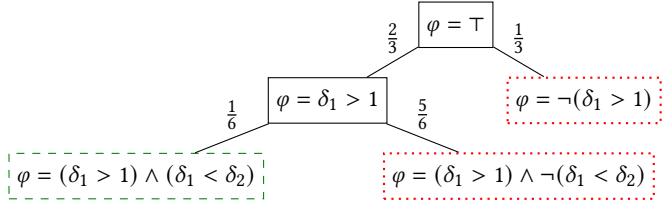
\begin{figure}[t]
	\centering
	\begin{subfigure}{.3\textwidth}
		\centering
		{\small
			\begin{algorithmic}[1]
				\State{$\mathtt{x}\sim\mathsf{UniformInt}(\mathtt{1},\mathtt{3})$}
				\State{$\mathtt{y}\sim\mathsf{UniformInt}(\mathtt{1},\mathtt{3})$}
				\If{$\mathtt{x} > \mathtt{1}$}\label{line:first_cond}
				\If{$\mathtt{x} < \mathtt{y}$}\label{line:second_cond}
				\State\Return{$\mathtt{True}$}
				\EndIf
				\Else
				\State\Return{$\mathtt{False}$}
				\EndIf
			\end{algorithmic}
		}
		\caption{Example program.}
		\label{fig:exam_program}
	\end{subfigure}
	\begin{subfigure}{.69\textwidth}
		\centering
		{\small
			\begin{forest}baseline,for tree=draw, for tree={s sep=5em}, grow={south},
				[{$\varphi=\top$},align=center, base=bottom
				[{$\varphi = \delta_1 > 1$}, align=center, base=bottom, edge label={node [midway,above] {$\frac{2}{3}$} }
				[{$\varphi = (\delta_1 > 1) \wedge (\delta_1 < \delta_2)$}, align=center, base=bottom, edge label={node [midway,above] {$\frac{1}{6}$} }, node options={dashed}, draw=darkgreen]
				[{$\varphi = (\delta_1 > 1) \wedge \neg(\delta_1 < \delta_2)$}, align=center, base=bottom, edge label={node [midway,above] {$\frac{5}{6}$}}, node options={dotted,thick}, draw=red]]
				[{$\varphi = \neg(\delta_1 > 1)$}, align=center, base=bottom, edge label={node [midway,above] {$\frac{1}{3}$} }, node options={dotted,thick}, draw=red]]
			\end{forest}
		}
		\caption{Symbolic execution tree.}
		\label{fig:exam_tree}
		
	\end{subfigure}
	\caption{A sample, randomized program and its associated symbolic execution tree annotated with probabilities.}
	\label{fig:example}
\end{figure}
\begin{paragraph}{Example}
	Consider the program in \Cref{fig:exam_program}, and suppose we wish to calculate the probability of the program returning \texttt{True}.
	Following~\Cref{alg:symb_ex}, we generate fresh probabilistic symbolic variables $\delta_1$ and $\delta_2$ for $\mathtt{x}$ and $\mathtt{y}$ respectively, and record their originating distribution $\mathcal{U}\{1,3\}$.
	Using \Cref{alg:branch} to process \cref{line:second_cond} of \Cref{fig:exam_program}, note that $\mathcal{D} = \{1,2,3\} \times \{1,2,3\}$ and 
	\begin{align*}
		p_c &= \Pr[\delta_1 < \delta_2 \mid \delta_1 > 1] = \frac{\Pr[(\delta_1 < \delta_2) \wedge (\delta_1 > 1)]}{\Pr[\delta_1 > 1]}\\
                                   &= \frac{\displaystyle\sum_{(v_1,v_2) \in \mathcal{D}} [(\delta_1 < \delta_2) \wedge (\delta_1 > 1)\{\delta_1 \mapsto v_1, \delta_2 \mapsto v_2\}]}{\displaystyle\sum_{(v_1,v_2) \in \mathcal{D}} [(\delta_1 > 1)\{\delta_1 \mapsto v_1\}]} = \frac{1}{6}.\\
	\end{align*}
  Therefore, the probability of taking the true branch of the inner \textbf{if} condition is only $1/6$, which makes sense as \texttt{x} is restricted to be either $2$ or $3$, but \texttt{y} can be either 1,2, or 3; however, only one setting of \texttt{x} and \texttt{y} will satisfy $(\mathtt{x} > 1) \wedge (\mathtt{x} < \mathtt{y})$, namely $\mathtt{x} = 2$ and $\mathtt{y} = 3$.
  We can use ~\Cref{alg:branch} on the remaining branches in \Cref{fig:exam_program} to yield the fully annotated symbolic execution tree in \Cref{fig:exam_tree}.
\end{paragraph}

\subsection{Query Generation}
\label{sec:query_gen}

Recall that our original goal was to prove bounds on the probability that a program \textsf{Prog} terminates in a state where an arbitrary predicate $\psi$ holds.
We frame this question as a logical query using the symbolic state found in the execution tree.
In general, our queries are of the following form:
\begin{equation}
  \label{eq:query}
	\forall \alpha_1,\ldots,\alpha_n~.~Enc_\psi \bowtie f(\alpha_1,\ldots,\alpha_n)
\end{equation}
where $\alpha_1,\ldots,\alpha_n$ are \textit{all} the universal symbolic variables found in \textsf{Prog}, $Enc_\psi$ is an expression representing the probability that $\psi$ holds in \textsf{Prog}, $\bowtie$ is a binary relation (e.g., $>,<,\geq,\leq$), and $f$ is some lower/upper bound function of the universal symbolic variables  that we want to prove \textsf{Prog} does not violate.
Another way of interpreting $f$ is a function which computes the \textit{desired} probability of $\psi$ occurring in \textsf{Prog} (e.g., a maximum acceptable error rate) given a setting of the universal symbolic variables.

We construct the expression $Enc_\psi$ on~\crefrange{line:beg_enc}{line:encoding_a} of~\Cref{alg:symb_ex}.
Once a path reaches a \textbf{halt} statement, $\psi$ is converted into its corresponding symbolic expression $\psi_{sym}$ using the expression map $\sigma$.
If $\psi_{sym}$ is satisfiable, we use~\Cref{alg:branch} to calculate the probability of $\psi_{sym}$ being true given $\varphi$, denoted by $p_\psi$.
We then add $p\cdot p_\psi$, the probability that we traversed $\varphi$ and $\psi_{sym}$ is true along path $\varphi$, to the current expression $Enc_\psi$. 
Note that some settings of the universal symbolic variables, $\alpha_1,\ldots,\alpha_n$, may not be reachable along path $\varphi$.
Since \Cref{alg:branch} embeds the path condition $\varphi$ into the probability expressions $(p_t,p_f)$, if $\varphi$ is falsified, then the probability of $\varphi$ is 0.
Once all reachable states in $E_s$ are exhausted, the total probability that $\psi$ holds for all reachable paths in \textsf{Prog} is $Enc_\psi$.
\textsc{Solve} is then called on \cref{line:solve_encoding}, which takes in $Enc_\psi$, constructs (\ref{eq:query}), and calls and automated decision procedure, such as an SMT solver, to answer the query.

The meaning of (\ref{eq:query}) depends on the termination condition of \Cref{alg:symb_ex}.
If all reachable paths are explored, then $Enc_\psi$ represents the exact probability that $\psi$ holds in the output distribution, parameterized by the unknown input variables.
However, we will not always be able to explore all paths.
In this case, $Enc_\psi$ may not be equal to the true probability of $\psi$, but it is always a sound \emph{lower bound} of the true probability.
While this may not be enough to verify typical correctness properties, it
can be useful for \emph{refuting} certain kinds of correctness properties. For
instance, if we want to check that the probability of $\psi$ is at most
$f(\alpha_1,\ldots,\alpha_n)$ for all settings of the input variables, and
\Cref{alg:symb_ex} produces a probability mass $Enc_\psi$ that exceeds
$f(\alpha_1,\ldots,\alpha_n)$ for some setting of the input variables, then the
original upper bound cannot hold.

\begin{paragraph}{Example}
  In~\Cref{sec:overview} we introduced the Monty Hall problem.
  We wanted to prove that if a contestant always switched doors the probability of winning the car is $2/3$.
  We can frame this property as a query of the form defined in (\ref{eq:query}).
  If $\alpha$ is the universal symbolic variable corresponding to \texttt{choice}, and $\beta$ is the universal symbolic variable corresponding to \texttt{door\_switch}, let $\psi \triangleq \beta \wedge \mathsf{monty\_hall}(\alpha,\beta)$, $\bowtie~\triangleq~=$, and $f \triangleq 2/3$.
  Then, our full query is $\forall \alpha,\beta~.~Enc_\psi = 2/3$.
  \Cref{alg:symb_ex} then constructs the tree found in~\Cref{fig:montyhall_tree_prob} and the probability expression for each path.
  If $\varphi_i$ is the path condition for the path ending in the $i$th-left-most leaf of~\Cref{fig:montyhall_tree_prob} and $p_i$ is the probability expression for $\varphi_i$ as created by~\Cref{alg:symb_ex}, then $Enc_\psi = p_3 + p_5 + p_7$.
  Lastly, \textsc{Solve} takes the query $\forall \alpha,\beta~.~p_3 + p_5 + p_7 = 2/3$ and uses an external automated decision procedure to prove for each possible setting of $\alpha$ and $\beta$ that $p_3 + p_5 + p_7 = 2/3$.
\end{paragraph}

\subsubsection{Expected Value Queries}
\label{sec:expected_value}

In addition to probability-bound queries, our technique can also prove bounds on the expected value of any variable with only minor modifications.
Suppose that we want to find the expected value of an arbitrary program variable named \texttt{v}, denoted $\E{\mathtt{v}}$.
Since each path has a probability $p_i$ and an expression map $\sigma_i$, we can construct a symbolic expression that computes $\E{\mathtt{v}}$, namely $\sum_{i=1}^n p_i \cdot \sigma_i[\mathtt{v}]$, as $\sigma_i[\mathtt{v}]$ is the symbolic expression which is the value of $\mathtt{v}$ along the $i^\text{th}$ path (\cref{line:encoding_a} of~\Cref{alg:symb_ex}).
Then, for \textsc{Solve} on~\cref{line:solve_encoding}, we use an automated tool to solve the query
\[
	\forall \alpha_1,\ldots,\alpha_n~.~Enc_\mathtt{v} \bowtie f(\alpha_1,\ldots,\alpha_n)
\]
where $\alpha_1,\ldots,\alpha_n$ are the universal symbolic variables representing the input variables of \textsf{Prog}.

\subsection{Formalization}
\label{sec:formalization}

In this section, we state a series of soundness theorems for~\Cref{alg:symb_ex}.
\iftoggle{SUPPLEMENTAL}{
    Program semantics are presented in~\Cref{sec:semantics} and complete proofs for each of the theorems are presented in~\Cref{sec:supplemental_formalism}.
}{}

\subsubsection{Soundness Theorems}
\label{sec:proofs}

Note that, in \Cref{alg:symb_ex}, we represent the symbolic program state as the five-tuple $S = (I, \varphi, \sigma, P, p)$; however, in our abstraction $R = (\varphi, \sigma,P)$, we omit $I$ and $p$.
We remove $I$ as we prove soundness locally for each type of statement in \textbf{pWhile} rather than prove soundness for the entirety of \Cref{alg:symb_ex}, so there is no need to track what the current instruction is.
We begin by proving that $p$ can be completely reconstructed using $\varphi$ and $P$, so there is no need to track it separately.
\begin{restatable}[Equivalency between $p$ and $(\varphi,P)$]{lemma}{equivpthm}
  For all program statements, $S$ in \textbf{pWhile}, \Cref{alg:symb_ex} maintains
  \[
    p = \frac{\displaystyle\sum_{(v_1,\ldots,v_n) \in \mathcal{D}} [ \varphi\{\delta_1 \mapsto v_1, \ldots, \delta_n \mapsto v_n\} ]}{\abs{\mathcal{D}}}
  \]
  where $\{\delta_1,\ldots,\delta_n\} = \dom{P}$ is the set of all the probabilistic symbolic variables in the program, and $\mathcal{D} = \dom{P[\delta_1]} \times \cdots \times \dom{P[\delta_n]}$ is the set of all assignments to the probabilistic symbolic variables $\delta_1,\ldots,\delta_n$.
\end{restatable}

We prove soundness for \Cref{alg:symb_ex} by showing that our abstraction $R$ respects the semantics for each of the three main types of statements in \textbf{pWhile}, namely \textit{assignment}, \textit{sampling}, and \textit{branching}.

\begin{restatable}[Correctness of Assignment (\Crefrange{line:beg_raw_assign}{line:end_raw_assign})]{theorem}{soundnessAssignment}
  If a distribution $\mu$ satisfies $R=(\varphi,\sigma,P)$ and~\Cref{alg:symb_ex} reaches an assignment statement of the form $\mathtt{x} \gets e$, then $\mu_{\mathtt{x}\gets e}$, the distribution of program memories after executing $\mathtt{x} \gets e$, satisfies $R'=(\varphi,\sigma',P)$.
\end{restatable}

\begin{restatable}[Correctness of Sampling (\Crefrange{line:pse_sym_sample}{line:end_raw_sample})]{theorem}{soundnessSampling}
  If a distribution $\mu$ satisfies $R = (\varphi,\sigma,P)$ and~\Cref{alg:symb_ex} reaches a sampling statement of the form $\mathtt{x} \sim d$, then $\mu_{\mathtt{x} \sim d}$, the distribution of program memories after executing $\mathtt{x} \sim d$, satisfies $R'=(\varphi,\sigma',P')$.
\end{restatable}

\begin{restatable}[Correctness of Branching (\Cref{line:pse_sym_branch,line:symbex_true_state,line:symbex_false_state})]{theorem}{soundnessBranching}\label{thm:branch}
  If a distribution $\mu$ satisfies $R=(\varphi,\sigma,P)$, and~\Cref{alg:symb_ex} encounters a branching statement of the form $\mathbf{if}~c~\mathbf{then~goto}~T$, then $\mu_c$ satisfies $R_{true} = (\varphi \wedge c_{sym}, \sigma, P)$ and $\mu_{\neg c}$ satisfies $R_{false} = (\varphi \wedge \neg c_{sym}, \sigma, P)$ if all the distributions in $P$ are uniform. Additionally, for all $a_{\forall} \in \mathit{A}_{\forall}$,
\[
    \sum_{m \in \{m \in \mathit{M} \mid \deno{c}m = 1\}}\mu(a_{\forall},m) = \deno{p_c} a_{\forall}\qquad\text{and}\quad\sum_{m \in \{m \in \mathit{M} \mid \deno{c}m = 0\}}\mu(a_{\forall},m) = \deno{p_c'} a_{\forall}
\]
  where $p_c' = 1 - p_c$.
\end{restatable}

\section{Case Studies}
\label{sec:case_studies}
In this section, we introduce the case studies that we will use in our evaluation.
For each case study, we give a brief description of the algorithm and the target property we wish to verify.
Additionally, we specify which variables are \textit{concretized} for our evaluation.
Like other symbolic execution methods, our tool requires fixing, or concretizing, some parameters, such as the lengths of input arrays and loop bounds.
We discuss in greater depth why concretization is necessary and discuss the implications it has on performance in \Cref{sec:evaluation}.
\iftoggle{SUPPLEMENTAL}{Pseudocode for each case study can be found in \Cref{sec:algorithms}.}{}

\paragraph*{Freivalds' Algorithm.}
Freivalds' algorithm~\citep{freivalds1977} is a randomized algorithm used to verify matrix multiplication.
Given three $n \times n$ matrices $A, B$, and $C$, if $A \times B = C$, then the algorithm always returns \texttt{true}.
However, if $A \times B \neq C$, the probability that Freivalds' algorithm returns \texttt{true} (i.e., a false positive) is at most $1/2$.

We want to verify that the false positive rate is at most $1/2$.
Letting
\[
  \psi \triangleq A \times B \neq C \wedge \mathsf{Freivalds}(A,B,C,n) = \texttt{true},
\]
our query is then:
\begin{equation*}
	\forall A,B,C~.~Enc_\psi \leq 1/2
\end{equation*}
where $n$ is concretized.
To reduce the probability of false positives, Freivalds' algorithm can be run $k$ times, returning \texttt{true} only if all calls return \texttt{true}, and \texttt{false} otherwise.
Since each trial is independent, the probability of a false positive given that $A \times B \neq C$ is at most $\left(1/2\right)^k$.
If we let $\psi \triangleq A \times B \neq C \wedge \mathsf{MultFreivalds}(A,B,C,n,k) = \texttt{true}$, our query then becomes: $\forall A,B,C~.~Enc_\psi \leq (1/2)^k$,
where $n$ and $k$ are concretized.

The assertion $A \times B \neq C$ can be encoded as the following logical formula, \textsc{SomeOff}:
\[
   \textsc{SomeOff} \triangleq \bigvee_{i=1}^n\bigvee_{j=1}^n (A \times B)_{i,j} \neq C_{i,j}
\]
which asserts that \textit{at least} one element differ between $A \times B$ and $C$.

We also consider two other encodings:
\[
   \textsc{AllOff} \triangleq \bigwedge_{i=1}^n\bigwedge_{j=1}^n (A \times B)_{i,j} \neq C_{i,j}\qquad\qquad
   \textsc{FirstOff} \triangleq (A \times B)_{0,0} \neq C_{0,0}
\]
\textsc{AllOff} states that \textit{every} element of $A \times B$ must be different from the corresponding element in $C$, while \textsc{FirstOff} states that the first element of each matrix is different.
These conditions imply \textsc{SomeOff}, so they give weaker guarantees: a program that satisfies the specification under \textsc{AllOff} or \textsc{FirstOff} might fail to satisfy the specification under \textsc{SomeOff}.
However, if a program \emph{fails} to satisfy the specification under \textsc{AllOff} or \textsc{FirstOff}, then they must also fail to satisfy the specification under \textsc{SomeOff}.
These different encodings impact the performance of our tool; we return to these details in~\Cref{sec:evaluation}.

\paragraph*{Reservoir Sampling.}
Reservoir Sampling \citep{vitter_1985} is a classical, online algorithm which produces a simple random sample $S$ of $k$ elements drawn from a population $A$ of $n$ elements.
We want to verify that the random samples returned by reservoir sampling are indeed \textit{simple}, or that each element of $A$ has equal probability of being in $S$.
To make this property easier to state, we consider an equivalent property: if we assume all elements of $A$ are distinct, then the probability that any element of $A$ is in $S$ is exactly $k/n$.
We assume that the elements of $A$ are distinct because if $A$ contained duplicate elements, then it is only true that the probability that each element of $A$ is in $S$ is \textit{at least} $k/n$, which does not necessarily imply that $S$ is a simple random sample of $A$.

We frame this property by asking what is the probability that the first element of $A$ appears in $S$.
Let $\psi \triangleq A[1] \neq \cdots \neq A[n] \wedge A[1] \in \mathsf{ReservoirSample}(A,k)$. We then take the query:
\begin{equation*}
  \forall A~.~Enc_\psi = k/n,
\end{equation*}
where $n$ and $k$ are concretized.

\paragraph*{Randomized Monotonicity Testing.}
Randomized monotonicity testing~\citep{goldreich_2017} is a binary-search-like algorithm to determine how far a function $f: [n] \rightarrow R$ is from being monotone increasing, where $[n] = \{1,\ldots,n\}$ and $R$ is some ordered set.
If $f$ is monotone, then the algorithm will return \texttt{true}; if $f$ is \textit{$\delta$-far} from monotone, it will reject with probability greater than $\delta$~\citep{goldreich_2017}, where $\delta$-far means that the value of $f$ needs to be changed on at most $\delta\cdot n$ points in order to arrive at a monotone function.
Let $\mathsf{DistToMono}(f)$ be the number of elements whose order needs to change to make $f$ monotone, $R = \mathbb{Z}$, and $\psi \triangleq \mathsf{MonotoneTest}(f,n) = \texttt{false}$.
Our query is:
\begin{equation*}
  \forall f: [n] \rightarrow \mathbb{Z}~.~Enc_\psi > \frac{\mathsf{DistToMono}(f)}{n}.
\end{equation*}
where $n$ is concretized, and we universally quantify over all functions with domain $[n]$ and codomain $\mathbb{Z}$.

\paragraph*{Randomized Quicksort.}
Quicksort~\citep{quicksort} is a popular sorting algorithm which uses partitioning in order to achieve efficient sorting.
While there are many ways of choosing a pivot element, one effective way is randomly.
Using this pivot method, the \textit{expected} number of comparisons required is approximately $1.386 n \log_2(n)$ where $n$ is the length of the array.
We use the method described in~\Cref{sec:expected_value} to compute $\E{\mathtt{num\_comps}}$, where \texttt{num\_comps} is a program variable which is initially set to 0 and is incremented whenever a comparison occurs for a concretized array length, $n$.

\paragraph*{Bloom Filter.}
A Bloom filter~\citep{bloom_1970} is a space-efficient, probabilistic data structure used to rapidly determine whether an element is in a set by allowing for false positives.
Bloom filters are generally parameterized by the expected number of \textit{unique} elements to be inserted, $m$, and a \textit{target} maximum false positive error rate after $m$ insertions, $\varepsilon$.
We want to verify that after inserting $m$ elements, the actual probability of a false positive is at most $\varepsilon$.

To capture this property, let $x_1,\ldots, x_m$ be $m$ unique elements to be inserted into a Bloom filter $B$.
Note that a false positive occurs when $\mathsf{BloomCheck}(B,y)$ returns \texttt{true} when $y \neq x_1 \wedge \cdots \wedge y \neq x_m$ given that and $x_1, \ldots, x_m$ were all previously inserted into $B$.
Therefore, let $y$ be such that $y \neq x_1 \wedge \cdots \wedge y \neq x_m$, and
\[
  \psi \triangleq B = \mathsf{BloomCreate}(m,\varepsilon) \wedge \left(\bigwedge_{i=1}^m\mathsf{BloomInsert}(B,x_i)\right) \wedge \mathsf{BloomCheck}(B,y) = \texttt{true}
  \]
and take the query to be $\forall x_1,\ldots,x_m,y~.~Enc_\psi \leq \varepsilon$, where $m$ and $\varepsilon$ are concretized.

\paragraph*{Count-min Sketch.}
A count-min sketch~\citep{graham_2005} is another space-efficient, probabilistic data structure which encodes a frequency table for a stream of data.
To reduce the space usage, the counts can be approximate: the \textit{reported} estimate for the frequency of an element, $x$, namely $\hat{a}_x$, has the property that $\hat{a}_x \leq \varepsilon \cdot N$ with probability $1-\gamma$, where $N = \sum_x a_x$, or the total number of elements seen in the sketch, and $a_x$ is the actual count for $x$.
To appropriately initialize the count-min sketch, most implementations are parameterized by the additive error factor $\varepsilon$ and the error probability $\gamma$.

We want to prove that after inserting $n$ distinct elements, the probability that the error is greater than the additive bound $\varepsilon$ is at most $1-\gamma$.
Assuming distinct elements simplifies the query, since the actual counts for each element in the count-min sketch is exactly $1$.
It is possible to remove this assumption, adjusting the following query to refer to the actual counts for each inserted element.

Similar to how we framed the Bloom filter query, let $x_1,\ldots,x_n$ be $n$ unique elements to be inserted into a count-min sketch $C$, parameterized by $\varepsilon$ and $\gamma$.
If we let
\[
  \psi \triangleq C = \mathsf{SketchCreate}(\varepsilon, \gamma) \wedge \left(\bigwedge_{i=1}^n\mathsf{SketchUpdate}(C,x_i)\right) \wedge \mathsf{SketchEstimate}(C,x_1) > 1 + \varepsilon \cdot n,
\]
then the query becomes
\begin{equation*}
  \forall x_1,\ldots,x_n~.~Enc_\psi \leq \gamma
\end{equation*}
where $n, \varepsilon$, and $\gamma$ are concretized parameters.

\section{Implementation}
\label{sec:implementation}

We implemented our method in a prototype tool called \SYSTEM, building on top of the \textsc{KLEE} symbolic execution engine~\citep{cadar2008}.\footnote{%
  Source code for \textsc{Plinko} and all case studies are freely available on Zenodo~\citep{plinkoArtifact2022}}
\SYSTEM constructs the queries described in \Cref{sec:query_gen}, and then
dispatches them to the Z3 SMT solver~\citep{demoura2008}; we use the array, bit-vector, and real number theories from SMT-LIB~\citep{smtlib}.
If \SYSTEM is unable to prove the query, it returns a counterexample consisting
of an input along with the computed probability of the target event.
By using \textsc{KLEE}, \SYSTEM is able to directly analyze LLVM bytecode with primitive calls for random sampling.
Thus, \SYSTEM can analyze any program written in any programming language that compiles down to LLVM intermediate representation (IR).
We require that the end-user annotate which variables should be universally quantified and which are probabilistic, along with the associated distribution, using \textsc{KLEE} intrinsic functions.

We have also implemented several optimizations for \SYSTEM, which aim to decrease the size and complexity of $Enc_\psi$, or the expression denoting the probability that the predicate $\psi$ is true.
Specifically, we have designed three optimizations: precomputing the probability for simple guards, applying algebraic simplifications, and compressing formulas by factoring out shared terms.
We discuss these optimizations here and in the next section~(\Cref{sec:evaluation}) evaluate their effectiveness on our case studies.

\paragraph*{Simple Guards}
We define a \textit{simple guard} as a predicate of the form $\delta \bowtie k$, where $\delta$ is a probabilistic symbolic variable, $P[\delta]$ is a uniform integer distribution, $\bowtie~\in \{<, >, \leq, \geq, =, \neq\}$, and $k$ is a constant.
For such guards, we can compute the probability of $c = (\delta \bowtie k)$ being true in constant time as there exist closed-form solutions for the number of satisfying assignments to $\delta$ which make $c$ true.
This effectively allows us to avoid using \Cref{alg:branch} to compute $p_c$ when $c_{sym}$ is a simple guard.

For example, reconsider the outermost conditional statement (\cref{line:first_cond}) presented in~\Cref{fig:exam_program}: $\delta_1 > 1$ where $\delta_1 \sim \mathcal{U}\{1,3\}$.
By default, \SYSTEM would use \Cref{alg:branch} to construct $p_{\delta_1 > 1} = [1 > 1] + [2 > 1] + [3 > 1]/3$. 
This symbolic expression then gets simplified by Z3 to the constant $2/3$.
However, for $\delta_1 > 1$, the numerator is just $\max\{\dom{\mathcal{U}\{1,3\}}\} - 1 = 2$, and so the probability of $\delta_1 > 1$ being true is simply $2/3$.
We have encoded these rewriting rules into~\SYSTEM for each primitive binary relation.

\paragraph*{Algebraic Simplifications}
In order to reduce the work that Z3 must do, we can apply elementary algebraic simplifications on $Enc_\psi$.
For example, we leverage the associative, commutative, distributive properties of addition and multiplication, along with basic boolean algebra laws such as the associativity and commutativity of conjunction and disjunction and De Morgan's laws.
We apply a simple fixed-point algorithm to apply each of these laws on $Enc_\psi$ as a preprocessing step before solving the query.

\paragraph*{Formula Sharing}
The encoded probability $Enc_\psi$ generated by our method is a sum of
expressions, one for each explored path. Since paths come from the same
symbolic execution tree, terms representing branch probabilities can be repeated
multiple times in $Enc_\psi$ across different paths with a shared prefix. These
repeated terms can be quite large and complex. Our
optimization factors out common terms in $Enc_\psi$ by creating additional
logical variables.

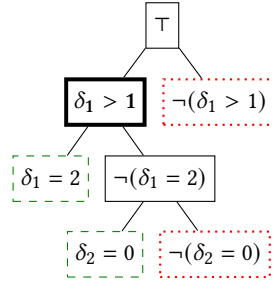
\begin{figure}[t]
  \centering
  \begin{subfigure}{.49\textwidth}
    \centering
    {\small
      \begin{algorithmic}[1]
        \State{$\mathtt{x}\sim\mathsf{UniformInt}(\mathtt{1},\mathtt{3})$}
        \If{$\mathtt{x} > \mathtt{1}$}
        \If{$\mathtt{x} = 2$}
        \State\Return{$\mathtt{True}$}
        \Else 
        \State{$\mathtt{y}\sim\mathsf{UniformInt}(\mathtt{0},\mathtt{1})$}
        \State\Return{$\mathtt{y} = 0$}
        \EndIf
        \Else
        \State\Return{$\mathtt{False}$}        
        \EndIf
      \end{algorithmic}
    }
    \caption{Program.}
    \label{fig:dyn_prog_example_code}
  \end{subfigure}
  \begin{subfigure}{.3\textwidth}
    \centering
		{\small
			\begin{forest}baseline,for tree=draw,
				[{$\top$},align=center, base=bottom
				[{$\mathbf{\delta_1 > 1}$}, align=center, base=bottom,line width=1.5pt
				[{$\delta_1 = 2$}, align=center, base=bottom, node options={dashed}, draw=darkgreen]
				[{$\neg(\delta_1 = 2)$}, align=center, base=bottom
				[{$\delta_2 = 0$}, align=center, base=bottom, node options={dashed}, draw=darkgreen]
				[{$\neg(\delta_2 = 0)$}, align=center, base=bottom, node options={dotted,thick}, draw=red]]]
				[{$\neg(\delta_1 > 1)$}, align=center, base=bottom, node options={dotted,thick}, draw=red]]
			\end{forest}
		}
    \caption{Symbolic execution tree.}
		\label{fig:dyn_prog_example_tree}
  \end{subfigure}
  \caption{An example program and symbolic execution tree which exhibits path overlap.}
  \label{fig:dyn_prog_example}
\end{figure}

To gain intuition, consider the example program in \Cref{fig:dyn_prog_example}.
There are two paths where the program in \Cref{fig:dyn_prog_example_code} returns \texttt{true}, namely $\varphi_1 = (\delta_1 > 1) \wedge (\delta_1 = 2)$ and $\varphi_2 = (\delta_1 > 1) \wedge \neg(\delta_1 = 2) \wedge (\delta_2 = 0)$.
Note that $\varphi_1$ and $\varphi_2$ share a common conjunct, namely $\delta_1 > 1$ which is the bolded node in \Cref{fig:dyn_prog_example_tree}.
By default, \Cref{alg:symb_ex} would construct
\[
  Enc_\psi = \Pr[\delta_1 > 1] \cdot \Pr[\delta_1 = 2 \mid \delta_1 > 1] + \Pr[\delta_1 > 1] \cdot \Pr[\delta_1 \neq 2 \mid \delta_1 > 1] \cdot \Pr[\delta_2 = 0 \mid \delta_1 \neq 2 \wedge \delta_1 > 1].
\]
However, since $\delta_1 > 1$ is a shared conjunct, we are able to factor out $\Pr[\delta_1 > 1]$ from both probability expressions by creating a fresh variable, $p_1 = \Pr[\delta_1 > 1]$ and replacing every occurrence of $\Pr[\delta_1 > 1]$ in $Enc_\psi$ with $p_1$.
We take this idea one step further by observing that $\delta_1 = 2$ appears in $\varphi_1$, while its negation appears in $\varphi_2$.
Since for any event $A$, $\Pr[A] = 1 - \Pr[\neg A]$, we can set $p_2 = \Pr[\delta_1 = 2 \mid \delta_1 > 1]$ and again apply a substitution on $Enc_{\psi}$, resulting in
\[
  Enc_{\psi} = p_1 \cdot p_2 + p_1 \cdot (1 - p_2) \cdot \Pr[\delta_2 = 0 \mid \delta_1 \neq 2 \wedge \delta_1 > 1].
\]
We then create assertions which assert the equality of $p_1$ and $p_2$ with their corresponding expressions, and pass these along with $Enc_\psi$ to the \textsc{Solve} procedure.

\section{Evaluation}
\label{sec:evaluation}
To evaluate \SYSTEM, we consider the following questions:
\begin{itemize}
  \item \Q{1} How efficiently can \SYSTEM prove probabilistic properties on complex programs?
  \item \Q{2} How does \SYSTEM compare with other general purpose approaches?
  \item \Q{3} What factors affect the performance of \SYSTEM?
  \item \Q{4} How do our optimizations impact the performance of \SYSTEM?
\end{itemize}
We conducted our experiments on a machine with a 3.3GHz Intel Core i7-5820K and 32 GB of RAM, running Arch Linux with kernel 5.17.11.

\paragraph*{\textbf{\Q{1} Discussion}}
We implemented each of our case studies presented in~\Cref{sec:case_studies} in C++, and verified them using \SYSTEM.
We used publicly available implementations for the Bloom filter\footnote{\url{https://github.com/jvirkki/libbloom}} and count-min sketch\footnote{\url{https://github.com/alabid/countminsketch}} case studies.

Like other symbolic execution methods, our tool requires concretizing some parameters, such as the sizes of input arrays and loop bounds.
For instance, \textsc{KLEE} has limited support for unknown (symbolic) array lengths and so it instead considers concrete array sizes when exploring paths.
Note that when an array length is concretized, the array contents remain symbolic.
For example, if the input array length is concretized to 5 in quicksort, {\SYSTEM} searches over all possible arrays of length 5.
The choice of concretization also has a large effect on the performance of \SYSTEM as larger concrete parameters generally produce more program paths, thus requiring more time to verify correctness.
In a realistic verification setting where concrete parameters are not known ahead of time, \SYSTEM could be used to consider longer and longer array inputs until it times out.

\Cref{tab:q1} summarizes our experimental results after verifying the target properties for each of the case studies presented in~\Cref{sec:case_studies}.
For each experiment, we report the following metrics: the amount of time taken by \SYSTEM to explore the paths of the program and to generate the query, the amount of time Z3 took to solve the query, the total amount of time elapsed, the number of lines in the C++ code, the number of paths \SYSTEM  explored, the number of random samples, and the maximum values of the concretized variables such that the query can be proven or disproven in the time allotted.
(Recall that \Cref{sec:case_studies} describes the concretized parameters for each of our case studies.)
All properties were verified within 10 minutes, except for Bloom filter due to a bug in the reference implementation.
All results in \Cref{tab:q1} use the default version of \SYSTEM (i.e., with only the simple guards optimization turned on).
\begin{table*}
  \centering
  \caption{Performance metrics for each of the case studies presented in~\Cref{sec:case_studies}.  ``{Freivalds'}'' uses the \textsc{SomeOff} specification over \texttt{int}s and ``Freivalds' (Multiple)'' uses the \textsc{FirstOff} specification over \texttt{char}s.}
  \label{tab:q1}
  \begin{tabular}{@{}lrrrrrrl@{}}
    \toprule
    & \multicolumn{3}{c}{Timing (sec.)} &&&\\ \cmidrule{2-4}
    \textbf{Case Study} & \textbf{KLEE} & \textbf{Z3} & \textbf{Total} & \textbf{Lines} & \textbf{Paths} & \textbf{Samples} & \textbf{Concretizations}\\ \midrule
    Freivalds' & 3 & 26 & \textbf{29} & 97 & 2 & 2 & $n=2$ \\
    Freivalds' (Multiple) & 6 & 259 & \textbf{265} & 96 & 8 & 21 & $(n, k) = (3, 7)$ \\
    Reservoir Sampling & 14 & 98 & \textbf{112} & 52 & 127 & 6 & $(n, k) = (13, 7)$ \\
    Reservoir Sampling & 460 & 1 & \textbf{461} & 52 & 4096 & 12 & $(n, k) = (13, 1)$ \\
    Monotone Testing & 6 & 384 & \textbf{390} & 69 & 36 & 1 & $n=27$ \\
    Quicksort & 14 & 114 & \textbf{128} & 65 & 120 & 10 & $n=5$ \\
    Bloom Filter & 18 & 395 & \textbf{413} & 386 & 83 & 8 & $(m, \varepsilon) = (3, 0.39)$ \\
    Count-min Sketch & 4 & 145 & \textbf{149} & 245 & 3 & 8 & $(n, \varepsilon, \gamma) = (4, 0.5, 0.25)$ \\
    \bottomrule
  \end{tabular}
\end{table*}

We make a few general observations about the results.
First, the case studies showcase a range of path counts and number of random samples.
For example, even when we restricted the input array to be of length 5, \SYSTEM still explored 120 paths containing 10 random samples when analyzing quicksort.
Second, the time~\SYSTEM takes to execute~\Cref{alg:symb_ex} is usually much less than the time Z3 takes to solve the query.
This suggests that the main bottleneck at the current scale of our case studies is constraint solving, not path exploration.

\paragraph*{Extended Length Experiments.}

To better understand the scaling of \SYSTEM, we ran additional experiments with a timeout of 2 hours.
In~\Cref{tab:q1-long} we report the performance results of these experiments.
The concretization settings reported are the highest such that the target property can be successfully verified within 2 hours.

We first note that the concretization settings of Freivalds' algorithm, quicksort, Bloom filter, and count-min sketch could not be increased beyond those shown in~\Cref{tab:q1}.
For Bloom filter we only tried to increase $m$ and for count-min sketch we only tried to increase $n$.
Additionally, only Freivalds' algorithm with $n = 3$ reached the 2-hour timeout, whereas quicksort, Bloom filter, and count-min sketch all ran out of memory before reaching the 2-hour timeout, at which point they were terminated by the operating system.
For reservoir sampling, we show two concretized parameter settings with the highest resource usage (both time and memory).
We note that the difficulty of the problem increases as $n$ increases and as $k$ decreases.
When $(n,k) = (41,38)$, we observed \SYSTEM using a maximum of $17.8$~GB of RAM, and when $(n,k)=(100,98)$, \SYSTEM peaked at $8.1$~GB of RAM.
On the other hand, Freivalds' (Multiple) and monotone testing were both more time constrained rather than memory constrained.

\begin{table*}
  \centering
  \caption{Performance metrics for each of the case studies presented in~\Cref{sec:case_studies} under a 2-hour timeout. ``Freivalds' (Multiple)'' uses the \textsc{FirstOff} specification over \texttt{char}s.}
  \label{tab:q1-long}
  \begin{tabular}{@{}lrrrrrrl@{}}
    \toprule
    & \multicolumn{3}{c}{Timing (sec.)} &&&\\ \cmidrule{2-4}
    \textbf{Case Study} & \textbf{KLEE} & \textbf{Z3} & \textbf{Total} & \textbf{Lines} & \textbf{Paths} & \textbf{Samples} & \textbf{Concretizations}\\ \midrule
    Freivalds' (Multiple) & 7 & 2027 & \textbf{2034} & 96 & 6 & 21 & $(n, k) = (4, 5)$ \\
    Reservoir Sampling & 5 & 85 & \textbf{90} & 52 & 15 & 3 & $(n, k) = (41, 38)$ \\
    Reservoir Sampling & 5 & 36 & \textbf{41} & 52 & 7 & 2 & $(n, k) = (100, 98)$ \\
    Monotone Testing & 6 & 5883 & \textbf{5889} & 69 & 42 & 1 & $n=32$ \\
    \bottomrule
  \end{tabular}
\end{table*}

These extended experiments seem to suggest that memory usage is a large performance bottleneck for \SYSTEM.
Generally, \SYSTEM consumes more memory as the number of paths and random samples increase.
Memory consumption also increases as the cardinality of the originating distribution's domain increases since we compute path probabilities by summing over every possible assignment to a random sample.

\paragraph*{\textbf{\Q{2} Discussion}}
We compare~\SYSTEM against two state-of-the-art systems: \textsc{Psi}~\citep{psi_2016,lambda_psi_2020}, an exact symbolic inference engine for probabilistic programs, and \textsc{Storm}~\citep{hensel_2022}, a leading probabilistic model checker used to check properties on finite-state probabilistic transition systems.
We also discuss the differences between the input language of \SYSTEM and those of \textsc{Psi} and \textsc{Storm}.

\paragraph*{\textsc{Psi} Comparison.}
To compare against \textsc{Psi} we first implemented each of our case studies as probabilistic programs in the \textsc{Psi} language.
Then, for each case study, we used \textsc{Psi} to produce a symbolic expression $e_\textsc{Psi}$ which encodes the joint posterior distribution of the program parameterized by the input variables.
We use $e_\textsc{Psi}$ in place of $Enc_\psi$ in our general formula for probabilistic queries: $\forall \alpha_1,\ldots,\alpha_n. e_\textsc{Psi} \bowtie f(\alpha_1,\ldots,\alpha_n)$,
where $\alpha_1,\ldots,\alpha_n$ are universally quantified program inputs, and then use Z3 to solve the query.

However, we were only able to evaluate \textsc{Psi} on four of our case studies: Freivalds' (single and multiple), reservoir sampling, and monotonicity testing.
For the Bloom filter and count-min sketch case studies, \textsc{Psi} returned an incorrect symbolic expression.
Additionally, since \textsc{Psi} does not support recursive functions, we were unable to encode a comparable version of quicksort into \textsc{Psi}.

In \Cref{fig:psi_comparison}, we present performance results directly comparing \textsc{Psi} and~\SYSTEM on our first four case studies.
We ran \textsc{Psi} and \SYSTEM on our case studies with a range of concretized program parameters---like our system, \textsc{Psi} requires that some input parameters (e.g., sizes of input arrays) be concretized---and we report the end-to-end time for verification.
For both tools, we used a 10-minute timeout denoted by a horizontal dashed line.
An ``X'' means that for that experiment either a timeout was reached or the process ran out of memory.
Recall that the default version of \SYSTEM only employs the simple guards optimization.

\begin{figure}
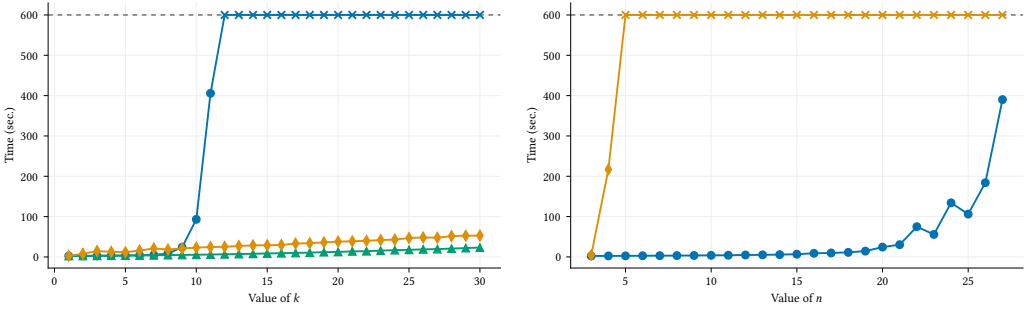
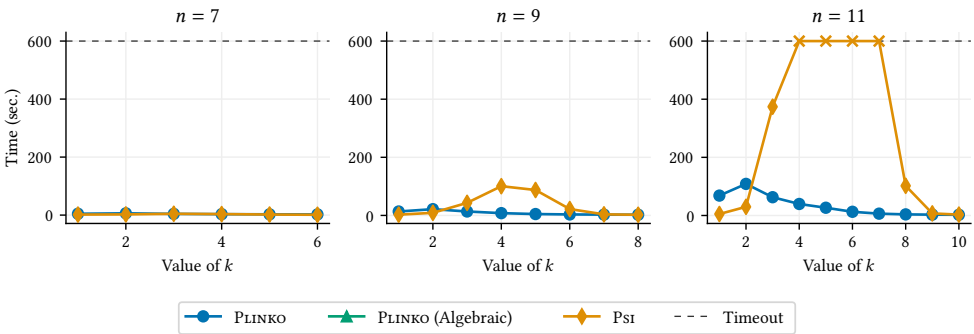

\centering
\begin{subfigure}{.5\textwidth}
    \centering
    \resizebox{\textwidth}{!}{\input{plots/psi_freivalds_mult_plot.pgf}}
    \caption{Freivalds' (\textsc{FirstOff}) over \texttt{char}s with $n=2$.}
    \label{fig:psi_freivalds}
\end{subfigure}%
\begin{subfigure}{.5\textwidth}
    \centering
    \resizebox{\textwidth}{!}{\input{plots/psi_monotone_plot.pgf}}
    \caption{Monotone Testing over \texttt{int}s}
    \label{fig:psi_monotone}
\end{subfigure}
\begin{subfigure}{\textwidth}
    \centering
    \resizebox{0.95\textwidth}{!}{\input{plots/psi_reservoir_sample_plot.pgf}}
    \caption{Reservoir Sampling where $n = 7,9$, and $11$.}
    \label{fig:psi_reservoir}
  \end{subfigure}%
\caption{Performance results comparing \SYSTEM against the \textsc{Psi}
inference engine on the Freivalds' algorithm, Monotone Testing, and Reservoir
Sample case studies. All experiments were run under a 10-minute timeout. An
``X'' marker means that an experiment either exceeded the timeout or ran out of memory.}
\label{fig:psi_comparison}
\end{figure}
Overall, we found that the performance of the default version of \SYSTEM scales significantly better than \textsc{Psi} as the values of the concretized parameters increase.
For monotone testing (\Cref{fig:psi_monotone}) in particular, we found that \textsc{Psi} could only solve the query when $n = 3$ or $4$, and timed out for all larger concretizations.
For reservoir sampling, there are two parameters, $n$ and $k$.  \Cref{fig:psi_reservoir} shows how the tools perform as $k$ varies, fixing $n=7,9,$ and $11$; we found similar results for other concrete values of $n$.
For smaller values of $n$ (e.g., $n=7$), \textsc{Psi} and \SYSTEM have about equal performance; however, as $n$ increases we found that \SYSTEM outperforms \textsc{Psi}.
For instance, \SYSTEM can handle all settings of $k$ when $n = 11$ while \textsc{Psi} timed out for all $4 \leq k \leq 7$.

The only case study where the stock version of \SYSTEM performs worse than
\textsc{Psi} is Freivalds' algorithm (\Cref{fig:psi_freivalds}).
While \textsc{Psi} significantly outperforms the default version of \SYSTEM on this benchmark, if we enable algebraic simplifications optimizations, denoted by the green line in \Cref{fig:psi_freivalds}, \SYSTEM was approximately $55\%$ faster than \textsc{Psi}.
We note that \textsc{Psi} performs similar algebraic simplifications.

Additionally, the vast majority of the time \SYSTEM (Algebraic) spent was on executing \textsc{KLEE} instead of Z3.
For this reason we hypothesize that the performance discrepancy between \textsc{Psi} and \SYSTEM on Freivalds' algorithm is largely due to the relative ease \textsc{KLEE} has in exploring all paths, whereas \textsc{Psi} appears to be constrained by the path explosion caused by higher settings of $n$.

In summary, we found that \SYSTEM is generally able to verify the evaluated benchmarks faster and for larger input concretizations than \textsc{Psi}.
At the same time, we stress that \textsc{Psi} and \SYSTEM are designed to solve different, yet complementary problems; the former performing exact probabilistic inference on probabilistic programs, and the latter verifying probabilistic properties over universally quantified input variables.
One avenue for future work would be to explore how these two systems could be merged such that \SYSTEM leverages \textsc{Psi}'s simplification engine to reduce the complexity and size of \SYSTEM's symbolic expressions before calling Z3.

\paragraph*{\textsc{Storm} Comparison.}
We compare the performance of \SYSTEM and \textsc{Storm} on two of our case studies: Freivalds' algorithm and reservoir sampling.
These two case studies involve reasoning about two different types of computations.
Freivalds' algorithm mostly involves arithmetic operations (i.e., matrix-vector products), whereas reservoir sampling mainly copies and moves data.

While \SYSTEM can operate on programs written in mainstream languages (e.g., C++), \textsc{Storm} operates on finite-state transition systems.
Encoding our programs as transition systems is non-trivial.
While our examples do have a finite state space, the state space is extremely large.
Furthermore, a core aim of our work is to prove probabilistic properties that universally quantify over all possible program inputs, but \textsc{Storm} does not directly support programs with unknown inputs.

To work around these difficulties, we encoded our examples as Markov decision processes (MDPs) written in the PRISM format~\citep{kwiatkowska2011prism}.
To model a universal quantification over the program inputs, our encoded MDPs use a non-deterministic choice over all possible inputs.
While the PRISM format made it easier to encode our imperative programs, the translation is not exact.
For instance, the generated transition system does not accurately model finite precision arithmetic since the PRISM language does not support finite-width datatypes.
Nevertheless, we believe that our encoding is a fair transition-system proxy for our target examples.

\begin{figure}
	\centering
	\resizebox{0.68\textwidth}{!}{\input{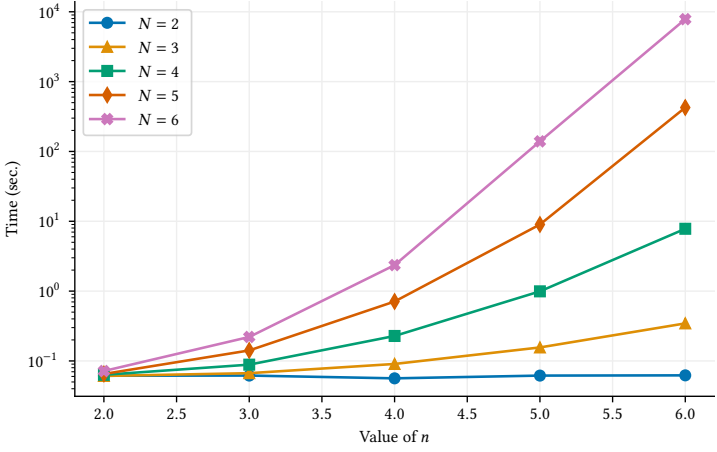}}
  \caption{\textsc{Storm} performance on Reservoir Sampling where $k=1$,
  $n=2..6$, and the array elements range from 1 to $N$. Note that the $y$-axis
is on a \textit{logarithmic} scale.}
	\label{fig:reservoir_storm_plot}
\end{figure}

We discuss performance results, as well as the difference in guarantees that \textsc{Storm} and \SYSTEM provide for both case studies below.
For Freivalds' algorithm we solely consider $2\times 2$ matrices.
In our original C++ program, these matrices contain 32-bit integers.
Due to the state space explosion, we needed to drastically reduce the domain size in order for \textsc{Storm} to successfully terminate.
For example, when we allow the three matrices to only contain elements from the set $\{1,2\}$, \textsc{Storm} can verify the upper bound for false positives in approximately 3 seconds where the constructed model contains roughly 42{,}000 states and 75{,}000 transitions.
When we increased the domain to be $\{1,2,3\}$, however, we had to manually terminate the experiment after 14 hours had passed with no result.
In contrast, \SYSTEM is able to verify the property in 75 seconds where the input domain consists of all 32-bit integers.

To rule out the chance that our encoding of Freivalds' algorithm as a non-deterministic choice over all possible inputs was not the bottleneck, we also ran \textsc{Storm} on each possible concrete setting of the three $2\times 2$ input matrices over the domains $\{1,2\}$ and $\{1,2,3\}$. 
For the domain $\{1,2\}$, this method took substantially more time, finishing in just over 20 seconds, while the domain $\{1,2,3\}$ also took more than 12 hours before we manually terminated the experiment.

We saw similar behavior with reservoir sampling.
To understand the dependence on the data size, we considered five input domains $\{1,\ldots,N\}$, where $2 \leq N \leq 6$.
For each input domain, we ran \textsc{Storm} on five different settings of $n$ and $k$, where $k=1$ and $2 \leq n \leq 6$.
\Cref{fig:reservoir_storm_plot} presents the performance results from these experiments.
Note that as the domain size increases linearly, the time required to check the property increases \textit{exponentially}. 
As we will show in \Q{3}, \SYSTEM also sees decreased performance as we increase the width of program variables, and we expect to see similar exponential scaling if we continued increasing the size of variables.
However, \SYSTEM is able to check the property for all settings of $n$ in approximately $3$ seconds when the input domain consists of all 32-bit integers.

\paragraph*{Input Language Differences.}

We conclude our evaluation against \textsc{Psi} and \textsc{Storm} by discussing the source language differences between these three tools.
As stated in~\Cref{sec:implementation}, \textsc{Plinko} consumes LLVM bytecode which allows us to verify programs written in any language supported by the LLVM front-end (e.g., C, C++, Objective C, Fortran, etc.).
This broadens the applicability of our analysis as most mainstream programming languages are able to compile down to LLVM; however, this also significantly increases the verification complexity as LLVM IR uses more complex datatypes, such as machine integers, arrays, etc.
In particular, \SYSTEM natively reasons over fixed-width bitvectors, as opposed to infinite precision, mathematical integers, allowing \SYSTEM to detect bugs related to integer overflow/underflow.

\textsc{Psi} and \textsc{Storm}, on the other hand, operate on more idealized core languages.
\textsc{Psi} provides a custom domain-specific language which does not support as many programming features as mainstream languages, such as recursion and fixed-width integers, and can only be used in conjunction with \textsc{Psi} (the programs themselves cannot be run, only analyzed by \textsc{Psi}).
\textsc{Storm}, as a probabilistic model checker, requires the programmer to encode their programs as either a discrete- or continuous-time Markov model instead of executable LLVM code.
For our evaluation, we use Markov decision processes (MDPs), which exactly encodes a probabilistic program with (a) a bounded, finite space of inputs and (b) exact arithmetic with \textit{infinite} precision.
While this is not an exact match for the semantics of our programming language, it is unclear how to generate MDPs that match the semantics of LLVM bytecode.
In any case, we do not believe that increasing the realism of the MDP encoding would significantly improve the performance of \textsc{Storm}.

\paragraph*{\textbf{\Q{3} Discussion}}
\begin{table}
  \centering
  \caption{Performance metrics of three different ways of specifying $A \times B \neq C$ in the query for Freivalds' algorithm. Here, $n=2$ and all elements of the matrices are C++ \texttt{int}s.}
  \label{tab:q2}
  \begin{tabular}{@{}lrrrr@{}}
    \toprule
    & \multicolumn{3}{c}{Timing (sec.)} &\\ \cmidrule{2-4}
    \textbf{Spec. for $A \times B \neq C$} & \textbf{KLEE} & \textbf{Z3} & \textbf{Total} & \textbf{Paths}\\ \midrule
    \textsc{AllOff} & 3 & 1 & \textbf{4} & 2\\
    \textsc{SomeOff}& 3 & 26 & \textbf{29} & 2\\
    \textsc{FirstOff}& 2 & 5 & \textbf{7} & 2\\
    \bottomrule
  \end{tabular}
\end{table}
\begin{table}
  \centering
  \caption{Performance metrics for four different domains from which the elements of $A,B$, and $C$ are drawn from in the implementation for Freivalds' Algorithm. We only consider $2 \times 2$ matrices for each data type and use \textsc{SomeOff} to specify that $A \times B \neq C$. In each variant,~\SYSTEM explored 2 paths.}
  \label{tab:q3}
  \begin{tabular}{@{}lrrr@{}}
    \toprule
    & \multicolumn{3}{c}{Timing (sec.)}\\ \cmidrule{2-4}
    \textbf{Data Type} & \textbf{KLEE} & \textbf{Z3} & \textbf{Total}\\ \midrule
    \texttt{long int} & 4 & 684 & \textbf{688}\\
    \texttt{int} & 3 & 26 & \textbf{29}\\
    \texttt{short int} & 2 & 4 & \textbf{6}\\
    \texttt{char} & 2 & 1 & \textbf{3}\\
    \bottomrule
  \end{tabular}
\end{table}
Recall that Freivalds' algorithm efficiently determines whether $A \times B = C$, where $A,B$, and $C$ are $n \times n$ matrices; however, the algorithm has a false positive error rate of at most $1/2$ if $A \times B \neq C$.
In order to verify this error rate, we must first assume that $A \times B \neq C$, and we propose three different encodings: \textsc{AllOff}, \textsc{SomeOff}, and \textsc{FirstOff}.
To determine the performance impact for each specification we ran Freivalds' algorithm where $k=1$ with $2 \times 2$ matrices which each contain 32-bit integers.
We present these performance results in~\Cref{tab:q2}.

In general, the results in~\Cref{tab:q2} suggest that simpler and more specific encodings of $A \times B \neq C$ increase performance at the cost of missing potential bugs.
The strongest and best performing specification was \textsc{AllOff}, whereas the most general and worst performing specification was \textsc{SomeOff}.
Intuitively, this makes sense as \textsc{AllOff} restricts the search space considerably more than the other specifications, whereas \textsc{SomeOff} requires Z3 to reason about \textit{all} matrices $A,B$, and $C$ such that $A \times B \neq C$.
Therefore, while performing the worst, \textsc{SomeOff} provides the strongest guarantee out of all the other variants, followed by \textsc{FirstOff}, and then finally, \textsc{AllOff}.

We also consider how the size of the domain (i.e., the C++ datatype) of the matrix elements impacts performance.
If each element of the matrix is only a single byte, there are only $2^{32}$ possible $2\times 2$ matrices, as opposed to elements of eight bytes, of which there are $2^{256}$ possible matrices.
For each data type, we again restrict ourselves to $2 \times 2$ matrices and specify that $A \times B \neq C$ using the \textsc{SomeOff} encoding.
The performance results are presented in~\Cref{tab:q3}.
On the evaluating machine, \texttt{long int}s are eight bytes, \texttt{int}s are four bytes, \texttt{short int}s are two bytes, and \texttt{char}s are a single byte.

Unsurprisingly, we found a direct correspondence between the size of the integer and the time it took to verify the property.
The time it took to verify Freivalds' algorithm with each data type seems to increase exponentially with the corresponding increase in integer size.
We must note, however, that using \texttt{char}s provides the weakest guarantee whereas \texttt{long int}s provides the strongest.
Although, for many applications, analyzing variants of the program with smaller data types might already provide sufficient confidence of correctness, or may surface bugs.

\begin{figure}[t]
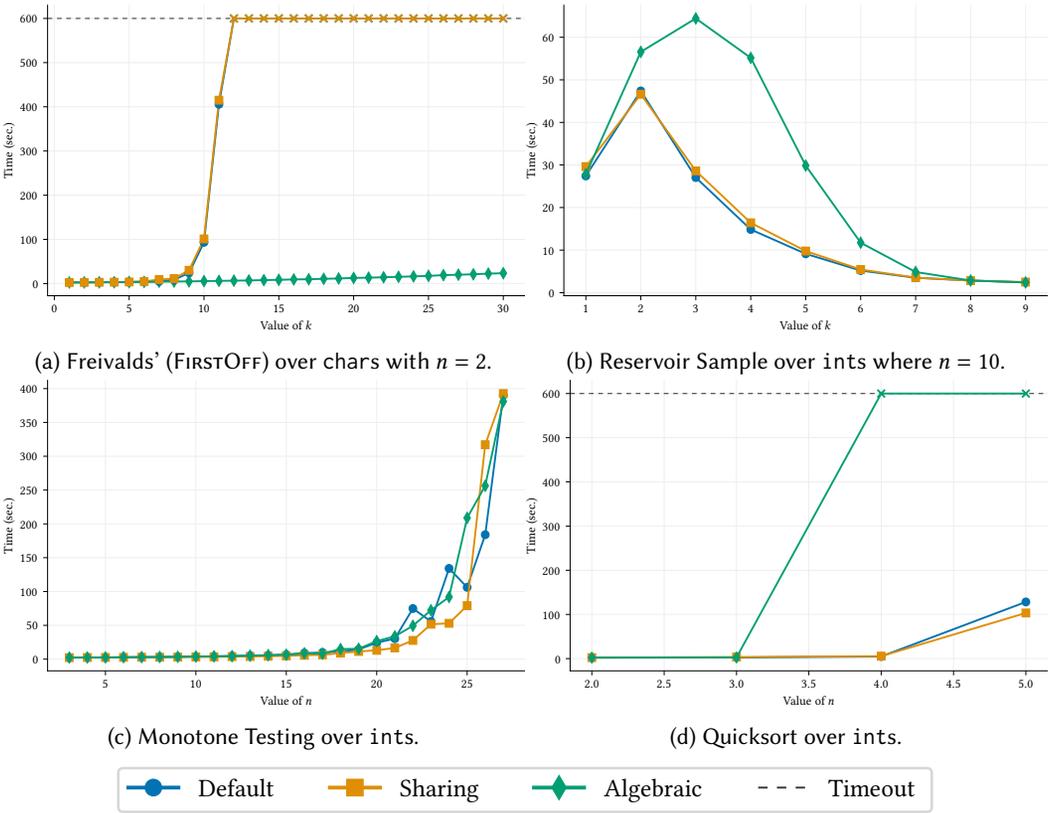

\centering
\begin{subfigure}{.5\textwidth}
    \centering
    \resizebox{\textwidth}{!}{\input{plots/freivalds_mult_plot.pgf}}
    \caption{Freivalds' (\textsc{FirstOff}) over \texttt{char}s with $n=2$.}
    \label{fig:freivalds_mult_opt}
\end{subfigure}%
\begin{subfigure}{.5\textwidth}
    \centering
    \resizebox{\textwidth}{!}{\input{plots/reservoir_sample_plot.pgf}}
    \caption{Reservoir Sample over \texttt{int}s where $n=10$.}
    \label{fig:reservoir_opt}
\end{subfigure}
\begin{subfigure}{.5\textwidth}
    \centering
    \resizebox{\textwidth}{!}{\input{plots/monotone_plot.pgf}}
    \caption{Monotone Testing over \texttt{int}s.}
    \label{fig:monotone_opt}
\end{subfigure}%
\begin{subfigure}{.5\textwidth}
    \centering
    \resizebox{\textwidth}{!}{\input{plots/quicksort_plot.pgf}}
    \caption{Quicksort over \texttt{int}s.}
    \label{fig:quicksort_opt}
\end{subfigure}
\begin{subfigure}{\textwidth}
    \centering
    \resizebox{\textwidth}{!}{%% Creator: Matplotlib, PGF backend
%%
%% To include the figure in your LaTeX document, write
%%   \input{<filename>.pgf}
%%
%% Make sure the required packages are loaded in your preamble
%%   \usepackage{pgf}
%%
%% Also ensure that all the required font packages are loaded; for instance,
%% the lmodern package is sometimes necessary when using math font.
%%   \usepackage{lmodern}
%%
%% Figures using additional raster images can only be included by \input if
%% they are in the same directory as the main LaTeX file. For loading figures
%% from other directories you can use the `import` package
%%   \usepackage{import}
%%
%% and then include the figures with
%%   \import{<path to file>}{<filename>.pgf}
%%
%% Matplotlib used the following preamble
%%   \usepackage[utf8x]{inputenc}\usepackage[T1]{fontenc}\newcommand{\vect}[1]{#1}
%%   \usepackage{fontspec}
%%
\begingroup%
\makeatletter%
\begin{pgfpicture}%
\pgfpathrectangle{\pgfpointorigin}{\pgfqpoint{5.788530in}{0.393620in}}%
\pgfusepath{use as bounding box, clip}%
\begin{pgfscope}%
\pgfsetbuttcap%
\pgfsetmiterjoin%
\definecolor{currentfill}{rgb}{1.000000,1.000000,1.000000}%
\pgfsetfillcolor{currentfill}%
\pgfsetlinewidth{0.000000pt}%
\definecolor{currentstroke}{rgb}{1.000000,1.000000,1.000000}%
\pgfsetstrokecolor{currentstroke}%
\pgfsetdash{}{0pt}%
\pgfpathmoveto{\pgfqpoint{0.000000in}{0.000000in}}%
\pgfpathlineto{\pgfqpoint{5.788530in}{0.000000in}}%
\pgfpathlineto{\pgfqpoint{5.788530in}{0.393620in}}%
\pgfpathlineto{\pgfqpoint{0.000000in}{0.393620in}}%
\pgfpathlineto{\pgfqpoint{0.000000in}{0.000000in}}%
\pgfpathclose%
\pgfusepath{fill}%
\end{pgfscope}%
\begin{pgfscope}%
\pgfsetbuttcap%
\pgfsetmiterjoin%
\definecolor{currentfill}{rgb}{1.000000,1.000000,1.000000}%
\pgfsetfillcolor{currentfill}%
\pgfsetfillopacity{0.800000}%
\pgfsetlinewidth{1.003750pt}%
\definecolor{currentstroke}{rgb}{0.800000,0.800000,0.800000}%
\pgfsetstrokecolor{currentstroke}%
\pgfsetstrokeopacity{0.800000}%
\pgfsetdash{}{0pt}%
\pgfpathmoveto{\pgfqpoint{0.669334in}{0.085009in}}%
\pgfpathlineto{\pgfqpoint{5.119196in}{0.085009in}}%
\pgfpathquadraticcurveto{\pgfqpoint{5.146973in}{0.085009in}}{\pgfqpoint{5.146973in}{0.112787in}}%
\pgfpathlineto{\pgfqpoint{5.146973in}{0.296398in}}%
\pgfpathquadraticcurveto{\pgfqpoint{5.146973in}{0.324176in}}{\pgfqpoint{5.119196in}{0.324176in}}%
\pgfpathlineto{\pgfqpoint{0.669334in}{0.324176in}}%
\pgfpathquadraticcurveto{\pgfqpoint{0.641557in}{0.324176in}}{\pgfqpoint{0.641557in}{0.296398in}}%
\pgfpathlineto{\pgfqpoint{0.641557in}{0.112787in}}%
\pgfpathquadraticcurveto{\pgfqpoint{0.641557in}{0.085009in}}{\pgfqpoint{0.669334in}{0.085009in}}%
\pgfpathlineto{\pgfqpoint{0.669334in}{0.085009in}}%
\pgfpathclose%
\pgfusepath{stroke,fill}%
\end{pgfscope}%
\begin{pgfscope}%
\pgfsetrectcap%
\pgfsetroundjoin%
\pgfsetlinewidth{1.505625pt}%
\definecolor{currentstroke}{rgb}{0.003922,0.450980,0.698039}%
\pgfsetstrokecolor{currentstroke}%
\pgfsetdash{}{0pt}%
\pgfpathmoveto{\pgfqpoint{0.697112in}{0.217787in}}%
\pgfpathlineto{\pgfqpoint{0.836001in}{0.217787in}}%
\pgfpathlineto{\pgfqpoint{0.974890in}{0.217787in}}%
\pgfusepath{stroke}%
\end{pgfscope}%
\begin{pgfscope}%
\pgfsetbuttcap%
\pgfsetroundjoin%
\definecolor{currentfill}{rgb}{0.003922,0.450980,0.698039}%
\pgfsetfillcolor{currentfill}%
\pgfsetlinewidth{1.003750pt}%
\definecolor{currentstroke}{rgb}{0.003922,0.450980,0.698039}%
\pgfsetstrokecolor{currentstroke}%
\pgfsetdash{}{0pt}%
\pgfsys@defobject{currentmarker}{\pgfqpoint{-0.041667in}{-0.041667in}}{\pgfqpoint{0.041667in}{0.041667in}}{%
\pgfpathmoveto{\pgfqpoint{0.000000in}{-0.041667in}}%
\pgfpathcurveto{\pgfqpoint{0.011050in}{-0.041667in}}{\pgfqpoint{0.021649in}{-0.037276in}}{\pgfqpoint{0.029463in}{-0.029463in}}%
\pgfpathcurveto{\pgfqpoint{0.037276in}{-0.021649in}}{\pgfqpoint{0.041667in}{-0.011050in}}{\pgfqpoint{0.041667in}{0.000000in}}%
\pgfpathcurveto{\pgfqpoint{0.041667in}{0.011050in}}{\pgfqpoint{0.037276in}{0.021649in}}{\pgfqpoint{0.029463in}{0.029463in}}%
\pgfpathcurveto{\pgfqpoint{0.021649in}{0.037276in}}{\pgfqpoint{0.011050in}{0.041667in}}{\pgfqpoint{0.000000in}{0.041667in}}%
\pgfpathcurveto{\pgfqpoint{-0.011050in}{0.041667in}}{\pgfqpoint{-0.021649in}{0.037276in}}{\pgfqpoint{-0.029463in}{0.029463in}}%
\pgfpathcurveto{\pgfqpoint{-0.037276in}{0.021649in}}{\pgfqpoint{-0.041667in}{0.011050in}}{\pgfqpoint{-0.041667in}{0.000000in}}%
\pgfpathcurveto{\pgfqpoint{-0.041667in}{-0.011050in}}{\pgfqpoint{-0.037276in}{-0.021649in}}{\pgfqpoint{-0.029463in}{-0.029463in}}%
\pgfpathcurveto{\pgfqpoint{-0.021649in}{-0.037276in}}{\pgfqpoint{-0.011050in}{-0.041667in}}{\pgfqpoint{0.000000in}{-0.041667in}}%
\pgfpathlineto{\pgfqpoint{0.000000in}{-0.041667in}}%
\pgfpathclose%
\pgfusepath{stroke,fill}%
}%
\begin{pgfscope}%
\pgfsys@transformshift{0.836001in}{0.217787in}%
\pgfsys@useobject{currentmarker}{}%
\end{pgfscope}%
\end{pgfscope}%
\begin{pgfscope}%
\definecolor{textcolor}{rgb}{0.000000,0.000000,0.000000}%
\pgfsetstrokecolor{textcolor}%
\pgfsetfillcolor{textcolor}%
\pgftext[x=1.086001in,y=0.169176in,left,base]{\color{textcolor}\rmfamily\fontsize{10.000000}{12.000000}\selectfont Default}%
\end{pgfscope}%
\begin{pgfscope}%
\pgfsetrectcap%
\pgfsetroundjoin%
\pgfsetlinewidth{1.505625pt}%
\definecolor{currentstroke}{rgb}{0.870588,0.560784,0.019608}%
\pgfsetstrokecolor{currentstroke}%
\pgfsetdash{}{0pt}%
\pgfpathmoveto{\pgfqpoint{1.813362in}{0.217787in}}%
\pgfpathlineto{\pgfqpoint{1.952251in}{0.217787in}}%
\pgfpathlineto{\pgfqpoint{2.091140in}{0.217787in}}%
\pgfusepath{stroke}%
\end{pgfscope}%
\begin{pgfscope}%
\pgfsetbuttcap%
\pgfsetmiterjoin%
\definecolor{currentfill}{rgb}{0.870588,0.560784,0.019608}%
\pgfsetfillcolor{currentfill}%
\pgfsetlinewidth{1.003750pt}%
\definecolor{currentstroke}{rgb}{0.870588,0.560784,0.019608}%
\pgfsetstrokecolor{currentstroke}%
\pgfsetdash{}{0pt}%
\pgfsys@defobject{currentmarker}{\pgfqpoint{-0.041667in}{-0.041667in}}{\pgfqpoint{0.041667in}{0.041667in}}{%
\pgfpathmoveto{\pgfqpoint{-0.041667in}{-0.041667in}}%
\pgfpathlineto{\pgfqpoint{0.041667in}{-0.041667in}}%
\pgfpathlineto{\pgfqpoint{0.041667in}{0.041667in}}%
\pgfpathlineto{\pgfqpoint{-0.041667in}{0.041667in}}%
\pgfpathlineto{\pgfqpoint{-0.041667in}{-0.041667in}}%
\pgfpathclose%
\pgfusepath{stroke,fill}%
}%
\begin{pgfscope}%
\pgfsys@transformshift{1.952251in}{0.217787in}%
\pgfsys@useobject{currentmarker}{}%
\end{pgfscope}%
\end{pgfscope}%
\begin{pgfscope}%
\definecolor{textcolor}{rgb}{0.000000,0.000000,0.000000}%
\pgfsetstrokecolor{textcolor}%
\pgfsetfillcolor{textcolor}%
\pgftext[x=2.202251in,y=0.169176in,left,base]{\color{textcolor}\rmfamily\fontsize{10.000000}{12.000000}\selectfont Sharing}%
\end{pgfscope}%
\begin{pgfscope}%
\pgfsetrectcap%
\pgfsetroundjoin%
\pgfsetlinewidth{1.505625pt}%
\definecolor{currentstroke}{rgb}{0.007843,0.619608,0.450980}%
\pgfsetstrokecolor{currentstroke}%
\pgfsetdash{}{0pt}%
\pgfpathmoveto{\pgfqpoint{2.943640in}{0.217787in}}%
\pgfpathlineto{\pgfqpoint{3.082529in}{0.217787in}}%
\pgfpathlineto{\pgfqpoint{3.221418in}{0.217787in}}%
\pgfusepath{stroke}%
\end{pgfscope}%
\begin{pgfscope}%
\pgfsetbuttcap%
\pgfsetmiterjoin%
\definecolor{currentfill}{rgb}{0.007843,0.619608,0.450980}%
\pgfsetfillcolor{currentfill}%
\pgfsetlinewidth{1.003750pt}%
\definecolor{currentstroke}{rgb}{0.007843,0.619608,0.450980}%
\pgfsetstrokecolor{currentstroke}%
\pgfsetdash{}{0pt}%
\pgfsys@defobject{currentmarker}{\pgfqpoint{-0.035355in}{-0.058926in}}{\pgfqpoint{0.035355in}{0.058926in}}{%
\pgfpathmoveto{\pgfqpoint{-0.000000in}{-0.058926in}}%
\pgfpathlineto{\pgfqpoint{0.035355in}{0.000000in}}%
\pgfpathlineto{\pgfqpoint{0.000000in}{0.058926in}}%
\pgfpathlineto{\pgfqpoint{-0.035355in}{0.000000in}}%
\pgfpathlineto{\pgfqpoint{-0.000000in}{-0.058926in}}%
\pgfpathclose%
\pgfusepath{stroke,fill}%
}%
\begin{pgfscope}%
\pgfsys@transformshift{3.082529in}{0.217787in}%
\pgfsys@useobject{currentmarker}{}%
\end{pgfscope}%
\end{pgfscope}%
\begin{pgfscope}%
\definecolor{textcolor}{rgb}{0.000000,0.000000,0.000000}%
\pgfsetstrokecolor{textcolor}%
\pgfsetfillcolor{textcolor}%
\pgftext[x=3.332529in,y=0.169176in,left,base]{\color{textcolor}\rmfamily\fontsize{10.000000}{12.000000}\selectfont Algebraic}%
\end{pgfscope}%
\begin{pgfscope}%
\pgfsetbuttcap%
\pgfsetroundjoin%
\pgfsetlinewidth{0.752812pt}%
\definecolor{currentstroke}{rgb}{0.200000,0.200000,0.200000}%
\pgfsetstrokecolor{currentstroke}%
\pgfsetdash{{3.750000pt}{3.750000pt}}{0.000000pt}%
\pgfpathmoveto{\pgfqpoint{4.185584in}{0.217787in}}%
\pgfpathlineto{\pgfqpoint{4.324473in}{0.217787in}}%
\pgfpathlineto{\pgfqpoint{4.463362in}{0.217787in}}%
\pgfusepath{stroke}%
\end{pgfscope}%
\begin{pgfscope}%
\definecolor{textcolor}{rgb}{0.000000,0.000000,0.000000}%
\pgfsetstrokecolor{textcolor}%
\pgfsetfillcolor{textcolor}%
\pgftext[x=4.574473in,y=0.169176in,left,base]{\color{textcolor}\rmfamily\fontsize{10.000000}{12.000000}\selectfont Timeout}%
\end{pgfscope}%
\end{pgfpicture}%
\makeatother%
\endgroup%
}
\end{subfigure}%
\caption{Performance metrics comparing the default, unoptimized version of
\SYSTEM, \SYSTEM with the formula sharing (``Sharing''), and \SYSTEM with
algebraic simplifications (``Algebraic''). An ``X'' means that the experiment
either exceeded the 10-minute timeout, or ran out of memory.}
\label{fig:optimizations}
\end{figure}
\paragraph*{\textbf{\Q{4} Discussion}}
We consider the performance impacts of our three main optimizations: precomputing the probability of taking branches containing simple guards, applying algebraic simplifications, and sharing formulas to exploit any overlap between paths.

The simple guard optimization had a large effect on our two case studies with simple guards: reservoir sampling and monotone testing.
Without this optimization, we could only verify the target properties for reservoir sampling when $(n,k) = (10,5)$ and monotone testing when $n=23$ in 128 seconds and 356 seconds, respectively.
Increasing the concretized values any higher caused solving to either exceed the 10-minute timeout, or to run out of memory.
Comparatively, we are able to explore larger concretizations under the same time limit ($(n,k) = (13,7)$, and $n = 27$) with this optimization.
These results show that it is possible to compute path probabilities more compactly, and suggest that further optimizations in this style may be possible.

To evaluate the effectiveness of the remaining two optimizations, we ran \SYSTEM on the first four case studies (we excluded the Bloom filter and count-min sketch case studies due to the number of concretized variables) with the default version of \SYSTEM, \SYSTEM with the algebraic simplifications optimization (``Algebraic''), and \SYSTEM with the formula sharing optimization (``Sharing'').
Performance results for each of the case studies are summarized in~\Cref{fig:optimizations}.

The results show that these two optimizations are not always effective on all programs, however, they can have a significant impact on performance, as showcased by the algebraic simplifications optimization on Freivalds' algorithm (\Cref{fig:freivalds_mult_opt}).
At the same time, we see that this optimization's impact is not always positive, as is the case with reservoir sampling and quicksort.
This is not so surprising as performance gains from this optimization can only occur in programs with large amounts of arithmetic, e.g., Freivalds' algorithm in the form of matrix-vector products.
The other case studies involve more comparisons and memory manipulations so the overhead of enabling this optimization is present without any of the benefits.
The sharing optimization increases performance slightly on the monotone testing case study, but otherwise saw around equal performance when compared to the stock version of \SYSTEM.

In general, ~\Cref{fig:optimizations} suggests that optimizations can significantly improve performance, but they are ultimately heuristics.
Accordingly, we envision that in order to get the most out of the optimizations, \SYSTEM should be run in parallel with different combinations enabled.
We hope to both refine and create more optimizations in future work.

\section{Further Possible Optimizations: Towards Path Filtering}
\label{sec:filtering}

In our evaluation, \Q{4} shows how simplifying the query generated by \SYSTEM can lead to significant performance gains.
In this section, we consider a different route for optimizations: path filtering.
We ask: of the paths explored by \Cref{alg:symb_ex}, is there a \emph{small} subset of paths that contain most of the probability mass? 
Clearly, we cannot hope for this to be so for all probabilistic programs; however, our experiments do suggest that such a subset does exist for some realistic probabilistic programs.
This does not immediately yield an optimization for \Cref{alg:symb_ex}---it is not clear how to prune paths automatically.
Still, it does indicate that a path pruning heuristic could further improve the performance of \Cref{alg:symb_ex}.

\paragraph*{Path partitioning}
Here we develop a path filtering scheme which aims to remove low-probability paths by selecting a subset of paths which account for the majority of the target probability mass.
This leads to smaller formulas that are hopefully faster to solve.
While the goal is simple, there are a few technical challenges.
First, the probability of a path may depend on the setting of the input variables---the same path may have high probability for some inputs, and low probability for other inputs.
Second, the space of program inputs can be very large, potentially even infinite.

To address these issues, we first partition all paths into a bounded set of \textit{abstract paths} using a user-provided partitioning function, $\Gamma: A_\forall \times A_p \rightarrow \mathbb{N}$, where $A_\forall$ is the set of all assignments of universal symbolic variables (i.e., program inputs) to values, and $A_p$ is the set of all assignments of probabilistic symbolic variables to values.
Note that a pair of assignment functions $(a_\forall,a_p) \in A_\forall\times A_p$ fully determine which path a program will follow.
So, we define $\Gamma(a_\forall,a_p) = n$ to mean that the program path described by $a_\forall$ and $a_p$ belongs to the abstract path $n$.

Referencing program paths using $a_\forall$ and $a_p$ is not a one-to-one mapping, however, as it is quite possible that there exist alternative assignment functions, $a_\forall'$ and $a_p'$ that produce the same path as $a_\forall$ and $a_p$.
For this reason we require $\Gamma$ to be \textit{path disjoint}: if $(a_\forall,a_p)$ and $(a_\forall',a_p')$ produce the same path, then $\Gamma(a_\forall,a_p) = \Gamma(a_\forall',a_p')$.
Additionally, $\Gamma$ should ideally induce an \emph{asymmetric distribution}, meaning that $\Gamma$ should partition the abstract paths such that most of the target probability mass is concentrated in a few abstract paths, allowing us to prune away a greater number of low-probability paths.

Now that we have partitioned all paths into a number of abstract paths, we need to determine which abstract paths to filter out.
We do this by estimating which abstract paths have high-probability mass using statistical sampling and counting.
For each abstract path $n$ we randomly sample $k_\forall$ assignment functions of the input variables $a_\forall^1,\ldots,a_\forall^{k_\forall}$ such that there exists an $a_p$ where $\Gamma(a_\forall^i,a_p) = n$.
In other words, we randomly sample settings of the program inputs such that it is possible to classify the resulting path as a member of the abstract path $n$.
Then, for each $a_\forall^i$, we apply the substitution $a_\forall^i$ on the input variables, execute the program $k_p$ times, and record the resulting assignments of the probabilistic variables, namely $a_p^1,\ldots,a_p^{k_p}$.
We then count how many of the produced probabilistic assignments result in the abstract path $n$.
In other words, we compute $C(a_\forall^i) \triangleq \sum_{j=1}^{k_p} [\Gamma(a_\forall^i,a_p^j)]$, where $[ \cdot ]$ are Iverson brackets.
We estimate how often the abstract path $n$ occurs by calculating $\max_{1 \leq i \leq k_\forall} C(a_\forall^i) / k_p$. 
Lastly, we prune all abstract paths that have a probability of occurring less than some user-defined threshold.
Additional details, including the formalized algorithm, are presented in~\Cref{app:filtering}.

\paragraph*{Example}
\begin{figure}[t]
  \centering
  \begin{subfigure}{0.4\textwidth}
    \centering
    \small{
      \begin{algorithmic}[1]
        \Function{CoinFlips}{$b_1, b_2$}
        \State{$heads \gets 0$}
        \For{$i\gets1,3$}
        \State{$t_1 \sim \mathsf{UniformInt}(0, 5)$}
        \State{$t_2 \sim \mathsf{UniformInt}(0, 5)$}
        \If{$t_1 \geq b_1$}
        \State{$heads \gets heads + 1$}
        \EndIf
        \If{$t_2 \geq b_2$}
        \State{$heads \gets heads + 1$}
        \EndIf
        \EndFor
        \State\Return{$heads$}
        \EndFunction
      \end{algorithmic}
    }
    \caption{The \textsc{CoinFlips} program.}
    \label{alg:coin_flips}
  \end{subfigure}
  \begin{subfigure}{0.49\textwidth}
    \centering
    \resizebox{\textwidth}{!}{\input{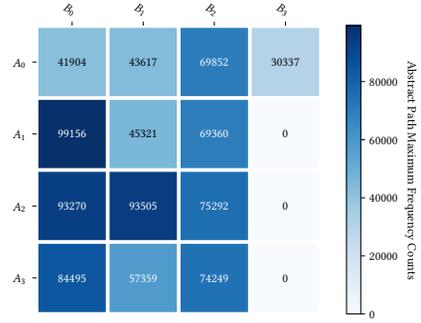}}
    \caption{Frequency count heatmap.}
    \label{fig:freq_counts_path_filtering}
  \end{subfigure}
  \caption{The \textsc{CoinFlips} program along with a heatmap showing the frequency counts for each of the 16 abstract paths defined by $\Gamma$.}
  \label{fig:path_filtering_example}
\end{figure}

To illustrate, consider the program in \Cref{alg:coin_flips}. This program generates samples from two biased coin-flip distributions three times each and returns the total number of heads.
The program inputs $b_1, b_2 \in [0, 5]$ determine the bias of the two biased coins.
We would like to verify that the \textit{maximum} expected value of $heads$ is 6 (i.e., when $b_0 = b_1 = 0$, or both coins always return heads).

First, we have to define our partitioning function $\Gamma$.
Let $A \in \{A_0,\ldots,A_3\}$ be a set of input variable classes, $B = \{B_0,\ldots,B_3\}$ be a set of probabilistic variable classes, and $t_1^i$ and $t_2^i$ be the corresponding random samples of iteration $i$ of~\Cref{alg:coin_flips}, where $1 \leq i \leq 3$.
Each input variable class corresponds to a set of assignments to the input variables, and each probabilistic variable class similarly corresponds to a set of assignment to the probabilistic variables.
In particular,
\begin{align*}
  A_0 &\triangleq b_1 > t_1^{0}  \land b_2 > t_2^{0} & B_0 &\triangleq S \geq 4 \land S \leq 6 \\
  A_1 &\triangleq b_1 \leq t_1^{0} \land b_2 \leq t_2^{0} & B_1 &\triangleq S = 3 \\
  A_2 &\triangleq b_1 \leq t_1^{0} \land b_2 > t_2^{0} & B_2 &\triangleq S = 2 \lor S = 1 \\
  A_3 &\triangleq b_1 > t_1^{0}  \land b_2 \leq t_2^{0} & B_3 &\triangleq S =0 
\end{align*}
where $S \triangleq \sum_{i=1, j=1}^{i=2, j=3} g_{ij}$ and $g_{ij} = 0$ if $t_1^{j} \geq b_i$ and $1$ otherwise.
We then define an abstract path to be a pair of input variable and probabilistic variable classes $(A_i,B_j) \in A \times B$.
Therefore, $\Gamma(a_\forall, a_p) = (A_i, B_j)$ if $a_\forall \in A_i$ and $a_p \in B_j$.

We show the maximum frequency counts for each of the 16 abstract paths as a heatmap in~\Cref{fig:freq_counts_path_filtering}.
Note that many of the abstract paths hold little probability mass, particularly those defined using $B_3$.
Using path filtering, we were able to verify that the maximum expected number of heads is indeed 6 using just 6 of the 16 abstract paths defined by $\Gamma$.
We stress that we are able to eliminate the remaining 10 abstract paths because $\Gamma$ induces an asymmetric distribution on the abstract paths.
An alternative partitioning function $\Gamma'$ which did not induce an asymmetric distribution required 12 abstract paths to verify the same property.

\paragraph*{Experimental results}
Full experimental results are presented in~\Cref{app:filtering}.
We invoked path filtering with hand-crafted partitioning functions on our case studies, and measured the percentage speedup, path reduction, and error compared to baseline \Cref{alg:symb_ex}.
To measure error, we found the minimum additive bound $\varepsilon$ such that for all inputs the probability expression $Enc_\psi$ is off by at most $\varepsilon$ from the target bound.
Some experiments saw large speedups (30.2\% for Bloom filter, 96.0\% speedup for monotone testing) and significant decreases in the number of paths (32.5\% for Bloom filter to 50.4\% for reservoir sampling), while incurring a small additive error bound $\varepsilon$ (ranging from $0.0$ to $0.18$).
These results suggest that given a sufficiently good partitioning scheme, path
filtering can produce significant performance gains for \Cref{alg:symb_ex}. The main challenge
is finding the partitioning automatically, which we leave for future work.

\section{Related Work}
\label{sec:related}

\paragraph*{Probabilistic Symbolic Execution.}
\citet{geldenhuys_2012} first proposed a method for probabilistic
symbolic execution. 
Given a standard, non-probabilistic program (i.e., programs without random sampling statements), their technique
assumes all inputs are drawn from a discrete uniform distribution, and then computes the probabilities of program paths using model counting.
Their tool then produces the posterior distribution parameterized by the return values of the programs which they use for bug finding and testing purposes.

While our technique and theirs both use a form of probabilistic symbolic variables and symbolic execution, there are a few crucial differences.
First, \SYSTEM operates on randomized programs with random sampling statements throughout the program instead of standard, deterministic programs with randomized inputs.
Second, \SYSTEM treats input variables as universally quantified, rather than assuming that inputs are uniformly distributed.
Just like typical properties of deterministic programs, target properties of randomized programs usually quantify over all input variables.
Supporting universally-quantified inputs makes symbolic execution more challenging.
While branch probabilities in \citet{geldenhuys_2012}'s setting are numeric constants, which allow \citet{geldenhuys_2012} to leverage methods like model counting and volume estimation, branch probabilities in our setting are symbolic expressions that can mention program inputs.
Accordingly, reasoning about path probabilities is more difficult in our setting.
Since all of our benchmarks aim to verify properties in the presence of unknown inputs, they cannot be handled using existing probabilistic symbolic execution methods.
Finally, \SYSTEM's main application is to verify probabilistic properties (e.g., false positive rates, expected value bounds) whereas \citet{geldenhuys_2012} uses probabilistic symbolic execution as a bug-finding tool by automatically calculating, but manually analyzing, path probabilities.

Later works use this idea for different applications: analyzing software
reliability~\citep{filieri_2013, borges_2014}, quantifying software
changes~\citep{filieri_2015}, generating performance
distribution~\citep{bihuan_2016}, and evaluating worst-case input
distributions~\citep{p4wn_2021}. Recent schemes apply volume computation instead
of model counting, which is a performance
bottleneck~\citep{sankaranarayanan_2013,ADDN17}.

Existing methods work with probabilistic programs where program inputs are either
known constants, or sampled from known distributions (often, the uniform
distribution). In contrast, our technique quantifies over all unknown
inputs, rather than assuming they are drawn from fixed distributions.

\paragraph*{Symbolic inference.}
\emph{Probabilistic programming languages} (PPLs) are languages enriched with
both \emph{sampling} and \emph{conditioning} operations. These two features allow
probabilistic programs to encode complex distributions.
In fact, many models of interest in machine learning can be expressed in
this way. A basic task is \emph{inference}: given an assertion $P$, what is the
probability that $P$ holds in the distribution described by the program?
Researchers have considered a variety of approaches, from weighted model
counting~\citep{holtzen_2020}, to analyzing Bayesian
networks~\citep{sampson_2014}, to applying computer algebra
systems~\citep{claret_2013, psi_2016, lambda_psi_2020}.

Most existing probabilistic programming languages assume that inputs are drawn
from known distributions---this is a natural simplification since PPLs are
typically concerned with analyzing a single complex distribution, rather than a
family of distributions---so they cannot be applied to prove our properties of
interest. Some recent PPLs do support reasoning about programs with unknown
parameters. Probably the most relevant such system is
\textsc{Psi}~\citep{psi_2016,lambda_psi_2020}. Like \SYSTEM, \textsc{Psi} is
designed for answering exact, symbolic queries about distributions generated by
probabilistic programs. Furthermore, \textsc{Psi} supports programs with unknown
parameters, much like our target programs. This is not a typical use-case of
\textsc{Psi}---these features aren't documented in the main paper, and we
encountered cases where \textsc{Psi} failed to compute the correct
probabilities---but \textsc{Psi} is capable of analyzing some of our benchmarks.
As our evaluation in \Cref{sec:evaluation} shows, \SYSTEM enjoys significantly
better performance and scaling when verifying probabilistic programs.

Overall, the comparison with PPLs is imperfect because PPLs are
optimized for reasoning about programs which use \emph{conditioning}, an
operation that is not supported by \SYSTEM. While conditioning is rarely used as
an operation in randomized algorithms---precisely because of its computational
intractability---it would be interesting to extend our work to handle
conditioning.

\paragraph*{Other automated methods for probabilistic programs.}
Automated verification of probabilistic programs is an active and diverse area
of research; we briefly survey several main lines of work.

\textsc{AxProf}~\citep{axprof_2019} uses \emph{statistical testing} to analyze
probabilistic programs, by running the target program multiple times on concrete
inputs in order to estimate probabilities and expected values. \textsc{AxProf}
is highly efficient and supports programs with unknown inputs. Unlike our work,
however, it can only explore a small subset of the input space and cannot
provide logical guarantees.

\emph{Probabilistic model checking} is a well-developed method for checking
logical formulas on probabilistic transition
systems~\citep{DBLP:conf/icalp/BaierCHKR97, DBLP:reference/mc/BaierAFK18,
kwiatkowska2011prism, hensel_2022}. These techniques can check properties that
are not easily handled by probabilistic symbolic execution, since assertions can be written in a variety of
temporal logics. However, our evaluation (cf.~\Q{2}) shows that probabilistic
model checkers perform much worse than our approach on our target programs and
properties.

\emph{Abstract interpretation} and \emph{algebraic program analysis} methods
have been developed for probabilistic programs~\citep{CousotM12,
DBLP:conf/pldi/WangHR18}. These methods abstract the probabilistic state, trading
precision in exchange for tractable analysis. Our method computes path
probabilities exactly in a symbolic form. It would be interesting to see if we
can leverage abstract interpretation ideas to avoid concretizing loop bounds and
unrolling loops.

There are also many \emph{domain-specific} automated analyses for specific
probabilistic properties, such as termination and resource
analysis~\citep{Chatterjee:2016:AAQ:2837614.2837639, wang_2021,
DBLP:conf/esop/MoosbruggerBKK21}, accuracy~\citep{Chakarov-martingale, SHA18},
reliability~\citep{DBLP:conf/pldi/CarbinKMR12}, differential
privacy~\citep{DBLP:journals/pacmpl/BartheCKS021,AH17} and other relational
properties~\citep{AH18, DBLP:conf/esop/FarinaCG21}, and long-run properties of
probabilistic loops~\citep{DBLP:conf/atva/BartocciKS19,
DBLP:conf/tacas/BartocciKS20}. Our approach aims to create a general-purpose analysis.

\paragraph*{Deductive verification for probabilistic programs.}
Finally, there is a wide variety of manual and semi-automated methods for
verifying probabilistic programs, which we cannot hope to fully survey here. Perhaps
the most well-developed method is Morgan and McIver's weakest pre-expectation
calculus~\citep{DBLP:journals/toplas/MorganMS96,DBLP:journals/pe/GretzKM14},
which manipulates quantitative assertions for probabilistic
programs~\citep{DBLP:journals/jcss/Kozen85}. In contrast to our method, this is
not easy to automate (though there have been some
efforts~\citep{DBLP:conf/qest/GretzKM13,BPHR21}), and targets the core
probabilistic language \textsc{pGCL} rather than a real implementation language.
Interested readers can consult the recent monograph~\citep{fopps} for an
overview of other methods.

\section{Conclusion and Future Directions}
\label{sec:conclusion}

We have presented a symbolic execution method for randomized programs to in order to automatically verify probabilistic properties which quantify over all unknown inputs.  Going forward, we see at least two promising directions for further
investigation.

\paragraph*{Optimizing probabilistic symbolic execution.}
In this work, we have made only preliminary efforts to optimize our symbolic
execution method. One natural direction is to develop heuristics for exploring
paths. Our experiments in \Cref{sec:filtering} suggest that path filtering can
reduce the number of paths that must be explored. Another possibility is to
develop better methods for simplifying path probability expressions, along the
lines of the three optimizations presented in~\Cref{sec:implementation}.

\paragraph*{Analyzing more complex probabilistic programs.}
So far, we have evaluated our implementation on standard randomized programs.
Both \textsc{KLEE} and Z3 support richer programs and hardware features. For instance, Z3
has support for reasoning about floating-point arithmetic. Recent work develops
an extension of \textsc{KLEE} that works on unbounded integers~\citep{kapus_2019}; it
could be interesting to see if this technique has better performance when
verifying randomized algorithms, which often work with mathematical integers.

\begin{acks}                            %% acks environment is optional
	%% contents suppressed with 'anonymous'
	%% Commands \grantsponsor{<sponsorID>}{<name>}{<url>} and
	%% \grantnum[<url>]{<sponsorID>}{<number>} should be used to
	%% acknowledge financial support and will be used by metadata
	%% extraction tools.
  We thank the anonymous reviewers for their detailed feedback. This work
  benefited from discussions with Dexter Kozen, Adrian Sampson, and Cornell
  PLDG. This work was partially supported by the
  \grantsponsor{GS100000001}{National Science
  Foundation}{http://dx.doi.org/10.13039/100000001} (Grant
  No.~\grantnum{GS100000001}{1943130}), the University of Wisconsin--Madison, and Cornell
  University.
  The second author is fully supported and the fourth author is partially supported by TCS Research via the TCS Research Scholar Fellowship program. 
\end{acks}

%% Bibliography
\bibliography{header,main}

%%% -*-BibTeX-*-
%%% Do NOT edit. File created by BibTeX with style
%%% ACM-Reference-Format-Journals [18-Jan-2012].

\newcommand{\SortNoop}[1]{}
\begin{thebibliography}{56}

%%% ====================================================================
%%% NOTE TO THE USER: you can override these defaults by providing
%%% customized versions of any of these macros before the \bibliography
%%% command.  Each of them MUST provide its own final punctuation,
%%% except for \shownote{}, \showDOI{}, and \showURL{}.  The latter two
%%% do not use final punctuation, in order to avoid confusing it with
%%% the Web address.
%%%
%%% To suppress output of a particular field, define its macro to expand
%%% to an empty string, or better, \unskip, like this:
%%%
%%% \newcommand{\showDOI}[1]{\unskip}   % LaTeX syntax
%%%
%%% \def \showDOI #1{\unskip}           % plain TeX syntax
%%%
%%% ====================================================================

\ifx \showCODEN    \undefined \def \showCODEN     #1{\unskip}     \fi
\ifx \showDOI      \undefined \def \showDOI       #1{#1}\fi
\ifx \showISBNx    \undefined \def \showISBNx     #1{\unskip}     \fi
\ifx \showISBNxiii \undefined \def \showISBNxiii  #1{\unskip}     \fi
\ifx \showISSN     \undefined \def \showISSN      #1{\unskip}     \fi
\ifx \showLCCN     \undefined \def \showLCCN      #1{\unskip}     \fi
\ifx \shownote     \undefined \def \shownote      #1{#1}          \fi
\ifx \showarticletitle \undefined \def \showarticletitle #1{#1}   \fi
\ifx \showURL      \undefined \def \showURL       {\relax}        \fi
% The following commands are used for tagged output and should be
% invisible to TeX
\providecommand\bibfield[2]{#2}
\providecommand\bibinfo[2]{#2}
\providecommand\natexlab[1]{#1}
\providecommand\showeprint[2][]{arXiv:#2}

\bibitem[Albarghouthi et~al\mbox{.}(2017)]%
        {ADDN17}
\bibfield{author}{\bibinfo{person}{Aws Albarghouthi}, \bibinfo{person}{Loris
  D'Antoni}, \bibinfo{person}{Samuel Drews}, {and} \bibinfo{person}{Aditya~V.
  Nori}.} \bibinfo{year}{2017}\natexlab{}.
\newblock \showarticletitle{{FairSquare}: Probabilistic Verification of Program
  Fairness}.
\newblock \bibinfo{journal}{\emph{Proceedings of the {ACM} on Programming
  Languages}} \bibinfo{volume}{1}, \bibinfo{number}{OOPSLA}, Article
  \bibinfo{articleno}{80} (\bibinfo{year}{2017}).
\newblock
\urldef\tempurl%
\url{https://doi.org/10.1145/3133904}
\showDOI{\tempurl}


\bibitem[Albarghouthi and Hsu(2018a)]%
        {AH18}
\bibfield{author}{\bibinfo{person}{Aws Albarghouthi} {and}
  \bibinfo{person}{Justin Hsu}.} \bibinfo{year}{2018}\natexlab{a}.
\newblock \showarticletitle{Constraint-Based Synthesis of Coupling Proofs}. In
  \bibinfo{booktitle}{\emph{International Conference on Computer Aided
  Verification (CAV), Oxford, England}}.
\newblock
\urldef\tempurl%
\url{https://doi.org/10.1007/978-3-319-96145-3_18}
\showDOI{\tempurl}
\showeprint[arxiv]{1804.04052}


\bibitem[Albarghouthi and Hsu(2018b)]%
        {AH17}
\bibfield{author}{\bibinfo{person}{Aws Albarghouthi} {and}
  \bibinfo{person}{Justin Hsu}.} \bibinfo{year}{2018}\natexlab{b}.
\newblock \showarticletitle{Synthesizing Coupling Proofs of Differential
  Privacy}.
\newblock \bibinfo{journal}{\emph{Proceedings of the {ACM} on Programming
  Languages}} \bibinfo{volume}{2}, \bibinfo{number}{POPL}, Article
  \bibinfo{articleno}{58} (\bibinfo{date}{Jan.} \bibinfo{year}{2018}).
\newblock
\urldef\tempurl%
\url{https://doi.org/10.1145/3158146}
\showDOI{\tempurl}
\showeprint[arxiv]{1709.05361}


\bibitem[Baier et~al\mbox{.}(1997)]%
        {DBLP:conf/icalp/BaierCHKR97}
\bibfield{author}{\bibinfo{person}{Christel Baier}, \bibinfo{person}{Edmund~M.
  Clarke}, \bibinfo{person}{Vasiliki Hartonas{-}Garmhausen},
  \bibinfo{person}{Marta~Z. Kwiatkowska}, {and} \bibinfo{person}{Mark Ryan}.}
  \bibinfo{year}{1997}\natexlab{}.
\newblock \showarticletitle{Symbolic Model Checking for Probabilistic
  Processes}. In \bibinfo{booktitle}{\emph{International Colloquium on
  Automata, Languages and Programming (ICALP), Bologna, Italy}}
  \emph{(\bibinfo{series}{Lecture Notes in Computer Science},
  Vol.~\bibinfo{volume}{1256})}. \bibinfo{publisher}{Springer},
  \bibinfo{pages}{430--440}.
\newblock
\urldef\tempurl%
\url{https://doi.org/10.1007/3-540-63165-8\_199}
\showDOI{\tempurl}


\bibitem[Baier et~al\mbox{.}(2018)]%
        {DBLP:reference/mc/BaierAFK18}
\bibfield{author}{\bibinfo{person}{Christel Baier}, \bibinfo{person}{Luca de
  Alfaro}, \bibinfo{person}{Vojtech Forejt}, {and} \bibinfo{person}{Marta
  Kwiatkowska}.} \bibinfo{year}{2018}\natexlab{}.
\newblock \showarticletitle{Model Checking Probabilistic Systems}.
\newblock In \bibinfo{booktitle}{\emph{Handbook of Model Checking}}.
  \bibinfo{publisher}{Springer-Verlag}, \bibinfo{pages}{963--999}.
\newblock
\urldef\tempurl%
\url{https://doi.org/10.1007/978-3-319-10575-8_28}
\showDOI{\tempurl}


\bibitem[Bao et~al\mbox{.}(2022)]%
        {BPHR21}
\bibfield{author}{\bibinfo{person}{Jialu Bao}, \bibinfo{person}{Nitesh
  Trivedi}, \bibinfo{person}{Drashti Pathak}, \bibinfo{person}{Justin Hsu},
  {and} \bibinfo{person}{Subhajit Roy}.} \bibinfo{year}{2022}\natexlab{}.
\newblock \showarticletitle{Data-Driven Invariant Learning for Probabilistic
  Programs}. In \bibinfo{booktitle}{\emph{International Conference on Computer
  Aided Verification (CAV), Haifa, Israel}}.
\newblock
\urldef\tempurl%
\url{https://doi.org/10.1007/978-3-031-13185-1\_3}
\showDOI{\tempurl}
\showeprint[arxiv]{2106.05421}


\bibitem[Barrett et~al\mbox{.}(2016)]%
        {smtlib}
\bibfield{author}{\bibinfo{person}{Clark Barrett}, \bibinfo{person}{Pascal
  Fontaine}, {and} \bibinfo{person}{Cesare Tinelli}.}
  \bibinfo{year}{2016}\natexlab{}.
\newblock \bibinfo{title}{The Satisfiability Modulo Theories Library
  (SMT-LIB)}.
\newblock \bibinfo{howpublished}{{\tt www.SMT-LIB.org}}.
\newblock


\bibitem[Barthe et~al\mbox{.}(2021)]%
        {DBLP:journals/pacmpl/BartheCKS021}
\bibfield{author}{\bibinfo{person}{Gilles Barthe}, \bibinfo{person}{Rohit
  Chadha}, \bibinfo{person}{Paul Krogmeier}, \bibinfo{person}{A.~Prasad
  Sistla}, {and} \bibinfo{person}{Mahesh Viswanathan}.}
  \bibinfo{year}{2021}\natexlab{}.
\newblock \showarticletitle{Deciding Accuracy of Differential Privacy Schemes}.
\newblock \bibinfo{journal}{\emph{Proceedings of the {ACM} on Programming
  Languages}} \bibinfo{volume}{5}, \bibinfo{number}{POPL}, Article
  \bibinfo{articleno}{8} (\bibinfo{date}{Jan.} \bibinfo{year}{2021}).
\newblock
\urldef\tempurl%
\url{https://doi.org/10.1145/3434289}
\showDOI{\tempurl}


\bibitem[Barthe et~al\mbox{.}(2020)]%
        {fopps}
\bibfield{editor}{\bibinfo{person}{Gilles Barthe},
  \bibinfo{person}{Joost-Pieter Katoen}, {and} \bibinfo{person}{Alexandra
  Silva}} (Eds.). \bibinfo{year}{2020}\natexlab{}.
\newblock \bibinfo{booktitle}{\emph{Foundations of Probabilistic Programming
  Languages}}.
\newblock \bibinfo{publisher}{Cambridge University Press}. 145--184 pages.
\newblock
\urldef\tempurl%
\url{https://doi.org/10.1017/9781108770750.006}
\showDOI{\tempurl}


\bibitem[Bartocci et~al\mbox{.}(2019)]%
        {DBLP:conf/atva/BartocciKS19}
\bibfield{author}{\bibinfo{person}{Ezio Bartocci}, \bibinfo{person}{Laura
  Kov{\'{a}}cs}, {and} \bibinfo{person}{Miroslav Stankovic}.}
  \bibinfo{year}{2019}\natexlab{}.
\newblock \showarticletitle{Automatic Generation of Moment-Based Invariants for
  Prob-Solvable Loops}. In \bibinfo{booktitle}{\emph{International Symposium on
  Automated Technology for Verification and Analysis (ATVA), Taipei City,
  Taiwan}} \emph{(\bibinfo{series}{Lecture Notes in Computer Science},
  Vol.~\bibinfo{volume}{11781})}. \bibinfo{publisher}{Springer-Verlag},
  \bibinfo{pages}{255--276}.
\newblock
\urldef\tempurl%
\url{https://doi.org/10.1007/978-3-030-31784-3\_15}
\showDOI{\tempurl}


\bibitem[Bartocci et~al\mbox{.}(2020)]%
        {DBLP:conf/tacas/BartocciKS20}
\bibfield{author}{\bibinfo{person}{Ezio Bartocci}, \bibinfo{person}{Laura
  Kov{\'{a}}cs}, {and} \bibinfo{person}{Miroslav Stankovic}.}
  \bibinfo{year}{2020}\natexlab{}.
\newblock \showarticletitle{Mora -- Automatic Generation of Moment-Based
  Invariants}. In \bibinfo{booktitle}{\emph{International Conference on Tools
  and Algorithms for the Construction and Analysis of Systems (TACAS), Dublin,
  Ireland}} \emph{(\bibinfo{series}{Lecture Notes in Computer Science},
  Vol.~\bibinfo{volume}{12078})}. \bibinfo{publisher}{Springer-Verlag},
  \bibinfo{pages}{492--498}.
\newblock
\urldef\tempurl%
\url{https://doi.org/10.1007/978-3-030-45190-5\_28}
\showDOI{\tempurl}


\bibitem[Bessey et~al\mbox{.}(2010)]%
        {bessey_2010}
\bibfield{author}{\bibinfo{person}{Al Bessey}, \bibinfo{person}{Ken Block},
  \bibinfo{person}{Ben Chelf}, \bibinfo{person}{Andy Chou},
  \bibinfo{person}{Bryan Fulton}, \bibinfo{person}{Seth Hallem},
  \bibinfo{person}{Charles Henri-Gros}, \bibinfo{person}{Asya Kamsky},
  \bibinfo{person}{Scott McPeak}, {and} \bibinfo{person}{Dawson Engler}.}
  \bibinfo{year}{2010}\natexlab{}.
\newblock \showarticletitle{A Few Billion Lines of Code Later: Using Static
  Analysis to Find Bugs in the Real World}.
\newblock \bibinfo{journal}{\emph{Commun. ACM}} \bibinfo{volume}{53},
  \bibinfo{number}{2} (\bibinfo{date}{Feb.} \bibinfo{year}{2010}),
  \bibinfo{pages}{66--75}.
\newblock
\showISSN{0001-0782}
\urldef\tempurl%
\url{https://doi.org/10.1145/1646353.1646374}
\showDOI{\tempurl}


\bibitem[Bloom(1970)]%
        {bloom_1970}
\bibfield{author}{\bibinfo{person}{Burton~H. Bloom}.}
  \bibinfo{year}{1970}\natexlab{}.
\newblock \showarticletitle{Space/Time Trade-offs in Hash Coding with Allowable
  Errors}.
\newblock \bibinfo{journal}{\emph{Commun. ACM}} \bibinfo{volume}{13},
  \bibinfo{number}{7} (\bibinfo{date}{July} \bibinfo{year}{1970}),
  \bibinfo{pages}{422--426}.
\newblock
\showISSN{0001-0782}
\urldef\tempurl%
\url{https://doi.org/10.1145/362686.362692}
\showDOI{\tempurl}


\bibitem[Borges et~al\mbox{.}(2014)]%
        {borges_2014}
\bibfield{author}{\bibinfo{person}{Mateus Borges}, \bibinfo{person}{Antonio
  Filieri}, \bibinfo{person}{Marcelo d'Amorim}, \bibinfo{person}{Corina~S.
  P\u{a}s\u{a}reanu}, {and} \bibinfo{person}{Willem Visser}.}
  \bibinfo{year}{2014}\natexlab{}.
\newblock \showarticletitle{Compositional Solution Space Quantification for
  Probabilistic Software Analysis}. In \bibinfo{booktitle}{\emph{{ACM SIGPLAN
  Conference on Programming Language Design and Implementation (PLDI)},
  Edinburgh, Scotland}}. \bibinfo{pages}{123--132}.
\newblock
\urldef\tempurl%
\url{https://doi.org/10.1145/2666356.2594329}
\showDOI{\tempurl}


\bibitem[Cadar et~al\mbox{.}(2008)]%
        {cadar2008}
\bibfield{author}{\bibinfo{person}{Cristian Cadar}, \bibinfo{person}{Daniel
  Dunbar}, {and} \bibinfo{person}{Dawson Engler}.}
  \bibinfo{year}{2008}\natexlab{}.
\newblock \showarticletitle{{KLEE}: Unassisted and Automatic Generation of
  High-Coverage Tests for Complex Systems Programs}. In
  \bibinfo{booktitle}{\emph{USENIX Symposium on Operating Systems Design and
  Implementation (OSDI), San Diego, California}}. \bibinfo{pages}{209--224}.
\newblock
\urldef\tempurl%
\url{http://www.usenix.org/events/osdi08/tech/full\_papers/cadar/cadar.pdf}
\showURL{%
\tempurl}


\bibitem[Carbin et~al\mbox{.}(2012)]%
        {DBLP:conf/pldi/CarbinKMR12}
\bibfield{author}{\bibinfo{person}{Michael Carbin}, \bibinfo{person}{Deokhwan
  Kim}, \bibinfo{person}{Sasa Misailovic}, {and} \bibinfo{person}{Martin~C
  Rinard}.} \bibinfo{year}{2012}\natexlab{}.
\newblock \showarticletitle{Proving Acceptability Properties of Relaxed
  Nondeterministic Approximate Programs}. In \bibinfo{booktitle}{\emph{{ACM
  SIGPLAN Conference on Programming Language Design and Implementation (PLDI)},
  Beijing, China}}. \bibinfo{pages}{169--180}.
\newblock
\urldef\tempurl%
\url{https://doi.org/10.1145/2254064.2254086}
\showDOI{\tempurl}


\bibitem[Chakarov and Sankaranarayanan(2013)]%
        {Chakarov-martingale}
\bibfield{author}{\bibinfo{person}{Aleksandar Chakarov} {and}
  \bibinfo{person}{Sriram Sankaranarayanan}.} \bibinfo{year}{2013}\natexlab{}.
\newblock \showarticletitle{Probabilistic Program Analysis with Martingales}.
  In \bibinfo{booktitle}{\emph{International Conference on Computer Aided
  Verification (CAV), Saint Petersburg, Russia}}
  \emph{(\bibinfo{series}{Lecture Notes in Computer Science},
  Vol.~\bibinfo{volume}{8044})}. \bibinfo{pages}{511--526}.
\newblock
\urldef\tempurl%
\url{https://doi.org/10.1007/978-3-642-39799-8_34}
\showDOI{\tempurl}


\bibitem[Chatterjee et~al\mbox{.}(2016)]%
        {Chatterjee:2016:AAQ:2837614.2837639}
\bibfield{author}{\bibinfo{person}{Krishnendu Chatterjee},
  \bibinfo{person}{Hongfei Fu}, \bibinfo{person}{Petr Novotn\'{y}}, {and}
  \bibinfo{person}{Rouzbeh Hasheminezhad}.} \bibinfo{year}{2016}\natexlab{}.
\newblock \showarticletitle{Algorithmic Analysis of Qualitative and
  Quantitative Termination Problems for Affine Probabilistic Programs}. In
  \bibinfo{booktitle}{\emph{{ACM} {SIGPLAN--SIGACT} {S}ymposium on {P}rinciples
  of {P}rogramming {L}anguages ({POPL}), Saint Petersburg, Florida}}.
  \bibinfo{pages}{327--342}.
\newblock
\showISBNx{978-1-4503-3549-2}
\urldef\tempurl%
\url{https://doi.org/10.1145/2837614.2837639}
\showDOI{\tempurl}


\bibitem[Chen et~al\mbox{.}(2016)]%
        {bihuan_2016}
\bibfield{author}{\bibinfo{person}{Bihuan Chen}, \bibinfo{person}{Yang Liu},
  {and} \bibinfo{person}{Wei Le}.} \bibinfo{year}{2016}\natexlab{}.
\newblock \showarticletitle{Generating Performance Distributions via
  Probabilistic Symbolic Execution}. In \bibinfo{booktitle}{\emph{International
  Conference on Software Engineering (ICSE), Austin, Texas}}.
  \bibinfo{pages}{49--60}.
\newblock
\showISBNx{9781450339001}
\urldef\tempurl%
\url{https://doi.org/10.1145/2884781.2884794}
\showDOI{\tempurl}


\bibitem[Claret et~al\mbox{.}(2013)]%
        {claret_2013}
\bibfield{author}{\bibinfo{person}{Guillaume Claret},
  \bibinfo{person}{Sriram~K. Rajamani}, \bibinfo{person}{Aditya~V. Nori},
  \bibinfo{person}{Andrew~D. Gordon}, {and} \bibinfo{person}{Johannes
  Borgstr\"{o}m}.} \bibinfo{year}{2013}\natexlab{}.
\newblock \showarticletitle{Bayesian Inference Using Data Flow Analysis}. In
  \bibinfo{booktitle}{\emph{Joint Meeting of the European Software Engineering
  Conference and the {ACM} {SIGSOFT} Symposium on the Foundations of Software
  Engineering (ESEC/FSE), Saint Petersburg, Russia}}. \bibinfo{pages}{92--102}.
\newblock
\urldef\tempurl%
\url{https://doi.org/10.1145/2491411.2491423}
\showDOI{\tempurl}


\bibitem[Cormode and Muthukrishnan(2004)]%
        {graham_2005}
\bibfield{author}{\bibinfo{person}{Graham Cormode} {and} \bibinfo{person}{S.
  Muthukrishnan}.} \bibinfo{year}{2004}\natexlab{}.
\newblock \showarticletitle{An Improved Data Stream Summary: The Count-Min
  Sketch and Its Applications}. In \bibinfo{booktitle}{\emph{Latin American
  Symposium on Theoretical Informatics (LATIN), Buenos Aires, Argentina}}
  \emph{(\bibinfo{series}{Lecture Notes in Computer Science},
  Vol.~\bibinfo{volume}{2976})}. \bibinfo{publisher}{Springer-Verlag},
  \bibinfo{pages}{29--38}.
\newblock
\urldef\tempurl%
\url{https://doi.org/10.1007/978-3-540-24698-5\_7}
\showDOI{\tempurl}


\bibitem[Cousot and Monerau(2012)]%
        {CousotM12}
\bibfield{author}{\bibinfo{person}{Patrick Cousot} {and}
  \bibinfo{person}{Michael Monerau}.} \bibinfo{year}{2012}\natexlab{}.
\newblock \showarticletitle{Probabilistic Abstract Interpretation}. In
  \bibinfo{booktitle}{\emph{European Symposium on Programming (ESOP), Tallinn,
  Estonia}} \emph{(\bibinfo{series}{Lecture Notes in Computer Science},
  Vol.~\bibinfo{volume}{7211})}. \bibinfo{publisher}{Springer-Verlag},
  \bibinfo{pages}{169--193}.
\newblock
\urldef\tempurl%
\url{https://doi.org/10.1007/978-3-642-28869-2\_9}
\showDOI{\tempurl}


\bibitem[de~Moura and Bj{\o}rner(2008)]%
        {demoura2008}
\bibfield{author}{\bibinfo{person}{Leonardo de Moura} {and}
  \bibinfo{person}{Nikolaj Bj{\o}rner}.} \bibinfo{year}{2008}\natexlab{}.
\newblock \showarticletitle{Z3: An Efficient SMT Solver}. In
  \bibinfo{booktitle}{\emph{International Conference on Tools and Algorithms
  for the Construction and Analysis of Systems (TACAS), Budapest, Hungary}}
  \emph{(\bibinfo{series}{Lecture Notes in Computer Science},
  Vol.~\bibinfo{volume}{4963})}. \bibinfo{publisher}{Springer-Verlag},
  \bibinfo{pages}{337--340}.
\newblock
\showISBNx{978-3-540-78800-3}
\urldef\tempurl%
\url{https://doi.org/10.1007/978-3-540-78800-3_24}
\showDOI{\tempurl}


\bibitem[Farina et~al\mbox{.}(2019)]%
        {farina_2019}
\bibfield{author}{\bibinfo{person}{Gian~Pietro Farina},
  \bibinfo{person}{Stephen Chong}, {and} \bibinfo{person}{Marco Gaboardi}.}
  \bibinfo{year}{2019}\natexlab{}.
\newblock \showarticletitle{Relational Symbolic Execution}. In
  \bibinfo{booktitle}{\emph{ACM SIGPLAN International Conference on Principles
  and Practice of Declarative Programming (PPDP), Porto, Portugal}}. Article
  \bibinfo{articleno}{10}.
\newblock
\urldef\tempurl%
\url{https://doi.org/10.1145/3354166.3354175}
\showDOI{\tempurl}


\bibitem[Farina et~al\mbox{.}(2021)]%
        {DBLP:conf/esop/FarinaCG21}
\bibfield{author}{\bibinfo{person}{Gian~Pietro Farina},
  \bibinfo{person}{Stephen Chong}, {and} \bibinfo{person}{Marco Gaboardi}.}
  \bibinfo{year}{2021}\natexlab{}.
\newblock \showarticletitle{Coupled Relational Symbolic Execution for
  Differential Privacy}. In \bibinfo{booktitle}{\emph{European Symposium on
  Programming (ESOP), Luxembourg City, Luxembourg}}
  \emph{(\bibinfo{series}{Lecture Notes in Computer Science},
  Vol.~\bibinfo{volume}{12648})}. \bibinfo{publisher}{Springer-Verlag},
  \bibinfo{pages}{207--233}.
\newblock
\urldef\tempurl%
\url{https://doi.org/10.1007/978-3-030-72019-3\_8}
\showDOI{\tempurl}


\bibitem[Filieri et~al\mbox{.}(2013)]%
        {filieri_2013}
\bibfield{author}{\bibinfo{person}{Antonio Filieri}, \bibinfo{person}{Corina~S.
  P\u{a}s\u{a}reanu}, {and} \bibinfo{person}{Willem Visser}.}
  \bibinfo{year}{2013}\natexlab{}.
\newblock \showarticletitle{Reliability Analysis in Symbolic Pathfinder}. In
  \bibinfo{booktitle}{\emph{International Conference on Software Engineering
  (ICSE), San Francisco, California}}. \bibinfo{pages}{622--631}.
\newblock
\urldef\tempurl%
\url{https://doi.org/10.1109/ICSE.2013.6606608}
\showDOI{\tempurl}


\bibitem[Filieri et~al\mbox{.}(2015)]%
        {filieri_2015}
\bibfield{author}{\bibinfo{person}{Antonio Filieri}, \bibinfo{person}{Corina~S.
  P\u{a}s\u{a}reanu}, {and} \bibinfo{person}{Guowei Yang}.}
  \bibinfo{year}{2015}\natexlab{}.
\newblock \showarticletitle{Quantification of Software Changes through
  Probabilistic Symbolic Execution}. In \bibinfo{booktitle}{\emph{IEEE/ACM
  International Conference on Automated Software Engineering (ASE), Lincoln,
  Nebraska}}. \bibinfo{pages}{703--708}.
\newblock
\urldef\tempurl%
\url{https://doi.org/10.1109/ASE.2015.78}
\showDOI{\tempurl}


\bibitem[Freivalds(1977)]%
        {freivalds1977}
\bibfield{author}{\bibinfo{person}{R\={u}si\c{n}\v{s} Freivalds}.}
  \bibinfo{year}{1977}\natexlab{}.
\newblock \showarticletitle{Probabilistic Machines Can Use Less Running Time}.
  In \bibinfo{booktitle}{\emph{{IFIP} World Congress, Toronto, Canada}}.
  \bibinfo{publisher}{North-Holland}, \bibinfo{pages}{839--842}.
\newblock


\bibitem[Gehr et~al\mbox{.}(2016)]%
        {psi_2016}
\bibfield{author}{\bibinfo{person}{Timon Gehr}, \bibinfo{person}{Sasa
  Misailovic}, {and} \bibinfo{person}{Martin Vechev}.}
  \bibinfo{year}{2016}\natexlab{}.
\newblock \showarticletitle{{PSI}: Exact Symbolic Inference for Probabilistic
  Programs}. In \bibinfo{booktitle}{\emph{International Conference on Computer
  Aided Verification (CAV), Toronto, Ontario}} \emph{(\bibinfo{series}{Lecture
  Notes in Computer Science}, Vol.~\bibinfo{volume}{9779})}.
  \bibinfo{publisher}{Springer-Verlag}, \bibinfo{pages}{62--83}.
\newblock
\urldef\tempurl%
\url{https://doi.org/10.1007/978-3-319-41528-4\_4}
\showDOI{\tempurl}


\bibitem[Gehr et~al\mbox{.}(2020)]%
        {lambda_psi_2020}
\bibfield{author}{\bibinfo{person}{Timon Gehr}, \bibinfo{person}{Samuel
  Steffen}, {and} \bibinfo{person}{Martin Vechev}.}
  \bibinfo{year}{2020}\natexlab{}.
\newblock \showarticletitle{{\(\lambda\)PSI}: Exact Inference for Higher-Order
  Probabilistic Programs}. In \bibinfo{booktitle}{\emph{{ACM SIGPLAN Conference
  on Programming Language Design and Implementation (PLDI)}, London, England}}.
  \bibinfo{pages}{883--897}.
\newblock
\urldef\tempurl%
\url{https://doi.org/10.1145/3385412.3386006}
\showDOI{\tempurl}


\bibitem[Geldenhuys et~al\mbox{.}(2012)]%
        {geldenhuys_2012}
\bibfield{author}{\bibinfo{person}{Jaco Geldenhuys},
  \bibinfo{person}{Matthew~B. Dwyer}, {and} \bibinfo{person}{Willem Visser}.}
  \bibinfo{year}{2012}\natexlab{}.
\newblock \showarticletitle{Probabilistic Symbolic Execution}. In
  \bibinfo{booktitle}{\emph{{ACM} {SIGSOFT} International Symposium on Software
  Testing and Analysis {(ISSTA)}, Minneapolis, Minnesota}}.
  \bibinfo{pages}{166--176}.
\newblock
\urldef\tempurl%
\url{https://doi.org/10.1145/2338965.2336773}
\showDOI{\tempurl}


\bibitem[Goldreich(2017)]%
        {goldreich_2017}
\bibfield{author}{\bibinfo{person}{Oded Goldreich}.}
  \bibinfo{year}{2017}\natexlab{}.
\newblock \bibinfo{booktitle}{\emph{Introduction to Property Testing}}.
\newblock \bibinfo{publisher}{Cambridge University Press}.
\newblock
\showISBNx{978-1-107-19405-2}
\urldef\tempurl%
\url{https://doi.org/10.1017/9781108135252}
\showDOI{\tempurl}


\bibitem[Gretz et~al\mbox{.}(2013)]%
        {DBLP:conf/qest/GretzKM13}
\bibfield{author}{\bibinfo{person}{Friedrich Gretz},
  \bibinfo{person}{Joost{-}Pieter Katoen}, {and} \bibinfo{person}{Annabelle
  McIver}.} \bibinfo{year}{2013}\natexlab{}.
\newblock \showarticletitle{\textsc{Prinsys}---On a Quest for Probabilistic
  Loop Invariants}. In \bibinfo{booktitle}{\emph{International Conference on
  Quantitative Evaluation of Systems (QEST), Buenos Aires, Argentina}}
  \emph{(\bibinfo{series}{Lecture Notes in Computer Science},
  Vol.~\bibinfo{volume}{8054})}. \bibinfo{publisher}{Springer-Verlag},
  \bibinfo{pages}{193--208}.
\newblock
\urldef\tempurl%
\url{https://doi.org/10.1007/978-3-642-40196-1\_17}
\showDOI{\tempurl}


\bibitem[Gretz et~al\mbox{.}(2014)]%
        {DBLP:journals/pe/GretzKM14}
\bibfield{author}{\bibinfo{person}{Friedrich Gretz},
  \bibinfo{person}{Joost{-}Pieter Katoen}, {and} \bibinfo{person}{Annabelle
  McIver}.} \bibinfo{year}{2014}\natexlab{}.
\newblock \showarticletitle{Operational versus weakest pre-expectation
  semantics for the probabilistic guarded command language}.
\newblock \bibinfo{journal}{\emph{Performance Evaluation}}
  \bibinfo{volume}{73} (\bibinfo{year}{2014}), \bibinfo{pages}{110--132}.
\newblock
\urldef\tempurl%
\url{https://doi.org/10.1016/j.peva.2013.11.004}
\showDOI{\tempurl}


\bibitem[Hensel et~al\mbox{.}(2022)]%
        {hensel_2022}
\bibfield{author}{\bibinfo{person}{Christian Hensel},
  \bibinfo{person}{Sebastian Junges}, \bibinfo{person}{Joost{-}Pieter Katoen},
  \bibinfo{person}{Tim Quatmann}, {and} \bibinfo{person}{Matthias Volk}.}
  \bibinfo{year}{2022}\natexlab{}.
\newblock \showarticletitle{The probabilistic model checker \textsc{Storm}}.
\newblock \bibinfo{journal}{\emph{International Journal on Software Tools for
  Technology Transfer}} \bibinfo{volume}{24}, \bibinfo{number}{4}
  (\bibinfo{year}{2022}), \bibinfo{pages}{589--610}.
\newblock
\urldef\tempurl%
\url{https://doi.org/10.1007/s10009-021-00633-z}
\showDOI{\tempurl}
\showeprint[arxiv]{2002.07080}


\bibitem[Hoare(1961)]%
        {quicksort}
\bibfield{author}{\bibinfo{person}{C.~A.~R. Hoare}.}
  \bibinfo{year}{1961}\natexlab{}.
\newblock \showarticletitle{Algorithms 63-64: partition and quicksort}.
\newblock \bibinfo{journal}{\emph{Commun. ACM}} \bibinfo{volume}{4},
  \bibinfo{number}{7} (\bibinfo{date}{July} \bibinfo{year}{1961}),
  \bibinfo{pages}{321}.
\newblock
\urldef\tempurl%
\url{https://doi.org/10.1145/366622.366642}
\showDOI{\tempurl}


\bibitem[Holtzen et~al\mbox{.}(2020)]%
        {holtzen_2020}
\bibfield{author}{\bibinfo{person}{Steven Holtzen}, \bibinfo{person}{Guy
  Van~den Broeck}, {and} \bibinfo{person}{Todd Millstein}.}
  \bibinfo{year}{2020}\natexlab{}.
\newblock \showarticletitle{Scaling Exact Inference for Discrete Probabilistic
  Programs}.
\newblock \bibinfo{journal}{\emph{Proceedings of the {ACM} on Programming
  Languages}} \bibinfo{volume}{4}, \bibinfo{number}{{OOPSLA}}, Article
  \bibinfo{articleno}{140} (\bibinfo{date}{Nov.} \bibinfo{year}{2020}).
\newblock
\urldef\tempurl%
\url{https://doi.org/10.1145/3428208}
\showDOI{\tempurl}


\bibitem[Joshi et~al\mbox{.}(2019)]%
        {axprof_2019}
\bibfield{author}{\bibinfo{person}{Keyur Joshi}, \bibinfo{person}{Vimuth
  Fernando}, {and} \bibinfo{person}{Sasa Misailovic}.}
  \bibinfo{year}{2019}\natexlab{}.
\newblock \showarticletitle{Statistical Algorithmic Profiling for Randomized
  Approximate Programs}. In \bibinfo{booktitle}{\emph{International Conference
  on Software Engineering (ICSE), Montr{\'e}al, Qu{\'e}bec}}.
  \bibinfo{pages}{608--618}.
\newblock
\urldef\tempurl%
\url{https://doi.org/10.1109/ICSE.2019.00071}
\showDOI{\tempurl}


\bibitem[Kang et~al\mbox{.}(2021)]%
        {p4wn_2021}
\bibfield{author}{\bibinfo{person}{Qiao Kang}, \bibinfo{person}{Jiarong Xing},
  \bibinfo{person}{Yiming Qiu}, {and} \bibinfo{person}{Ang Chen}.}
  \bibinfo{year}{2021}\natexlab{}.
\newblock \showarticletitle{Probabilistic Profiling of Stateful Data Planes for
  Adversarial Testing}. In \bibinfo{booktitle}{\emph{International Conference
  on Architectural Support for Programming Langauages and Operating Systems
  (ASPLOS)}}. \bibinfo{pages}{286--301}.
\newblock
\urldef\tempurl%
\url{https://doi.org/10.1145/3445814.3446764}
\showDOI{\tempurl}


\bibitem[Kapus et~al\mbox{.}(2019)]%
        {kapus_2019}
\bibfield{author}{\bibinfo{person}{Timotej Kapus}, \bibinfo{person}{Martin
  Nowack}, {and} \bibinfo{person}{Cristian Cadar}.}
  \bibinfo{year}{2019}\natexlab{}.
\newblock \showarticletitle{Constraints in Dynamic Symbolic Execution:
  Bitvectors or Integers?}. In \bibinfo{booktitle}{\emph{International
  Conference on Tests and Proofs (TAP), Porto, Portugal}}
  \emph{(\bibinfo{series}{Lecture Notes in Computer Science},
  Vol.~\bibinfo{volume}{11823})}. \bibinfo{publisher}{Springer-Verlag},
  \bibinfo{pages}{41--54}.
\newblock
\urldef\tempurl%
\url{https://doi.org/10.1007/978-3-030-31157-5\_3}
\showDOI{\tempurl}


\bibitem[Kiezun et~al\mbox{.}(2012)]%
        {adam_2012}
\bibfield{author}{\bibinfo{person}{Adam Kiezun}, \bibinfo{person}{Vijay
  Ganesh}, \bibinfo{person}{Shay Artzi}, \bibinfo{person}{Philip~J. Guo},
  \bibinfo{person}{Pieter Hooimeijer}, {and} \bibinfo{person}{Michael~D.
  Ernst}.} \bibinfo{year}{2012}\natexlab{}.
\newblock \showarticletitle{HAMPI: A Solver for Word Equations over Strings,
  Regular Expressions, and Context-Free Grammars}.
\newblock \bibinfo{journal}{\emph{ACM Transactions on Software Engineering and
  Methodology}} \bibinfo{volume}{21}, \bibinfo{number}{4}, Article
  \bibinfo{articleno}{25} (\bibinfo{date}{Nov.} \bibinfo{year}{2012}).
\newblock
\urldef\tempurl%
\url{https://doi.org/10.1145/2377656.2377662}
\showDOI{\tempurl}


\bibitem[King(1976)]%
        {king1976}
\bibfield{author}{\bibinfo{person}{James~C. King}.}
  \bibinfo{year}{1976}\natexlab{}.
\newblock \showarticletitle{Symbolic Execution and Program Testing}.
\newblock \bibinfo{journal}{\emph{Commun. ACM}} \bibinfo{volume}{19},
  \bibinfo{number}{7} (\bibinfo{date}{July} \bibinfo{year}{1976}),
  \bibinfo{pages}{385--394}.
\newblock
\urldef\tempurl%
\url{https://doi.org/10.1145/360248.360252}
\showDOI{\tempurl}


\bibitem[Kozen(1985)]%
        {DBLP:journals/jcss/Kozen85}
\bibfield{author}{\bibinfo{person}{Dexter Kozen}.}
  \bibinfo{year}{1985}\natexlab{}.
\newblock \showarticletitle{A Probabilistic {PDL}}.
\newblock \bibinfo{journal}{\emph{J. Comput. System Sci.}}
  \bibinfo{volume}{30}, \bibinfo{number}{2} (\bibinfo{year}{1985}),
  \bibinfo{pages}{162--178}.
\newblock
\urldef\tempurl%
\url{https://doi.org/10.1016/0022-0000(85)90012-1}
\showDOI{\tempurl}


\bibitem[Kwiatkowska et~al\mbox{.}(2011)]%
        {kwiatkowska2011prism}
\bibfield{author}{\bibinfo{person}{Marta Kwiatkowska}, \bibinfo{person}{Gethin
  Norman}, {and} \bibinfo{person}{David Parker}.}
  \bibinfo{year}{2011}\natexlab{}.
\newblock \showarticletitle{{PRISM} 4.0: Verification of Probabilistic
  Real-Time Systems}. In \bibinfo{booktitle}{\emph{International Conference on
  Computer Aided Verification (CAV), Snowbird, Utah}}
  \emph{(\bibinfo{series}{Lecture Notes in Computer Science},
  Vol.~\bibinfo{volume}{6806})}. \bibinfo{publisher}{Springer-Verlag},
  \bibinfo{pages}{585--591}.
\newblock
\urldef\tempurl%
\url{https://doi.org/10.1007/978-3-642-22110-1\_47}
\showDOI{\tempurl}


\bibitem[Moosbrugger et~al\mbox{.}(2021)]%
        {DBLP:conf/esop/MoosbruggerBKK21}
\bibfield{author}{\bibinfo{person}{Marcel Moosbrugger}, \bibinfo{person}{Ezio
  Bartocci}, \bibinfo{person}{Joost{-}Pieter Katoen}, {and}
  \bibinfo{person}{Laura Kov{\'{a}}cs}.} \bibinfo{year}{2021}\natexlab{}.
\newblock \showarticletitle{Automated Termination Analysis of Polynomial
  Probabilistic Programs}. In \bibinfo{booktitle}{\emph{European Symposium on
  Programming (ESOP), Luxembourg City, Luxembourg}}
  \emph{(\bibinfo{series}{Lecture Notes in Computer Science},
  Vol.~\bibinfo{volume}{12648})}. \bibinfo{publisher}{Springer-Verlag},
  \bibinfo{pages}{491--518}.
\newblock
\urldef\tempurl%
\url{https://doi.org/10.1007/978-3-030-72019-3\_18}
\showDOI{\tempurl}


\bibitem[Morgan et~al\mbox{.}(1996)]%
        {DBLP:journals/toplas/MorganMS96}
\bibfield{author}{\bibinfo{person}{Carroll Morgan}, \bibinfo{person}{Annabelle
  McIver}, {and} \bibinfo{person}{Karen Seidel}.}
  \bibinfo{year}{1996}\natexlab{}.
\newblock \showarticletitle{Probabilistic Predicate Transformers}.
\newblock \bibinfo{journal}{\emph{ACM Transactions on Programming Languages and
  Systems}} \bibinfo{volume}{18}, \bibinfo{number}{3} (\bibinfo{year}{1996}),
  \bibinfo{pages}{325--353}.
\newblock
\urldef\tempurl%
\url{https://doi.org/10.1145/229542.229547}
\showDOI{\tempurl}


\bibitem[Sampson et~al\mbox{.}(2014)]%
        {sampson_2014}
\bibfield{author}{\bibinfo{person}{Adrian Sampson}, \bibinfo{person}{Pavel
  Panchekha}, \bibinfo{person}{Todd Mytkowicz}, \bibinfo{person}{Kathryn~S.
  McKinley}, \bibinfo{person}{Dan Grossman}, {and} \bibinfo{person}{Luis
  Ceze}.} \bibinfo{year}{2014}\natexlab{}.
\newblock \showarticletitle{Expressing and Verifying Probabilistic Assertions}.
  In \bibinfo{booktitle}{\emph{{ACM SIGPLAN Conference on Programming Language
  Design and Implementation (PLDI)}, Edinburgh, Scotland}}.
  \bibinfo{pages}{112--122}.
\newblock
\urldef\tempurl%
\url{https://doi.org/10.1145/2594291.2594294}
\showDOI{\tempurl}


\bibitem[Sankaranarayanan et~al\mbox{.}(2013)]%
        {sankaranarayanan_2013}
\bibfield{author}{\bibinfo{person}{Sriram Sankaranarayanan},
  \bibinfo{person}{Aleksandar Chakarov}, {and} \bibinfo{person}{Sumit
  Gulwani}.} \bibinfo{year}{2013}\natexlab{}.
\newblock \showarticletitle{Static Analysis for Probabilistic Programs:
  Inferring Whole Program Properties from Finitely Many Paths}. In
  \bibinfo{booktitle}{\emph{{ACM SIGPLAN Conference on Programming Language
  Design and Implementation (PLDI)}, Seattle, Washington}}.
  \bibinfo{pages}{447--458}.
\newblock
\urldef\tempurl%
\url{https://doi.org/10.1145/2499370.2462179}
\showDOI{\tempurl}


\bibitem[Sasnauskas et~al\mbox{.}(2011)]%
        {raimondas_2011}
\bibfield{author}{\bibinfo{person}{Raimondas Sasnauskas},
  \bibinfo{person}{Oscar~Soria Dustmann}, \bibinfo{person}{Benjamin~Lucien
  Kaminski}, \bibinfo{person}{Klaus Wehrle}, \bibinfo{person}{Carsten Weise},
  {and} \bibinfo{person}{Stefan Kowalewski}.} \bibinfo{year}{2011}\natexlab{}.
\newblock \showarticletitle{Scalable Symbolic Execution of Distributed
  Systems}. In \bibinfo{booktitle}{\emph{International Conference on
  Distributed Computing Systems (ICDCS), Minneapolis, Minnesota}}.
  \bibinfo{pages}{333--342}.
\newblock
\urldef\tempurl%
\url{https://doi.org/10.1109/ICDCS.2011.28}
\showDOI{\tempurl}


\bibitem[Sasnauskas et~al\mbox{.}(2010)]%
        {raimondas_2010}
\bibfield{author}{\bibinfo{person}{Raimondas Sasnauskas}, \bibinfo{person}{Olaf
  Landsiedel}, \bibinfo{person}{Muhammad~Hamad Alizai},
  \bibinfo{person}{Carsten Weise}, \bibinfo{person}{Stefan Kowalewski}, {and}
  \bibinfo{person}{Klaus Wehrle}.} \bibinfo{year}{2010}\natexlab{}.
\newblock \showarticletitle{{KleeNet}: Discovering Insidious Interaction Bugs
  in Wireless Sensor Networks before Deployment}. In
  \bibinfo{booktitle}{\emph{{ACM/IEEE} International Conference on Information
  Processing in Sensor Networks (IPSN), Stockholm, Sweden}}.
  \bibinfo{pages}{186--196}.
\newblock
\urldef\tempurl%
\url{https://doi.org/10.1145/1791212.1791235}
\showDOI{\tempurl}


\bibitem[Selvin(1975)]%
        {selvin1975}
\bibfield{author}{\bibinfo{person}{Steve Selvin}.}
  \bibinfo{year}{1975}\natexlab{}.
\newblock \showarticletitle{Letters to the Editor}.
\newblock \bibinfo{journal}{\emph{The American Statistician}}
  \bibinfo{volume}{29}, \bibinfo{number}{1} (\bibinfo{year}{1975}),
  \bibinfo{pages}{67--71}.
\newblock
\urldef\tempurl%
\url{https://doi.org/10.1080/00031305.1975.10479121}
\showDOI{\tempurl}
\showeprint{https://doi.org/10.1080/00031305.1975.10479121}


\bibitem[Smith et~al\mbox{.}(2019)]%
        {SHA18}
\bibfield{author}{\bibinfo{person}{Calvin Smith}, \bibinfo{person}{Justin Hsu},
  {and} \bibinfo{person}{Aws Albarghouthi}.} \bibinfo{year}{2019}\natexlab{}.
\newblock \showarticletitle{Trace Abstraction modulo Probability}.
\newblock \bibinfo{journal}{\emph{Proceedings of the {ACM} on Programming
  Languages}} \bibinfo{volume}{3}, \bibinfo{number}{POPL}, Article
  \bibinfo{articleno}{39} (\bibinfo{date}{Jan.} \bibinfo{year}{2019}).
\newblock
\urldef\tempurl%
\url{https://doi.org/10.1145/3290352}
\showDOI{\tempurl}
\showeprint[arxiv]{1810.12396}~[cs.PL]


\bibitem[Susag et~al\mbox{.}(2022)]%
        {plinkoArtifact2022}
\bibfield{author}{\bibinfo{person}{Zachary Susag}, \bibinfo{person}{Sumit
  Lahiri}, \bibinfo{person}{Justin Hsu}, {and} \bibinfo{person}{Subhajit Roy}.}
  \bibinfo{year}{2022}\natexlab{}.
\newblock \bibinfo{title}{{Artifact for Symbolic Execution for Randomized
  Programs}}.
\newblock
\newblock
\urldef\tempurl%
\url{https://doi.org/10.5281/zenodo.7061819}
\showDOI{\tempurl}


\bibitem[Vitter(1985)]%
        {vitter_1985}
\bibfield{author}{\bibinfo{person}{Jeffrey~Scott Vitter}.}
  \bibinfo{year}{1985}\natexlab{}.
\newblock \showarticletitle{Random Sampling with a Reservoir}.
\newblock \bibinfo{journal}{\emph{ACM Trans. Math. Software}}
  \bibinfo{volume}{11}, \bibinfo{number}{1} (\bibinfo{date}{March}
  \bibinfo{year}{1985}), \bibinfo{pages}{37--57}.
\newblock
\urldef\tempurl%
\url{https://doi.org/10.1145/3147.3165}
\showDOI{\tempurl}


\bibitem[Wang et~al\mbox{.}(2018)]%
        {DBLP:conf/pldi/WangHR18}
\bibfield{author}{\bibinfo{person}{Di Wang}, \bibinfo{person}{Jan Hoffmann},
  {and} \bibinfo{person}{Thomas~W. Reps}.} \bibinfo{year}{2018}\natexlab{}.
\newblock \showarticletitle{{PMAF:} An Algebraic Framework for Static Analysis
  of Probabilistic Programs}. In \bibinfo{booktitle}{\emph{{ACM SIGPLAN
  Conference on Programming Language Design and Implementation (PLDI)},
  Philadelphia, Pennsylvania}}. \bibinfo{pages}{513--528}.
\newblock
\urldef\tempurl%
\url{https://doi.org/10.1145/3192366.3192408}
\showDOI{\tempurl}


\bibitem[Wang et~al\mbox{.}(2021)]%
        {wang_2021}
\bibfield{author}{\bibinfo{person}{Di Wang}, \bibinfo{person}{Jan Hoffmann},
  {and} \bibinfo{person}{Thomas~W. Reps}.} \bibinfo{year}{2021}\natexlab{}.
\newblock \showarticletitle{Central Moment Analysis for Cost Accumulators in
  Probabilistic Programs}. In \bibinfo{booktitle}{\emph{{ACM SIGPLAN Conference
  on Programming Language Design and Implementation (PLDI)}}}.
  \bibinfo{pages}{559--573}.
\newblock
\urldef\tempurl%
\url{https://doi.org/10.1145/3453483.3454062}
\showDOI{\tempurl}
\showeprint[arxiv]{2001.10150}


\end{thebibliography}
\pagebreak

\iftoggle{SUPPLEMENTAL}{
%% Appendix
\appendix
\allowdisplaybreaks

\section{Algorithms for Case Studies}
\label{sec:algorithms}
In this section we provide pseudocode for each of the case studies that we describe in~\Cref{sec:case_studies}.

\label{sec:freivalds_alg}
\begin{algorithm}[H]
  \caption{Freivalds' Algorithm}
  \label{alg:freivalds}
  \begin{algorithmic}[1]
    \Require $A,B$, and $C$ are $n \times n$ matrices
    \Function{Freivalds}{$A,B,C,n$}
    \Let{$\vec{r}$}{An empty $n \times 1$ vector}
    \For{$i\gets1,n$}
    \State{$\vec{r}[i] \sim \mathsf{UniformInt}(0,1)$} 
    \EndFor
    \Let{$\vec{D}$}{$A \times (B\vec{r}) - C\vec{r}$}
    \State\Return{$\vec{D} = \begin{pmatrix} 0 & \cdots & 0 \end{pmatrix}^{\text{T}}$}
    \EndFunction
  \end{algorithmic}
\end{algorithm}
\begin{algorithm}[H]
  \caption{Freivalds' Algorithm (Multiple)}
  \label{alg:mult_freivalds}
  \begin{algorithmic}[1]
    \Function{MultFreivalds}{$A,B,C,n,k$}
    \For{$i\gets1,k$}
    \If{\Call{Freivalds}{$A,B,C,n$} $= \mathtt{false}$}
    \State\Return{\texttt{false}}
    \EndIf
    \EndFor
    \State\Return{\texttt{true}}
    \EndFunction
  \end{algorithmic}
\end{algorithm}
\label{sec:reservoir_sampling_alg}
\begin{algorithm}[H]
	\caption{Reservoir Sampling}
	\label{alg:reservoir_sampling}
	\begin{algorithmic}[1]
    \Require $1 \leq k \leq n$
		\Function{ReservoirSample}{$A[1..n], k$}
    \Let{$S$}{An empty list of size $k$}
		\For{$i\gets1,k$}
		\Let{$S[i]$}{$A[i]$}
		\EndFor
		\For{$i\gets k+1,n$}
		\Let{$j$}{$\mathsf{UniformInt}(1,i)$}
		\If{$j \leq k$}
		\Let{$S[j]$}{$A[i]$}
		\EndIf
		\EndFor
		\State\Return{$S$}
		\EndFunction
	\end{algorithmic}
\end{algorithm}
\begin{algorithm}[H]
  \caption{Randomized Monotonicity Testing~\citep{goldreich_2017}}
  \label{alg:monotone_testing}
  \begin{algorithmic}[1]
    \Function{MonotoneTest}{$f,n$}
    \State $l \gets \lceil \log_2 n \rceil, a \gets 1, b \gets n$ \label{line:monotone_l}
    \Let{$i$}{$\mathsf{UniformInt}(1,n)$}
    \For{$t\gets1,l$}
    \Let{$p$}{$\lceil a+b/2 \rceil$}
    \If{$i \leq p$}
    \If{$f[i] > f[p]$}
    \State\Return \texttt{false}
    \EndIf
    \Let{$b$}{$p$}
    \Else
    \If{$f[i] < f[p]$}

    \algstore{monotone}
  \end{algorithmic}
\end{algorithm}
\begin{algorithm}[H]
  \begin{algorithmic}[1]
    \algrestore{monotone}
    \State\Return \texttt{false}
    \EndIf
    \Let{$a$}{$p+1$}
    \EndIf
    \EndFor
    \State\Return \texttt{true}
    \EndFunction
  \end{algorithmic}
\end{algorithm}
\begin{algorithm}[h]
  \caption{Quicksort}
  \label{alg:quicksort}  
  \begin{algorithmic}[1]
    \Function{Quicksort}{$A,p,r$}
    \If{$p < r$} \label{line:quicksort_comp1}
    \Let{$q$}{\Call{Partition}{$A,p,r$}}
    \State\Call{Quicksort}{$A,p,q-1$}
    \State\Call{Quicksort}{$A,q+1,r$}
    \EndIf
    \EndFunction
    \Function{Partition}{$A,p,r$}
    \Let{$i$}{$\mathsf{UniformInt}(p,r)$}
    \State exchange $A[r]$ with $A[i]$
    \Let{$x$}{$A[r]$}
    \Let{$i$}{$p-1$}
    \For{$j\gets p, r-1$}
    \If{$A[j] \leq x$} \label{line:quicksort_comp2}
    \Let{$i$}{$i + 1$}
    \State exchange $A[i]$ with $A[j]$
    \EndIf
    \EndFor
    \State exchange $A[i+1]$ with $A[r]$
    \State\Return $i+1$
    \EndFunction
  \end{algorithmic}
\end{algorithm}

To properly analyze programs which employ hashing we consider ideal, \textit{uniform} hash functions.
Therefore, we model hash functions through uniform random samples.
\iftoggle{SUPPLEMENTAL}{
  Pseudocode for how we model hash functions is provided in the function \textsc{HashCreate} in~\Cref{alg:hash_function}, Line~\ref{line:hashfns}.
}{}
Essentially, to hash an unseen element, $x$, we randomly sample an integer within a given interval $[min,max]$.
We then record that sample in the map in order to maintain determinism.
To hash $x$ again in the future, we simply retrieve the sample from the map.
This ensures that each element is hashed uniformly and independently the first time, while subsequent hashes of an element always return the same result.

\begin{algorithm}[H]
  \caption{Uniform Hash Function}
  \label{alg:hash_function}
  \begin{algorithmic}[1]
    \Function{HashCreate}{$\mathit{min},\mathit{max}$} \label{line:hashfns}
    \Let{$H$}{An empty, uninitialized hash function}
    \Let{$H.Map$}{An empty map}
    \Let{$H.min$}{$\mathit{min}$}
    \Let{$H.max$}{$\mathit{max}$}
    \State\Return{$H$}
    \EndFunction
    \Function{Hash}{$H,x$}
    \If{$x \in H.Map$}
    \State\Return{$H.Map[x]$}
    \Else
    \State{$i \sim \mathsf{UniformInt}(H.min,H.max)$}
    \Let{$H.Map[x]$}{$i$}
    \State\Return{$i$}
    \EndIf
    \EndFunction
  \end{algorithmic}
\end{algorithm}

\begin{algorithm}[H]
  \caption{Bloom Filter}
  \label{alg:bloom_filter}
  \begin{algorithmic}[1]
    \Function{BloomCreate}{$m,\varepsilon$}
    \Let{$B$}{An empty, uninitialized Bloom Filter}
    \Let{$B.Arr$}{A bit-array of size $-\frac{m\ln \varepsilon}{(\ln 2)^2}$}
    \Let{$B.H$}{A list of $-\frac{\ln \varepsilon}{\ln 2}$ independent hash functions}
    \State\Return{$B$}
    \EndFunction
    \Statex
    \Function{BloomInsert}{$B,x$}
    \ForAll{$h \in B.H$}
    \Let{$B.Arr[h(x)]$}{$1$}
    \EndFor
    \State\Return{$B$}
    \EndFunction
    \Statex
    \Function{BloomCheck}{$B,x$}
    \ForAll{$h \in B.H$}
    \If{$\neg B.Arr[h(x)]$}
    \State\Return{\texttt{false}}
    \EndIf
    \EndFor
    \State\Return{\texttt{true}}
    \EndFunction
  \end{algorithmic}
\end{algorithm}

\begin{algorithm}[H]
  \caption{Count-min Sketch}
  \label{alg:countminsketch}
  \begin{algorithmic}[1]
    \Function{SketchCreate}{$\varepsilon,\gamma$}
    \Let{$w$}{$\lceil \frac{e}{\varepsilon} \rceil$}
    \Let{$d$}{$\lceil \ln \frac{1}{\gamma} \rceil$}
    \Let{$C$}{An empty, uninitialized count-min sketch}
    \Let{$C.\mathit{sketch}$}{An empty $w \times d$ array}
    \Let{$C.H$}{A list of $d$ independent hash functions}
    \State\Return $C$
    \EndFunction
    \Statex
    \Function{SketchUpdate}{$C,x$}
    \For{$j\gets1,d$}
    \Let{$C.\mathit{sketch}[j,C.H[j](x)]$}{$C.\mathit{sketch}[j,C.H[j](x)] + 1$}
    \EndFor
    \State\Return $C$
    \EndFunction
    \Statex
    \Function{SketchEstimate}{$C,x$}
    \Let{$\hat{a}_x$}{$\infty$}
    \For{$j\gets1,d$}
    \Let{$\hat{a}_x$}{$\min(\hat{a}_x,C.\mathit{sketch}[j,C.H[j](x)])$}
    \EndFor
    \State\Return $\hat{a}_x$
    \EndFunction
  \end{algorithmic}
\end{algorithm}

\section{Program Semantics}
\label{sec:semantics}
Here we define the semantics for the three main types of statements which concerns probabilistic symbolic execution: assignments, probabilistic sampling, and branching.

\begin{definition}[Assignment Semantics]
  \label{def:assignment}
  Let $\mu \in \mathit{Dists_M}$ be a distribution of program memories parameterized by assignments to universal symbolic variables. Let $\mathtt{x} \gets e$ be an arbitrary assignment statement where $\mathtt{x} \in \mathcal{X}$ is assigned to the program expression $e$. Let $\mathsf{assign}_{\mathtt{x}\gets e} : \mathit{M} \rightarrow \mathit{M}$ is defined as
  \[
    \mathsf{assign}_{\mathtt{x}\gets e}(m) = \lambda(\mathtt{y} : \mathcal{X})
    \begin{cases}
      \deno{e}m     & \text{if $\mathtt{x} = \mathtt{y}$} \\
      m(\mathtt{y}) & \text{otherwise}
    \end{cases}
  \]
  and let $\mathsf{unassign}_{\mathtt{x}\gets e} = \mathsf{assign}_{\mathtt{x}\gets e}^{-1}$. Define $\mu_{\mathtt{x}\gets e}$ to be the distribution over program memories parameterized by assignments to universal symbolic variables after executing the assignment statement $\mathtt{x}\gets e$ to be
  \[
    \mu_{\mathtt{x}\gets e}(a_{\forall},m) = \sum_{m' \in \mathsf{unassign}_{\mathtt{x}\gets e}(m)} \mu(a_{\forall},m').
  \]
\end{definition}

\begin{definition}[Sampling Semantics]
  \label{def:sampling}
  Let $\mu \in \mathit{Dists_M}$ be a distribution of program memories parameterized by assignments to universal symbolic variables. Let $\mathtt{x} \sim d$ be an arbitrary sampling instruction which assigns the program variable $\mathtt{x} \in \mathcal{X}$ to a random element from the distribution of values parameterized by a memory, represented as a distribution expression $d \in \mathit{Dists}$. Let $\mathsf{desample} : \mathcal{X} \times \mathit{M} \rightarrow \mathcal{P}(\mathit{M})$ be defined as
  \[
    \mathsf{desample}(\mathtt{x},m) = \{ m'\in \mathit{M} \mid \forall (\mathtt{y} \in \mathcal{X})~.~(\mathtt{y} \neq \mathtt{x} \wedge m'(\mathtt{y}) = m(\mathtt{y}))\}.
  \]
  Define $\mu_{\mathtt{x} \sim d}$ to be the distribution over program memories parameterized by assignments to universal symbolic variables after executing the sampling statement $\mathtt{x} \sim d$ to be
  \[
    \mu_{\mathtt{x} \sim d}(a_{\forall},m) = \sum_{m' \in \mathsf{desample}(m)} (\deno{d}a_{\forall})(m(\mathtt{x})) \cdot \mu(a_{\forall},m').
  \]
\end{definition}

\begin{definition}[Branching Semantics]
  \label{def:branching}
  Let $a_{\forall} \in \mathit{A}_{\forall}$ be an arbitrary assignment of universal symbolic variables, $c$ be a guard of an \textbf{if-goto} statement, $\mu \in \mathit{Dists_M}$ be a distribution of program memories parameterized by assignments to universal symbolic variables, and $\mathtt{x}_1,\ldots,\mathtt{x}_n \in \mathcal{X}$ be all the program variables in $c$. Then for all program memories $m \in \mathit{M}$, $\mu$ conditioned on a guard $c$ being true, represented as $\mu_c$ is defined as
  \[
    \mu_c(a_{\forall},m) = \frac{\displaystyle\Pr_{m' \sim \mu(a_{\forall})}\left[ m' = m \land \deno{c}m' = \mathtt{true} \right]}{\displaystyle\Pr_{m' \sim \mu(a_{\forall})}\left[ \deno{c}m' = \mathtt{true} \right]}
  \]
  Similarly, for all program memories $m \in \mathit{M}$, $\mu$ conditioned on a guard $c$ being false, represented as $\mu_{\neg c}$ is defined as
  \[
    \mu_{\neg c}(a_{\forall},m) = \frac{\displaystyle\Pr_{m' \sim \mu(a_{\forall})}\left[ m' = m \land \deno{c}m' = \mathtt{false} \right]}{\displaystyle\Pr_{m' \sim \mu(a_{\forall})}\left[ \deno{c}m' = \mathtt{false} \right]}
  \]
\end{definition}

\section{Formalism}
\label{sec:supplemental_formalism}

\subsection{Notation \& Definitions}
\label{sec:notation}

To begin, we will define some notation: 
\begin{itemize}
\item Let $\mathcal{V}$ be the set of all values, $\mathcal{X}$ be the set of all program variables, $\mathcal{Z}_U$ be the set of all universal symbolic variables, $\mathcal{Z}_P$ be the set of all probabilistic symbolic variables, and $\mathcal{Z} = \mathcal{Z}_U \cup \mathcal{Z}_P$ be the combined set of all symbolic variables.
\item Let $a_{\forall}: \mathcal{Z}_U \rightarrow \mathcal{V}$ be an assignment of universal symbolic variables to values and let $\mathit{A}_{\forall}$ be the set of all such assignments.
\item Similarly, let $a_p : \mathcal{Z}_P \rightarrow \mathcal{V}$ be an assignment of probabilistic symbolic variables to values and let $\mathit{A}_p$ be the set of all such assignments.
\item Let $m : \mathcal{X} \rightarrow \mathcal{V}$ be a program memory which translates program variables into values, and let $\mathit{M}$ be the set of all program memories.
\item Let $d : \mathit{M} \rightarrow (\mathcal{V} \rightarrow [0,1])$ be a distribution expression parameterized by program memories, and let $\mathit{Dists}$ be the set of all distribution expressions.
\item Let $\mu : \mathit{A}_{\forall} \times \mathit{M} \rightarrow [0,1]$ be a distribution of program memories parameterized by assignments to universal symbolic variables and let $\mathit{Dists_M}$ be the set of all parameterized distributions of program memories.
\end{itemize}

Additionally, we will use emphatic brackets ($\deno{\cdot}$) for two purposes:
\begin{itemize}
\item If $e$ is a \textit{program} expression containing the program variables $\mathtt{x}_1,\ldots,\mathtt{x}_n \in \mathcal{X}$, and $m \in \mathit{M}$, then
  \[
    \deno{e}m = \mathsf{eval}(e[\mathtt{x}_1 \mapsto m(\mathtt{x}_1),\ldots,\mathtt{x}_n \mapsto m(\mathtt{x}_n)])
  \]
\item If $e$ is a \textit{symbolic} expression containing the symbolic variables $\alpha_1,\ldots,\alpha_n \in \mathcal{Z}_U$ and $\delta_1,\ldots,\delta_m \in \mathcal{Z}_P$, and $a_{\forall} \in \mathit{A}_{\forall}$ and $a_p \in \mathit{A}_p$, then $\deno{e}a_{\forall}a_p = \mathsf{eval}(e[\alpha_1\mapsto a_{\forall}(\alpha_1),\ldots,\alpha_n \mapsto a_{\forall}(\alpha_n), \delta_1 \mapsto a_p(\delta_1), \ldots, \delta_m \mapsto a_p(\delta_m)])$.
\end{itemize}

With this notation in hand, we can now define what it means for $R$ to be an abstraction of a distribution of programs memories.
\begin{definition}
  \label{def:main}
  Let $R = (\sigma,\varphi,P)$ be the abstraction generated by~\Cref{alg:symb_ex} where $\varphi: \mathit{A}_{\forall} \times \mathit{A}_p \rightarrow \{0,1\}$ denotes whether a path condition is true or false under the given assignments, $\sigma : \mathcal{X} \rightarrow SymExprs$ is a mapping from program variables to symbolic expressions generated through symbolic execution, and $P : \mathit{A}_{\forall} \rightarrow \mathcal{Z}_P \rightharpoonup (\mathcal{V} \rightarrow [0,1])$ is a mapping from probabilistic symbolic variables to the distribution it is sampled from, parameterized by assignments to universal symbolic variables.
  Additionally, for every assignment of universal symbolic variables, $a_{\forall} \in \mathit{A}_{\forall}$, $\dom{P(a_{\forall})} = \{\delta_1,\ldots,\delta_k\}$.
  Let $\alpha_1,\ldots,\alpha_l \in \mathcal{Z}_U$ be the universal symbolic variables which correspond to the $l$ parameters to the program.
  For every assignment of probabilistic and universal symbolic variables, $a_{\forall} \in \mathit{A}_{\forall}, a_p \in \mathit{A}_p$, let $\nu : \mathit{A}_{\forall} \rightarrow (\mathit{A}_p \rightarrow [0,1])$ be a distribution of assignments to probabilistic symbolic variables parameterized by assignments to universal symbolic variables, defined as
  \[
    \nu(a_{\forall},a_p) \triangleq \prod_{i=1}^k \Pr_{v \sim P(a_{\forall},\delta_i)}[v = a_p(\delta_i)].
  \]
  We say that a distribution $\mu$ satisfies our abstraction $R$ if, for all assignments to universal symbolic variables, $a_{\forall} \in \mathit{A}_{\forall}$, $\Pr_{a_p' \sim \nu(a_{\forall})}[\varphi(a_{\forall},a_p') = 1] > 0$, and if $\nu_{\varphi} : \mathit{A}_{\forall} \rightarrow (\mathit{A}_p \rightarrow [0,1])$ is defined as
  \[
    \nu_\varphi(a_{\forall},a_p) \triangleq \cfrac{\displaystyle\Pr_{a_p' \sim \nu(a_{\forall})}[a_p' = a_p \land \varphi(a_{\forall},a_p') = 1]}{\displaystyle\Pr_{a_p' \sim \nu(a_{\forall})}[\varphi(a_{\forall},a_p')=1]}.
  \]
  Additionally, define $\mathsf{toMem} : (\mathcal{X} \rightarrow SymExprs) \rightarrow \mathit{A}_{\forall} \rightarrow \mathit{A}_p \rightarrow M$ as
  \begin{equation*}
    \mathsf{toMem}(\sigma,a_{\forall},a_p) \triangleq \lambda (\mathtt{x}: \mathcal{X})~.~\deno{\sigma(\mathtt{x})}a_{\forall}a_p,
  \end{equation*}
  and let $\mathsf{fromMem}(\sigma,a_{\forall},m) = (\mathsf{toMem}(\sigma,a_{\forall}))^{-1}(m)$.
  Then,
  \[
    \mu(a_{\forall},m) \triangleq \sum_{a_p \in \mathsf{fromMem}(\sigma,a_{\forall},m)} \nu_{\varphi}(a_{\forall},a_p).
  \]
\end{definition}

\subsection{Omitted Proofs}
\label{sec:omitted_proofs}
\equivpthm*
\begin{proof}
  Let $S$ be an arbitrary \textbf{pWhile} program statement, $\varphi$ be the current path condition, $P$ be the current distribution map, and $p$ be the current path probability before executing $S$.
  Assume that
  \[
    p = \frac{\displaystyle\sum_{(v_1,\ldots,v_n) \in \mathcal{D}} [ \varphi\{\delta_1 \mapsto v_1, \ldots, \delta_n \mapsto v_n\} ]}{\abs{\mathcal{D}}}.
  \]
  We will proceed with a case analysis on $S$:
  \begin{itemize}
  \item \textit{Assignment (\Crefrange{line:beg_raw_assign}{line:end_raw_assign})}. Suppose $S$ is an assignment statement of the form $\mathtt{x} \gets e$, where \texttt{x} is an arbitrary program variable and $e$ is an arbitrary program expression.
    Note that $\varphi$, $P$, and $p$ do not change after executing $\mathtt{x} \gets e$ and so the statement trivially holds.
  \item \textit{Sampling (\Crefrange{line:pse_sym_sample}{line:end_raw_sample})}. Suppose $S$ is a sampling statement of the form $\mathtt{x} \sim d$, where \texttt{x} is an arbitrary program variable and $d$ is an arbitrary distribution expression.
    Additionally assume that there are $n$ probabilistic symbolic variables in $P$, namely $\delta_1,\ldots,\delta_n$, and $\delta_{n+1}$ is the fresh probabilistic symbolic variable created on \cref{line:fresh_psv} in \Cref{alg:symb_ex}.
    Then, $P_1 = P \cup \{\delta_{n+1} \mapsto d\}$, according to \cref{line:pse_sym_sample} where $P_1$ is the updated distribution map.
    If $\mathcal{D}_1 = \dom{P_1[\delta_1]} \times \cdots \times \dom{P_1[\delta_{n+1}]}$, then
    \begin{align*}
      &\frac{\displaystyle\sum_{(v_1,\ldots,v_{n+1}) \in \mathcal{D}_1} [ \varphi\{\delta_1 \mapsto v_1, \ldots, \delta_{n+1} \mapsto v_{n+1}\}]}{\abs{\mathcal{D}_1}} \\ &\qquad\qquad\qquad\qquad\qquad= \frac{\abs{\dom{P_1[\delta_{n+1}]}} \cdot \displaystyle\sum_{(v_1,\ldots,v_{n}) \in \mathcal{D}} [ \varphi\{\delta_1 \mapsto v_1, \ldots, \delta_{n} \mapsto v_{n}\}]}{\abs{\mathcal{D}_1}} \\
                                                                                                                         &\qquad\qquad\qquad\qquad\qquad= \frac{\abs{\dom{P_1[\delta_{n+1}]}} \cdot \displaystyle\sum_{(v_1,\ldots,v_{n}) \in \mathcal{D}} [ \varphi\{\delta_1 \mapsto v_1, \ldots, \delta_{n} \mapsto v_{n}\}]}{\abs{\dom{P_1[\delta_{n+1}]}} \cdot \abs{\mathcal{D}}} \\
                                                                                                                         &\qquad\qquad\qquad\qquad\qquad= \frac{\displaystyle\sum_{(v_1,\ldots,v_{n}) \in \mathcal{D}} [ \varphi\{\delta_1 \mapsto v_1, \ldots, \delta_{n} \mapsto v_{n}\}]}{\abs{\mathcal{D}}} \\
                                                                                                                         &\qquad\qquad\qquad\qquad\qquad= p
    \end{align*}
    where the first line follows from the fact that $\delta_{n+1}$ cannot appear in $\varphi$ as $\delta_{n+1}$ is a \textit{fresh} probabilistic symbolic variable.
  \item \textit{Branching (\Cref{line:pse_sym_branch,line:symbex_true_state,line:symbex_false_state})}. Suppose $S$ is a branching statement of the form $\mathbf{if}~c~\mathbf{then~goto}~T$ where $c$ is an arbitrary program guard.
    Additionally assume that there are $n$ probabilistic symbolic variables in $P$, namely $\delta_1,\ldots,\delta_n$.
    Then, according to \cref{line:branch_prob_compute} of \Cref{alg:branch},
    \[
      p_c = \frac{\displaystyle\sum_{(v_1,\ldots,v_n) \in \mathcal{D}} [(c_{sym} \wedge \varphi)\{\delta_1 \mapsto v_1,\ldots,\delta_n \mapsto v_n\}]}{\displaystyle\sum_{(v_1,\ldots,v_n) \in \mathcal{D}}[\varphi\{\delta_1 \mapsto v_1,\ldots,\delta_n \mapsto v_n\}]}
    \]
    and so $p_t = p_c$ and $p_f = 1-p_c$.
    If $p'_t = p \cdot p_t$ and $p'_f = p \cdot p_f$, as on lines \cref{line:symbex_true_state,line:symbex_false_state}, and $\varphi'_t = \varphi \wedge c_{sym}$ and $\varphi'_f = \varphi \wedge \neg c_{sym}$, then we claim that
    \[
      p'_t = \frac{\displaystyle\sum_{(v_1,\ldots,v_n) \in \mathcal{D}} [ \varphi'_t\{\delta_1 \mapsto v_1, \ldots, \delta_n \mapsto v_n\} ]}{\abs{\mathcal{D}}}
    \]
    and similarly
    \[
      p'_f = \frac{\displaystyle\sum_{(v_1,\ldots,v_n) \in \mathcal{D}} [ \varphi'_f\{\delta_1 \mapsto v_1, \ldots, \delta_n \mapsto v_n\} ]}{\abs{\mathcal{D}}}.
    \]

    We will prove the claim for $p'_t$ as the proof for $p'_f$ is symmetric.
    Note that,
    \begin{align*}
      p'_t &= p \cdot p_t \\
           &= \frac{\displaystyle\sum_{(v_1,\ldots,v_n) \in \mathcal{D}} [ \varphi\{\delta_1 \mapsto v_1, \ldots, \delta_n \mapsto v_n\} ]}{\abs{\mathcal{D}}} \cdot \frac{\displaystyle\sum_{(v_1,\ldots,v_n) \in \mathcal{D}} [(c_{sym} \wedge \varphi)\{\delta_1 \mapsto v_1,\ldots,\delta_n \mapsto v_n\}]}{\displaystyle\sum_{(v_1,\ldots,v_n) \in \mathcal{D}}[\varphi\{\delta_1 \mapsto v_1,\ldots,\delta_n \mapsto v_n\}]} \\
           &= \frac{\displaystyle\sum_{(v_1,\ldots,v_n) \in \mathcal{D}} [(c_{sym} \wedge \varphi)\{\delta_1 \mapsto v_1,\ldots,\delta_n \mapsto v_n\}]}{\abs{\mathcal{D}}}\\
           &= \frac{\displaystyle\sum_{(v_1,\ldots,v_n) \in \mathcal{D}} [\varphi'_t\{\delta_1 \mapsto v_1,\ldots,\delta_n \mapsto v_n\}]}{\abs{\mathcal{D}}}
    \end{align*}
  \end{itemize}
  Since no state changes if $S$ is a \textbf{halt} statement, the theorem is trivially true for this case.
\end{proof}

\soundnessAssignment*
\begin{proof}
  First assume that $\mu$ satisfies the input abstraction $R=(\varphi,\sigma,P)$. Note that by~\Cref{def:assignment},
  \[
    \mu_{\mathtt{x}\gets e}(a_{\forall},m) = \displaystyle\sum_{m' \in \mathsf{unassign}_{\mathtt{x}\gets e}(m)} \mu(a_{\forall},m')
  \]
  for all $a_{\forall} \in \mathit{A}_{\forall}$ and $m \in \mathit{M}$. Also note that by~\Cref{def:main},
  \[
    \mu(a_{\forall},m) = \displaystyle\sum_{a_p \in \mathsf{fromMem}(\sigma,a_{\forall},m)} \nu_{\varphi}(a_{\forall},a_p).
  \]
  So,
  \begin{align*}
    \mu_{\mathtt{x}\gets e}(a_{\forall},m) &= \sum_{m' \in \mathsf{unassign}_{\mathtt{x}\gets e}(m)} \mu(a_{\forall},m')\\
                                           &= \sum_{m' \in \mathsf{unassign}_{\mathtt{x}\gets e}(m)} \sum_{a_p \in \mathsf{fromMem}(\sigma,a_{\forall},m')} \nu_{\varphi}(a_{\forall},a_p).
  \end{align*}

  Thus, in order to satisfy~\Cref{def:main} it suffices to show that
  \[
    \sum_{a_p \in \mathsf{fromMem}(\sigma',a_{\forall},m)} \nu_{\varphi}(a_{\forall},a_p) = \sum_{m' \in \mathsf{unassign}_{\mathtt{x}\gets e}(m)} \sum_{a_p \in \mathsf{fromMem}(\sigma,a_{\forall},m')} \nu_{\varphi}(a_{\forall},a_p).
  \]

  Note that if $m_1,m_2 \in \mathit{M}$ where $m_1 \neq m_2$, then for all $a_{\forall} \in \mathit{A}_{\forall}$ and $\sigma$, $\mathsf{fromMem}(\sigma,a_{\forall},m_1)$ and $\mathsf{fromMem}(\sigma,a_{\forall},m_2)$ are disjoint. Hence, it suffices to show that
  \[
    \mathsf{fromMem}(\sigma',a_{\forall},m) = \bigcup_{m' \in \mathsf{unassign}_{\mathtt{x}\gets e}(m)} \mathsf{fromMem}(\sigma,a_{\forall},m').
  \]
  to prove the theorem. We will show this through double containment.

  First, let $a_p \in \bigcup_{m' \in \mathsf{unassign}_{\mathtt{x}\gets e}(m)} \mathsf{fromMem}(\sigma,a_{\forall},m')$ be arbitrary. This means that for some $m' \in \mathsf{unassign}_{\mathtt{x}\gets e}(m)$, $m' = \lambda (\mathtt{y} : \mathcal{X}) . \deno{\sigma(\mathtt{y})}a_{\forall}a_p$. Furthermore, since $m' \in \mathsf{unassign}_{\mathtt{x}\gets e}(m)$,
  \[
    m = \lambda (\mathtt{y} : \mathcal{X})
    \begin{cases}
      \deno{e}m' & \text{if $\mathtt{x}=\mathtt{y}$}\\
      m'(\mathtt{y}) & \text{otherwise}
    \end{cases}
  \]
  Note that we need to show that $a_p \in \mathsf{fromMem}(\sigma',a_{\forall},m)$. This is equivalent to showing that $m =\mathsf{toMem}(\sigma',a_{\forall},a_p)$. Thus, it is sufficient to prove that
  \[
    \lambda (\mathtt{y} : \mathcal{X}) . \deno{\sigma'(\mathtt{y})}a_{\forall}a_p = \lambda (\mathtt{y} : \mathcal{X})
    \begin{cases}
      \deno{e}m' & \text{if $\mathtt{x}=\mathtt{y}$}\\
      m'(\mathtt{y}) & \text{otherwise}
    \end{cases}.
  \]
  Note that the domains of each of these functions are equal, namely the set of all program variables present within the program. Now we need to show that for every $\mathtt{z} \in \mathcal{X}$, the two functions compute the same result. We have two cases:
  \begin{enumerate}
  \item \textbf{Case 1:} Assume that $\mathtt{z} = \mathtt{x}$. Then we have to show that $\deno{e}m' = \deno{\sigma'(\mathtt{x})}a_{\forall}a_p$. However, note that by~\Crefrange{line:beg_raw_assign}{line:end_raw_assign}, $\sigma'(\mathtt{x}) = e_{sym}$. Recall that $m' = \lambda (\mathtt{y} : \mathcal{X}) . \deno{\sigma(x)}a_{\forall}a_p$. If $\mathtt{x}_1,\ldots,\mathtt{x}_k \in \mathcal{X}$ are all the program variables present, then
    \begin{align*}
      \deno{e}m' &= \mathsf{eval}(e[\mathtt{x}_1 \mapsto \deno{\sigma(\mathtt{x}_1)}a_{\forall}a_p, \ldots, \mathtt{x}_k \mapsto \deno{\sigma(\mathtt{x}_k)}a_{\forall}a_p])
    \end{align*}
    However, note that $e_{sym} = e[\mathtt{x}_1 \mapsto \sigma(\mathtt{x}_1),\ldots,\mathtt{x}_k \mapsto \sigma(\mathtt{x}_k)]$ and so,
    \begin{align*}
      \deno{e_{sym}}a_{\forall}a_p &= \deno{e[\mathtt{x}_1 \mapsto \sigma(\mathtt{x}_1),\ldots,\mathtt{x}_k \mapsto \sigma(\mathtt{x}_k)]}a_{\forall}a_p\\
                             &= \mathsf{eval}(e[\mathtt{x}_1 \mapsto \deno{\sigma(\mathtt{x}_1)}a_{\forall}a_p, \ldots, \mathtt{x}_k \mapsto \deno{\sigma(\mathtt{x}_k)}a_{\forall}a_p]).
    \end{align*}
    Hence, $\deno{e}m' = \deno{\sigma'(\mathtt{x})}a_{\forall}a_p$.
  \item \textbf{Case 2:} Now assume that $\mathtt{z} \neq \mathtt{x}$. We claim that $m'(\mathtt{z}) = \deno{\sigma'(\mathtt{z})}a_{\forall}a_p$. Note that $m'(\mathtt{z}) = \deno{\sigma(\mathtt{z})}a_{\forall}a_p$. Since $\mathtt{z}\neq\mathtt{x}$, then $\sigma(\mathtt{z}) = \sigma'(\mathtt{z})$ and we are done.
  \end{enumerate}
  Therefore,
  \[
    \lambda (\mathtt{y} : \mathcal{X}) . \deno{\sigma'(\mathtt{y})}a_{\forall}a_p = \lambda (\mathtt{y} : \mathcal{X})
    \begin{cases}
      \deno{e}m' & \text{if $\mathtt{x}=\mathtt{y}$}\\
      m'(\mathtt{y}) & \text{otherwise}
    \end{cases}
  \] and so $a_p \in \mathsf{fromMem}(\sigma',a_{\forall},m)$.
  
  For the forward direction of the double containment, assume that $a_p \in \mathsf{fromMem}(\sigma',a_{\forall},m)$ be arbitrary. This implies that $m = \mathsf{toMem}(\sigma',a_{\forall},a_p) = \lambda (\mathtt{y} : \mathcal{X}) . \deno{\sigma'(\mathtt{y})}a_{\forall}a_p$. We need to show that $a_p \in  \bigcup_{m' \in \mathsf{unassign}_{\mathtt{x}\gets e}(m)} \mathsf{fromMem}(\sigma,a_{\forall},m')$ which is equivalent to showing that for some $m'$, $a_p = \mathsf{fromMem}(\sigma,a_{\forall},m')$. We claim that $m' = \lambda (\mathtt{y} : \mathcal{X}) . \deno{\sigma(\mathtt{y})}a_{\forall}a_p \in \mathsf{unassign}_{\mathtt{x}\gets e}(m)$. We can prove this by showing that $\mathsf{assign}_{\mathtt{x}\gets e}(m') = m$. Note that
  \[
    \mathsf{assign}_{\mathtt{x}\gets e}(m') = \lambda (\mathtt{y} : \mathcal{X}) .
    \begin{cases}
      \deno{e}m' & \text{if $\mathtt{x} = \mathtt{y}$}\\
      m'(\mathtt{y}) & \text{otherwise}
    \end{cases}
  \]
  Both functions have the same domain, namely the program variables present in the program. We will show that both functions give the same output on the arbitrary input $\mathtt{z} \in \mathcal{X}$. There are two cases:
  \begin{enumerate}
  \item \textbf{Case 1:} Assume that $\mathtt{z} = \mathtt{x}$. Then, we need to show $\deno{e}m' = \deno{\sigma'(\mathtt{y})}a_{\forall}a_p$.
    However, note that by~\Crefrange{line:beg_raw_assign}{line:end_raw_assign}, $\sigma'(\mathtt{x}) = e_{sym}$. Recall that $m' = \lambda (\mathtt{y} : \mathcal{X}) . \deno{\sigma(x)}a_{\forall}a_p$. If $\mathtt{x}_1,\ldots,\mathtt{x}_k \in \mathcal{X}$ are all the program variables present, then $\deno{e}m' = \mathsf{eval}(e[\mathtt{x}_1 \mapsto \deno{\sigma(\mathtt{x}_1)}a_{\forall}a_p, \ldots, \mathtt{x}_k \mapsto \deno{\sigma(\mathtt{x}_k)}a_{\forall}a_p])$.
    However, note that $e_{sym} = e[\mathtt{x}_1 \mapsto \sigma(\mathtt{x}_1),\ldots,\mathtt{x}_k \mapsto \sigma(\mathtt{x}_k)]$.
    and so,
    \begin{align*}
      \deno{e_{sym}}a_{\forall}a_p &= \deno{e[\mathtt{x}_1 \mapsto \sigma(\mathtt{x}_1),\ldots,\mathtt{x}_k \mapsto \sigma(\mathtt{x}_k)]}a_{\forall}a_p\\
                             &= \mathsf{eval}(e[\mathtt{x}_1 \mapsto \deno{\sigma(\mathtt{x}_1)}a_{\forall}a_p, \ldots, \mathtt{x}_k \mapsto \deno{\sigma(\mathtt{x}_k)}a_{\forall}a_p]).
    \end{align*}
    Hence, $\deno{e}m' = \deno{\sigma'(\mathtt{x})}a_{\forall}a_p$.
  \item \textbf{Case 2:} Now assume that $\mathtt{z} \neq \mathtt{x}$. Then we have to show that $m'(\mathtt{z}) = \deno{\sigma'(\mathtt{z})}a_{\forall}a_p$. Since $m'(\mathtt{z}) = \deno{\sigma(\mathtt{z})}a_{\forall}a_p$ and $\mathtt{z}\neq\mathtt{x}$, then $\sigma(\mathtt{z}) = \sigma'(\mathtt{z})$ and we are done.
  \end{enumerate}
  Therefore, $m' \in \mathsf{unassign}_{\mathtt{x}\gets e}(m)$. Since $m' = \mathsf{toMem}(\sigma,a_{\forall},a_p)$, we have $a_p \in \mathsf{fromMem}(\sigma,a_{\forall},m')$ which implies that $a_p \in  \bigcup_{m' \in \mathsf{unassign}_{\mathtt{x}\gets e}(m)} \mathsf{fromMem}(\sigma,a_{\forall},m')$. By double containment,
  \[
    \mathsf{fromMem}(\sigma',a_{\forall},m) = \bigcup_{m' \in \mathsf{unassign}_{\mathtt{x}\gets e}(m)} \mathsf{fromMem}(\sigma,a_{\forall},m')
  \]
  and so
  \[
    \mu_{\mathtt{x}\gets e}(a_{\forall},m) = \sum_{a_p \in \mathsf{fromMem}(\sigma',a_{\forall},m)} \nu_{\varphi}(a_{\forall},a_p).
  \]
  Since $\nu_{\varphi}$ is unchanged, $\mu_{\mathtt{x}\gets e}$ satisfies $(\varphi,\sigma',P)$.
\end{proof}

\soundnessSampling*
\begin{proof}
  First assume that $\mu$ satisfies the input abstraction $R = (\varphi,\sigma,P)$. We want to show that $\mu_{\mathtt{x} \sim d}$ satisfies the output abstraction $(\varphi,\sigma',P')$. To do this, we need to define $\nu'$ and $\nu_{\varphi}'$ in terms of $\sigma'$ and $P'$ and show that $\nu_{\varphi}$ satisfies~\Cref{def:main}. Let $\delta_{k+1}$ be the fresh probabilistic symbolic variable allocated on~\cref{line:fresh_psv} of~\Cref{alg:symb_ex}.

  Note according to~\Cref{def:main},
  \[
    \nu'(a_{\forall},a_p) = \prod_{i=1}^{k+1} \Pr_{v \sim P'(a_{\forall},\delta_i)}[v = a_p(\delta_i)].
  \]
  
  However, we want
  \[
    \nu(a_{\forall},a_p) = \prod_{i=1}^{k} \Pr_{v \sim P(a_{\forall},\delta_i)}[v = a_p(\delta_i)]
  \]
  by definition. Since for all $1 \leq i \leq k$, $P(\delta_i) = P'(\delta_i)$,
  \begin{align*}
    \nu'(a_{\forall},a_p) &= \prod_{i=1}^{k+1} \Pr_{v \sim P'(a_{\forall},\delta_i)}[v = a_p(\delta_i)]\\
                  &= \left( \prod_{i=1}^{k} \Pr_{v \sim P'(a_{\forall},\delta_i)}[v = a_p(\delta_i)] \right) \cdot \Pr_{v \sim P'(a_{\forall},\delta_{k+1})}[v = a_p(\delta_{k+1})] \\
                  &= \left( \prod_{i=1}^{k} \Pr_{v \sim P(a_{\forall},\delta_i)}[v = a_p(\delta_i)] \right) \cdot \Pr_{v \sim P'(a_{\forall},\delta_{k+1})}[v = a_p(\delta_{k+1})] 
  \end{align*}
  Note that the domain of $\nu$ is $\{\delta_1,\ldots,\delta_k\}$. So, if we let $a_p = a_p' \cup \{\delta_{k+1} \mapsto v\}$ where $v \in \mathcal{V}$ is arbitrary,
  \[
    \left( \prod_{i=1}^{k} \Pr_{v \sim P(a_{\forall},\delta_i)}[v = a_p(\delta_i)] \right) \cdot \Pr_{v \sim P'(a_{\forall},\delta_{k+1})}[v = a_p(\delta_{k+1})] = \nu(a_{\forall},a_p') \cdot \Pr_{v' \sim P'(a_{\forall},\delta_{k+1})}[v' = v]
  \]
  By definition
  \[
    \Pr_{v' \sim P'(a_{\forall},\delta_{k+1})}[v' = v] = \sum_{v' \in \{v' \in \mathcal{V} \mid v' = v\}} P'(a_{\forall},\delta_{k+1})(v') = P'(a_{\forall},\delta_{k+1})(v).
  \]
  According to \cref{line:pse_sym_sample}, $P'(a_{\forall},\delta_{k+1}) = \deno{d} a_{\forall}$, which is a distribution of values. Thus,
  \[
    \nu'(a_{\forall},a_p) = (a_p' \cup \{\delta_{k+1} \mapsto v\}) = \nu(a_{\forall},a_p') \cdot (\deno{d} a_{\forall})(v).
  \]

  Now, we can define $\nu_{\varphi}'$ as
  \[
    \nu_{\varphi}' = \frac{\Pr_{a_p' \sim \nu'(a_{\forall})}[a_p' = a_p \wedge \varphi(a_{\forall},a_p') = 1]}{\Pr_{a_p' \sim \nu'(a_{\forall})}[\varphi(a_{\forall},a_p') = 1]}.
  \]

  Before moving onto the main result, we will prove a lemma.
  \begin{lemma}
    \label{lem:split}
    For all $a_{\forall} \in \mathit{A}_{\forall}, a_p \in \mathit{A}_p$ where the domain of $a_p$ is $\{\delta_1,\ldots,\delta_k\}$ and $\delta_{k+1} \in \mathcal{Z}_P$ is a newly allocated probabilistic symbolic variable created on~\Cref{line:fresh_psv}. Also let $v \in \mathcal{V}$ be arbitrary. Then,
    \[
      \nu_{\varphi}'(a_{\forall},a_p \cup \{\delta_{k+1} \mapsto v\}) = \nu_{\varphi}(a_{\forall},a_p) \cdot (\deno{d} a_{\forall})(v).
    \]
  \end{lemma}
  \begin{proof}
    Let $a = a_p \cup \{\delta_{k+1} \mapsto v\}$. By definition,
    \begin{align*}
      \nu_{\varphi}'(a_{\forall}, a) = \frac{\Pr_{a_p' \sim \nu'(a_{\forall})}[a_p' = a \wedge \varphi(a_{\forall},a_p') = 1]}{\Pr_{a_p' \sim \nu'(a_{\forall})}[\varphi(a_{\forall},a_p') = 1]}
      &= \frac{\displaystyle \smashoperator[r]{\sum_{a_p' \in \{a_p' \in \mathit{A}_p \mid a_p' = a \wedge \varphi(a_{\forall},a_p') = 1\}}} \nu'_{\mathtt{x}\sim d}(a_{\forall},a_p')}{\displaystyle \smashoperator[r]{\sum_{a_p' \in \{a_p' \in \mathit{A}_p \mid \varphi(a_{\forall},a_p') = 1\}}} \nu'_{\mathtt{x}\sim d}(a_{\forall},a_p')}\\
      &= \frac{\displaystyle \smashoperator[r]{\sum_{a_p' \in \{a_p' \in \mathit{A}_p \mid a_p' = a \wedge \varphi(a_{\forall},a_p') = 1\}}} \nu(a_{\forall},a_p)\cdot(\deno{d}a_{\forall})(v)}{\displaystyle \smashoperator[r]{\sum_{a_p' \cup \{\delta_{k+1} \mapsto v'\} \in \{a_p' \in \mathit{A}_p \mid \varphi(a_{\forall},a_p') = 1\}}} \nu(a_{\forall},a_p')\cdot(\deno{d}a_{\forall})(v')}
    \end{align*}
    Note that the set being summed over in the numerator is $\{a_p' \in \mathit{A}_p \mid a_p' = a \wedge \varphi(a_{\forall},a_p') = 1\} = \{a\}$ as $\varphi$ does not reference $\delta_{k+1}$ so, as long as $\varphi(a_{\forall},a_p)=1$, $\varphi(a_{\forall},a)=1$. Therefore, we can remove the summation to get:
    \[
      \frac{\displaystyle \smashoperator[r]{\sum_{a_p' \in \{a_p' \in \mathit{A}_p \mid a_p' = a \wedge \varphi(a_{\forall},a_p') = 1\}}} \nu(a_{\forall},a_p)\cdot(\deno{d}a_{\forall})(v)}{\displaystyle \smashoperator[r]{\sum_{a_p' \cup \{\delta_{k+1} \mapsto v'\} \in \{a_p' \in \mathit{A}_p \mid \varphi(a_{\forall},a_p') = 1\}}}\nu(a_{\forall},a_p')\cdot(\deno{d}a_{\forall})(v')} = \frac{\nu(a_{\forall},a_p)\cdot(\deno{d}a_{\forall})(v)}{\displaystyle \smashoperator[r]{\sum_{a_p' \cup \{\delta_{k+1} \mapsto v'\} \in \{a_p'' \in \mathit{A}_p \mid \varphi(a_{\forall},a_p'') = 1\}}}\nu(a_{\forall},a_p')\cdot(\deno{d}a_{\forall})(v')}.
    \]
    Now, $\{a_p' \in \mathit{A}_p \mid \varphi(a_{\forall},a_p') = 1\}$ is the set of all probabilistic assignments which satisfy the path condition $\varphi$ with the assignment to universal symbolic variables, $a_{\forall}$. Since $\delta_{k+1}$ does not appear in $\varphi$, whatever value $a_p'$ assigns to $\delta_{k+1}$ will not affect whether $\varphi$ is satisfied. Therefore, for each value $v'$ in the domain of $\deno{d}a_{\forall}$, $a_p' \cup \{\delta_{k+1} = v'\} \in \{a_p'' \in \mathit{A}_p \mid \varphi(a_{\forall},a_p'') = 1\}$, where $a_p' \in \mathit{A}_p$ is some probabilistic assignment such that $\varphi(a_{\forall},a_p')=1$. Hence,
    \[
      \smashoperator[r]{\sum_{a_p' \cup \{\delta_{k+1} \mapsto v'\} \in \{a_p'' \in \mathit{A}_p \mid \varphi(a_{\forall},a_p'') = 1\}}} \nu(a_{\forall},a_p')\cdot(\deno{d}a_{\forall})(v')= \smashoperator[r]{\sum_{a_p' \cup \{\delta_{k+1} \mapsto v'\} \in \{a_p'' \in \mathit{A}_p \mid \varphi(a_{\forall},a_p'') = 1\}}} \nu(a_{\forall},a_p').
    \]
    Combining this all together, we have
    \begin{align*}
      \nu_{\varphi}'(a_{\forall}, a) &= \frac{\nu(a_{\forall},a_p)\cdot(\deno{d}a_{\forall})(v)}{\displaystyle \smashoperator[r]{\sum_{a_p' \cup \{\delta_{k+1} \mapsto v'\} \in \{a_p'' \in \mathit{A}_p \mid \varphi(a_{\forall},a_p'') = 1\}}}\nu(a_{\forall},a_p')\cdot(\deno{d}a_{\forall})(v')} \\
                       &= \frac{\nu(a_{\forall},a_p)\cdot(\deno{d}a_{\forall})(v)}{\displaystyle \smashoperator[r]{\sum_{a_p' \in \{a_p' \in \mathit{A}_p \mid \varphi(a_{\forall},a_p') = 1\}}} \nu(a_{\forall},a_p')}\\
                       &= (\deno{d}a_{\forall})(v) \cdot \frac{\displaystyle \smashoperator[r]{\sum_{a_p' \in \{a_p' \in \mathit{A}_p \mid a_p' = a_p \wedge \varphi(a_{\forall},a_p') = 1\}}}\nu(a_{\forall},a_p' )}{\displaystyle \smashoperator[r]{\sum_{a_p' \in \{a_p' \in \mathit{A}_p \mid \varphi(a_{\forall},a_p') = 1\}}} \nu(a_{\forall},a_p')}\\
                         &= (\deno{d}a_{\forall})(v) \cdot \frac{\displaystyle\Pr_{a_p' \sim \nu(a_{\forall})}[a_p' = a_p \wedge \varphi(a_{\forall},a_p') = 1]}{\displaystyle\Pr_{a_p' \sim \nu(a_{\forall})}[\varphi(a_{\forall},a_p') = 1]}\\
                       &= (\deno{d}a_{\forall})(v) \cdot \nu_{\varphi}(a_{\forall},a_p)
    \end{align*}
  \end{proof}
  We can now prove that $R'$ satisfies~\Cref{def:main}. Note that by~\Cref{def:sampling}, 
  \[
    \mu_{\mathtt{x} \sim d}(a_{\forall},m) = \sum_{m' \in \mathsf{desample}(m)} (\deno{d}a_{\forall})(m(\mathtt{x})) \cdot \mu(a_{\forall},m').
  \]
  for all $a_{\forall} \in \mathit{A}_{\forall}$ and $m \in \mathit{M}$. Also, by~\Cref{def:main},
  \[
    \mu(a_{\forall},m) = \sum_{a_p \in \mathsf{fromMem}(\sigma,a_{\forall},m)} \nu_{\varphi}(a_{\forall},a_p).
  \]
  and so,
  \begin{align*}
    \mu_{\mathtt{x} \sim d} &= \sum_{m' \in \mathsf{desample}(m)} (\deno{d}a_{\forall})(m(\mathtt{x})) \cdot \mu(a_{\forall},m')\\
                            &= \sum_{m' \in \mathsf{desample}(m)} (\deno{d}a_{\forall})(m(\mathtt{x})) \cdot \left( \sum_{a_p \in \mathsf{fromMem}(\sigma,a_{\forall},m')} \nu_{\varphi}(a_{\forall},a_p) \right)\\
                            &= \sum_{m' \in \mathsf{desample}(m)} \sum_{a_p \in \mathsf{fromMem}(\sigma,a_{\forall},m')} (\deno{d}a_{\forall})(m(\mathtt{x})) \cdot \nu_{\varphi}(a_{\forall},a_p).
  \end{align*}
  Now it suffices to show the following:
  \begin{enumerate}
  \item For all $a_{\forall}\in \mathit{A}_{\forall}$ and $a_p \in \mathit{A}_p$,
    \[
      \nu_{\varphi}''(a_{\forall},a_p) = (\deno{d}a_{\forall})(v) \cdot \nu_{\varphi}(a_{\forall},a_p')
    \]
    where $\nu_{\varphi}''$ is defined as in~\Cref{def:main} but with $R'$ and $a_p = a_p' \cup \{\delta_{k+1} \mapsto v\}$. This immediately follows from~\Cref{lem:split}.
  \item Note that if $m_1,m_2 \in \mathit{M}$ where $m_1\neq m_2$ then for all $a_{\forall} \in \mathit{A}_{\forall}$ and $\sigma$, $\mathsf{fromMem}(\sigma,a_{\forall},m_1)$ and $\mathsf{fromMem}(\sigma,a_{\forall},m_2)$ are disjoint. Thus, it suffices to show that
    \[
      \mathsf{fromMem}(\sigma',a_{\forall},m) = \bigcup_{m' \in \mathsf{desample}(m)} \mathsf{fromMem}(\sigma,a_{\forall},m').
    \]
  \end{enumerate}
  For (2), we will show this through double containment. First, let $a_p \in \mathsf{fromMem}(\sigma',a_{\forall},m)$ be arbitrary. By definition,
  $m = \mathsf{toMem}(\sigma',a_{\forall},a_p) = \lambda (\mathtt{y} : \mathcal{X}) . \deno{\sigma'(\mathtt{y})}a_{\forall}a_p$. Since $\sigma'(\mathtt{x}) = \delta_{k+1}$, where $\delta_{k+1}$ is a fresh probabilistic symbolic variable. Therefore, $a_p(\delta_{k+1}) = m(\mathtt{x})$. We need to show that $a_p \in \mathsf{fromMem}(\sigma,a_{\forall},m')$ for some $m' \in \mathsf{desample}(m)$. Note that $m \in \mathsf{desample}(m)$. Furthermore, for every program variable $\mathtt{y} \in \mathcal{X}$, $\delta_{k+1}$ does not appear in $\sigma(\mathtt{y})$. Therefore, $a_p \in \mathsf{fromMem}(\sigma,a_{\forall},m)$ as
  \[
    \mathsf{fromMem}(\sigma,a_{\forall},m) = \{ a_p' \in \mathit{A}_p \mid \forall (\mathtt{y} : \mathcal{X})~.~(\mathtt{y} \neq \mathtt{x} \wedge a_p'(\mathtt{y}) = a_p(\mathtt{y}))\}.
  \]
  Since $a_p$ was arbitrary, we have proven the forward direction of the containment.

  Now, let $a_p \in \mathsf{fromMem}(\sigma,a_{\forall},m')$ for some $m' \in \mathsf{desample}(m)$. Since for all $\mathtt{y} \in \mathcal{X}$, $\delta_{k+1}$ does not appear in $\sigma(\mathtt{y})$, $a_p(\delta_{k+1})$ can be any value. To prove the reverse containment, it suffices to show that $m = \mathsf{toMem}(\sigma',a_{\forall},a_p)$. Note that for every $a_p' \in \mathit{A}_p$ where $n = \mathsf{toMem}(\sigma',a_{\forall},a_p')$, $n(\mathtt{x}) = a_p'(\delta_{k+1})$. Furthermore, for every $\mathtt{y} \in \mathcal{X}$ where $\mathtt{y} \neq \mathtt{x}$, $m(\mathtt{y}) = m'(\mathtt{y})$ and $\sigma'(\mathtt{y}) = \sigma(\mathtt{y})$. This implies that $\mathsf{toMem}(\sigma,a_{\forall},a_p)(\mathtt{y}) = m'(\mathtt{y}) = m(\mathtt{y}) = \mathsf{toMem}(\sigma',a_{\forall},a_p)$.

  Without loss of generality, assume that $m(\mathtt{x}) = z$. Since $a_p(\mathtt{x})$ can be any value, let this particular value be $z$. Thus, $m(\mathtt{x}) = a_p(\delta_{k+1})$. Therefore, $m = \mathsf{toMem}(\sigma',a_{\forall},a_p)$ for all program variables. This completes the double containment.
\end{proof}

\soundnessBranching*
Before proving~\Cref{thm:branch}, we will prove a lemma.
\begin{lemma}
  \label{lem:condition}
  Let $a_{\forall} \in \mathit{A}_{\forall}$ and $a_p \in \mathit{A}_p$ be arbitrary. Then,
  \[
    \nu_{\varphi}(a_{\forall},a_p) = \Pr_{a_p' \sim \nu(a_{\forall})}[ a_p' = a_p \mid \varphi(a_{\forall},a_p') = 1 \mid \deno{c_{sym}} a_{\forall} a_p'= 1 ]
  \]
  and
  \[
    \nu_{\varphi}^{\neg}(a_{\forall},a_p) = \Pr_{a_p' \sim \nu(a_{\forall})}[ a_p' = a_p \mid \varphi(a_{\forall},a_p') = 1 \mid \deno{c_{sym}} a_{\forall} a_p' = 0 ].
  \]
\end{lemma}

\begin{proof}
  Let $a_{\forall} \in \mathit{A}_{\forall}$ and $a_p \in \mathit{A}_p$ be arbitrary. Consider the quantity
  \[
    \Pr_{a_p' \sim \nu(a_{\forall})}[ a_p' = a_p \mid \varphi(a_{\forall},a_p') = 1 \mid \deno{c_{sym}} a_{\forall} a_p'= 1 ].
  \]
  By definition of conditional probability,
  \begin{align*}
    \Pr_{a_p' \sim \nu(a_{\forall})}[ a_p' = a_p &\mid \varphi(a_{\forall},a_p') = 1 \mid \deno{c_{sym}} a_{\forall} a_p'= 1 ]\\&= \frac{\displaystyle\Pr_{a_p' \sim \nu(a_{\forall})}[a_p' = a_p \wedge \varphi(a_{\forall},a_p') = 1 \mid \deno{c_{sym}} a_{\forall} a_p' = 1 ]}{\displaystyle\Pr_{a_p' \sim \nu(a_{\forall})}[\varphi(a_{\forall},a_p') = 1 \mid \deno{c_{sym}} a_{\forall} a_p' = 1]} \\
                                                                                           &= \frac{\frac{\displaystyle\Pr_{a_p' \sim \nu(a_{\forall})}[a_p' = a_p \wedge \varphi(a_{\forall},a_p') = 1 \wedge \deno{c_{sym}} a_{\forall} a_p' = 1 ]}{\displaystyle\Pr_{a_p' \sim \nu(a_{\forall})}[ \deno{c_{sym}} a_{\forall} a_p' = 1]}}{\frac{\displaystyle\Pr_{a_p' \sim \nu(a_{\forall})}[\varphi(a_{\forall},a_p') = 1 \wedge \deno{c_{sym}} a_{\forall} a_p' = 1 ]}{\displaystyle\Pr_{a_p' \sim \nu(a_{\forall})}[ \deno{c_{sym}} a_{\forall} a_p' = 1]}} \\
                                                                                           &= \frac{\displaystyle\Pr_{a_p' \sim \nu(a_{\forall})}[a_p' = a_p \wedge \varphi(a_{\forall},a_p') = 1 \wedge \deno{c_{sym}} a_{\forall} a_p' = 1 ]}{\displaystyle\Pr_{a_p' \sim \nu(a_{\forall})}[\varphi(a_{\forall},a_p') = 1 \wedge \deno{c_{sym}} a_{\forall} a_p' = 1 ]} \\
                                                                                           &= \nu_{\varphi}(a_{\forall},a_p).
  \end{align*}
  The proof of the second quantity, $\nu_{\varphi}^{\neg}$, follows identically as above.
\end{proof}
Now we can prove~\Cref{thm:branch}.
\begin{proof}
  First assume that $\mu$ satisfies the input abstraction $R=(\varphi,\sigma,P)$. Note that by definition, for all program memories $m \in \mathit{M}$, and assignments to universal symbolic variables $a_{\forall} \in \mathit{A}_{\forall}$,
  \[
    \mu_c(a_{\forall},m) = \frac{\displaystyle\Pr_{m' \sim \mu(a_{\forall})}\left[ m' = m \land \deno{c}m' = 1 \right]}{\displaystyle\Pr_{m' \sim \mu(a_{\forall})}\left[ \deno{c}m' = 1 \right]}\quad\text{and}\quad \mu_{\neg c}(a_{\forall},m) = \frac{\displaystyle\Pr_{m' \sim \mu(a_{\forall})}\left[ m' = m \land \deno{c}m' = 0 \right]}{\displaystyle\Pr_{m' \sim \mu(a_{\forall})}\left[ \deno{c}m' = 0 \right]}.
  \]
  Also note that by definition, $\mu(a_{\forall},m) = \sum_{a_p \in \mathsf{fromMem}(\sigma,a_{\forall},m)} \nu_{\varphi}(a_{\forall},a_p)$. Combining together, we get,
  \begin{align*}
    \mu_c(a_{\forall},m) &= \frac{\displaystyle\Pr_{m' \sim \mu(a_{\forall})}\left[ m' = m \land \deno{c}m' = 1 \right]}{\displaystyle\Pr_{m' \sim \mu(a_{\forall})}\left[ \deno{c}m' = 1 \right]} \\
                 &= \frac{\displaystyle\sum_{m' \in \{m' \in \mathit{M} \mid m' = m \wedge \deno{c}m' = 1\}}\mu(a_{\forall},m')}{\displaystyle\sum_{m' \in \{m' \in \mathit{M} \mid \deno{c}m' = 1\}}\mu(a_{\forall},m')}\\
                 &= \frac{\displaystyle\sum_{m' \in \{m' \in \mathit{M} \mid m' = m \wedge \deno{c}m' = 1\}}\sum_{a_p \in \mathsf{fromMem}(\sigma,a_{\forall},m')} \nu_{\varphi}(a_{\forall},a_p)}{\displaystyle\sum_{m' \in \{m' \in \mathit{M} \mid \deno{c}m' = 1\}}\sum_{a_p \in \mathsf{fromMem}(\sigma,a_{\forall},m')} \nu_{\varphi}(a_{\forall},a_p)}\\
                 &= \frac{\displaystyle\sum_{m' \in \{m' \in \mathit{M} \mid m' = m \wedge \deno{c}m' = 1\}}\sum_{a_p \in \mathsf{fromMem}(\sigma,a_{\forall},m')} \frac{\displaystyle\Pr_{a_p' \sim \nu(a_{\forall})}[a_p' = a_p \wedge \varphi(a_{\forall},a_p')=1]}{\displaystyle\Pr_{a_p' \sim \nu(a_{\forall})}[\varphi(a_{\forall},a_p')=1]}}{\displaystyle\sum_{m' \in \{m' \in \mathit{M} \mid \deno{c}m' = 1\}}\sum_{a_p \in \mathsf{fromMem}(\sigma,a_{\forall},m')} \frac{\displaystyle\Pr_{a_p' \sim \nu(a_{\forall})}[a_p' = a_p \wedge \varphi(a_{\forall},a_p')=1]}{\displaystyle\Pr_{a_p' \sim \nu(a_{\forall})}[\varphi(a_{\forall},a_p')=1]}}\\
                 &= \frac{\displaystyle\sum_{m' \in \{m' \in \mathit{M} \mid m' = m \wedge \deno{c}m' = 1\}}\sum_{a_p \in \mathsf{fromMem}(\sigma,a_{\forall},m')} \Pr_{a_p' \sim \nu(a_{\forall})}[a_p' = a_p \wedge \varphi(a_{\forall},a_p')=1]}{\displaystyle\sum_{m' \in \{m' \in \mathit{M} \mid \deno{c}m' = 1\}}\sum_{a_p \in \mathsf{fromMem}(\sigma,a_{\forall},m')} \Pr_{a_p' \sim \nu(a_{\forall})}[a_p' = a_p \wedge \varphi(a_{\forall},a_p')=1]}
  \end{align*}
  Now consider the denominator of the above expression. The left-most summation is collecting all the memories which satisfy the guard, $c$. From these memories, an associated assignment to probabilistic variables is created using $\sigma$. Since $c_{sym} = \deno{c} \sigma$, it must be that $\deno{c_{sym}} a_{\forall} a_p = 1$. Therefore, the denominator is simply computing $\Pr_{a_p \sim \nu(a_{\forall})}[\varphi(a_{\forall},a_p) = 1 \wedge \deno{c_{sym}} a_{\forall} a_p = 1]$. As for the numerator, we can apply a similar argument to conclude that\begin{multline*}
    \sum_{m' \in \{m' \in \mathit{M} \mid m' = m \wedge \deno{c}m' = 1\}}\sum_{a_p \in \mathsf{fromMem}(\sigma,a_{\forall},m')} \Pr_{a_p' \sim \nu(a_{\forall})}[a_p' = a_p \wedge \varphi(a_{\forall},a_p')=1] \\= \sum_{a_p \in \mathsf{fromMem}(\sigma,a_{\forall},m)}\Pr_{a_p' \sim \nu(a_{\forall})}[a_p' = a_p \wedge \varphi(a_{\forall},a_p') = 1 \wedge \deno{c_{sym}} a_{\forall} a_p' = 1].
  \end{multline*}
  Thus, we have
  \begin{equation*}
    \begin{split}
      &\frac{\displaystyle\sum_{m' \in \{m' \in \mathit{M} \mid m' = m \wedge \deno{c}m' = 1\}}\sum_{a_p \in \mathsf{fromMem}(\sigma,a_{\forall},m')} \Pr_{a_p' \sim \nu(a_{\forall})}[a_p' = a_p \wedge \varphi(a_{\forall},a_p')=1]}{\displaystyle\sum_{m' \in \{m' \in \mathit{M} \mid \deno{c}m' = 1\}}\sum_{a_p \in \mathsf{fromMem}(\sigma,a_{\forall},m')} \Pr_{a_p' \sim \nu(a_{\forall})}[a_p' = a_p \wedge \varphi(a_{\forall},a_p')=1]}\\
      &= \sum_{a_p \in \mathsf{fromMem}(\sigma,a_{\forall},m)}\frac{\displaystyle\Pr_{a_p' \sim \nu(a_{\forall})}[a_p' = a_p \wedge \varphi(a_{\forall},a_p') = 1 \wedge \deno{c_{sym}} a_{\forall} a_p' = 1]}{\displaystyle\Pr_{a_p' \sim \nu(a_{\forall})}[\varphi(a_{\forall},a_p') = 1 \wedge \deno{c_{sym}} a_{\forall} a_p' = 1]}\\
      &= \sum_{a_p \in \mathsf{fromMem}(\sigma,a_{\forall},m)} \nu_{\varphi}(a_{\forall},a_p)
    \end{split}
  \end{equation*}
  We have thus proven condition 2 for $\mu_c$. An symmetric argument would apply to $\mu_{\neg c}$. Therefore, $\mu_c$ satisfies $R_{true}$ and $\mu_{\neg c}$ satisfies $R_{false}$.

  We will now prove the added condition that for all $a_{\forall} \in \mathit{A}_{\forall}$, 
  \[
    \sum_{m \in \{m \in \mathit{M} \mid \deno{c}m = 1\}}\mu(a_{\forall},m) = \deno{p_c} a_{\forall}
  \]
  and
  \[
    \sum_{m \in \{m \in \mathit{M} \mid \deno{c}m = 0\}}\mu(a_{\forall},m) = \deno{p_c'} a_{\forall}.
  \]

  Consider the left-hand side of the proposed equality. By~\Cref{def:main},
  \begin{align*}
    \sum_{m \in \{m \in \mathit{M} \mid \deno{c}m = 1\}}\mu(a_{\forall},m) &= \sum_{m \in \{m \in \mathit{M} \mid \deno{c}m = 1\}} \sum_{a_p \in \mathsf{fromMem}(\sigma,a_{\forall},m)} \nu_{\varphi}(a_{\forall},a_p)\\
                                                         &= \sum_{m \in \{m \in \mathit{M} \mid \deno{c}m = 1\}} \sum_{a_p \in \mathsf{fromMem}(\sigma,a_{\forall},m)} \Pr_{a_p' \sim \nu(a_{\forall})}[a_p' = a_p \mid \varphi(a_{\forall},a_p') = 1] 
  \end{align*}
  Note that by~\Cref{line:guard_convert} of~\Cref{alg:symb_ex}, $c_{sym} = \deno{c} \sigma$. Note that for every $a_p \in \mathsf{fromMem}(\sigma,a_{\forall},m)$, $m = \lambda(\mathtt{x} : \mathcal{X}) . \deno{\sigma(\mathtt{x})} a_{\forall} a_p$. Therefore, since $\deno{c} m = 1$, and $\deno{c} (\lambda(\mathtt{x} : \mathcal{X}) . \deno{\sigma(\mathtt{x})} a_{\forall} a_p) = \deno{c_{sym}} a_{\forall} a_p$,
  then $\deno{c_{sym}} a_{\forall} a_p = 1$. Hence,
  \begin{multline*}
    \sum_{m \in \{m \in \mathit{M} \mid \deno{c}m = 1\}} \sum_{a_p \in \mathsf{fromMem}(\sigma,a_{\forall},m)} \Pr_{a_p' \sim \nu(a_{\forall})}[a_p' = a_p \mid \varphi(a_{\forall},a_p') = 1]\\
    = \Pr_{a_p \sim \nu(a_{\forall})}[\deno{c_{sym}} a_{\forall} a_p \mid \varphi(a_{\forall},a_p') = 1].
  \end{multline*}
  Now, using standard transformations, we have
  \begin{align*}
    \Pr_{a_p \sim \nu(a_{\forall})}[\deno{c_{sym}} a_{\forall} a_p \mid \varphi(a_{\forall},a_p') = 1] &= \frac{\Pr_{a_p \sim \nu(a_{\forall})}[\deno{c_{sym}} a_{\forall} a_p = 1 \wedge \varphi(a_{\forall},a_p) = 1]}{\Pr_{a_p \sim \nu(a_{\forall})}[\varphi(a_{\forall},a_p) = 1]}\\
                                                                       &= \frac{\sum_{a_p \in \{a_p \mid \deno{c_{sym}} a_{\forall} a_p = 1 \wedge \varphi(a_{\forall},a_p) = 1\}}\nu(a_{\forall},a_p)}{\sum_{a_p \in \{a_p \mid \varphi(a_{\forall},a_p) = 1\}}\nu(a_{\forall},a_p)}\\
                                                                       &= \frac{\sum_{a_p \in \{a_p \mid \deno{c_{sym}} a_{\forall} a_p = 1 \wedge \varphi(a_{\forall},a_p) = 1\}}\nu(a_{\forall},a_p)}{\sum_{a_p \in \{a_p \mid \varphi(a_{\forall},a_p) = 1\}}\nu(a_{\forall},a_p)}
  \end{align*}
  Now consider the right-hand side of the proposed equality. By Algorithm~\ref{alg:branch},
  \[
    p_c = \frac{\displaystyle\sum_{(v_1,\ldots,v_n) \in \mathrm{domain}(D_1) \times \cdots\times \mathrm{domain}(D_n)} [(c_{sym} \wedge \varphi)\{\delta_1 \mapsto v_1,\ldots,\delta_n \mapsto v_n\}]}{\displaystyle\sum_{(v_1,\ldots,v_n) \in \mathrm{domain}(D_1) \times \cdots\times \mathrm{domain}(D_n)}[\varphi\{\delta_1 \mapsto v_1,\ldots,\delta_n \mapsto v_n\}]}.
  \]
  Note that the substitution of the probabilistic symbolic variables $\delta_1,\ldots,\delta_n$, along with the assignment of universal symbolic variables, $a_{\forall}$, represents an assignment of probabilistic symbolic variables, $a_p$. Therefore,
  \[
    \frac{\sum_{a_p \in \{a_p \mid \deno{c_{sym}} a_{\forall} a_p = 1 \wedge \varphi(a_{\forall},a_p) = 1\}}\nu(a_{\forall},a_p)}{\sum_{a_p \in \{a_p \mid \varphi(a_{\forall},a_p) = 1\}}\nu(a_{\forall},a_p)} = \deno{p_c} a_{\forall}.
  \]
  Thus,
  \[
    \sum_{m \in \{m \in \mathit{M} \mid \deno{c}m = 1\}}\mu(a_{\forall},m) = \deno{p_c} a_{\forall}.
  \]
  A symmetric argument as above can be applied to the equality
  \[
    \sum_{m \in \{m \in \mathit{M} \mid \deno{c}m = 0\}}\mu(a_{\forall},m) = \deno{p_{c}'} a_{\forall}.
  \]
\end{proof}

%\section{Path filtering: further details}
%\label{app:filtering} 
%Hoisted Up.
%
%\begin{table*}
%  \centering
%	\caption{Applying path filtering to a subset of the case studies, compared to the baselines in~\Cref{sec:case_studies}. ``Path Reduction (\%)'' is the percentage decrease in the number of paths, ``Min. Error Bound'' is the minimum $\varepsilon$ which allows the target query to be verified, and ``Speedup (\%)'' is the percentage decrease in \textbf{total} time taken for verification.}
%	\label{tab:heuristics}
%	\begin{tabular}{@{}lrrrrrrr@{}}
%		\toprule
%    & \multicolumn{3}{c}{Timing (sec.)} &&\multicolumn{1}{c}{\textbf{Path}}&\multicolumn{1}{c}{\textbf{Min. Error}}\\ \cmidrule{2-4}
%    \textbf{Case Study} &\textbf{KLEE} & \textbf{Z3} & \textbf{Total} & \textbf{Paths} & \textbf{Reduction (\%)} & \textbf{Bound} & \textbf{Speedup (\%)}\\ \midrule
%		Reservoir Sampling & 14 & 90 & \textbf{104} & 63 & 50.4\%& $0.08$ & 19.4\% \\
%		Monotone Testing & 5 & 3 & \textbf{8} & 16 & 50.0\% & 0 & 96.0\%\\
%		Quicksort & 20 & 182 & \textbf{202} & 71 & 40.8\% & $0.08$ & 45.9\%\\
%		Bloom Filter & 16 & 270 & \textbf{286} & 56 & 32.5\% & $0.18$ & 30.2\%\\
%		\bottomrule
%	\end{tabular}
%\end{table*}
%
%\Cref{tab:heuristics} presents the experimental results from running the filtered versions of a Subset of the case studies presented in~\Cref{sec:case_studies}. 
%
% The Count-min sketch and Freivalds' algorithm case studies were excluded from testing due to the relatively few number of paths present resulting in a limited ability to derive an effective partitioning scheme.

\section{Path Filtering}
\label{app:filtering}

Here we provide additional details related to the path filtering optimization.
Recall that we define a partitioning function $\Gamma: A_\forall \times A_p \rightarrow \mathbb{N}$ which maps an assignment function of the universal symbolic variables, $a_\forall \in A_\forall$ and an assignment function of the probabilistic symbolic variables, $a_p \in A_p$ to a natural number, indicating which \textit{abstract} path $a_\forall$ and $a_p$ combined produce.

Let $\Gamma^{[n]}_\forall$ and $\Gamma^{[n]}_p$ denote the set of universal and probabilistic variable states (resp.) that classify the resulting path as a member of the abstract path $n$, and let $[\vec{y} \sim Prog(\vec{x})]$ be the event that an execution of $Prog$ from $\vec{x}$ results in a state $\vec{y}$ for the probabilistic variables.
An abstract path $n$ can be eliminated if the probability for every universal symbolic variable states is less than some threshold $\epsilon$: $\text{max}_{\vec{x} \in \Gamma_\forall^{[n]}} \Pr[ \vec{y} \sim Prog(x) \land \vec{y} \in \Gamma_p^{[n]} ] < \epsilon$.
\begin{algorithm}[h]
	\caption{Path Filtering}
	\label{alg:path_filter}
	\begin{algorithmic}[1]
        \Function{PathFilter}{$\Gamma^{[n]}_\forall, \Gamma^{[n]}_p, k_\forall, k_p, N, \epsilon$}
        \State{$h \gets \emptyset$} \label{alg:histready}
        \For{$n \gets 1, N$} \label{alg:nloop}
        \State{$count \gets \emptyset$}
        \For{$u \gets 1, k_\forall$}
        \State{$\vec{x} \sim \Gamma^{[n]}_\forall$}
        \For{$v \gets 1, k_p$} \label{alg:loopstart}
        \State{$\vec{y} \gets \textsc{GetProbSamples}(\textrm{Prog}(\vec{x}))$}
        %\State{$n \gets \Gamma(\vec{x}, \vec{y})$}
        \State{$count[\vec{x}] \gets count[\vec{x}] + [\vec{y} \in \Gamma^{[n]}_p]$}
		\EndFor \label{alg:loopend}
		\EndFor
        \State{$h[n] \gets \textrm{ max}(count)$}
		\EndFor
        \State{$pths \gets \{n \mid h[n] / k_p \geq \epsilon \}$}
		\State\Return{$pths$}
		\EndFunction
	\end{algorithmic}
\end{algorithm}
Algorithm~\ref{alg:path_filter} shows our path-filtering algorithm: we invoke it with the partitioning function ($\Gamma^{[n]}_\forall, \Gamma^{[n]}_p$), number of abstract paths ($N$), and hyperparameters $k_\forall, k_p$ that control the number of universal symbolic variable settings to be evaluated and the number of stochastic runs for estimation of probability. 
The algorithm starts by initializing the distribution histogram $h$ to capture the maximum value of the counts (Line~\ref{alg:histready}).
Next, it runs over all the abstract paths $1\dots N$ (Line~\ref{alg:nloop}).
The algorithm repeatedly samples a random program input $\vec{x}$ uniformly from those corresponding to abstract path $n$, and then makes multiple stochastic runs on $Prog(\vec{x})$, recording the setting of the probabilistic variables $\vec{y}$ and counting how often $(\vec{x}, \vec{y})$ belong to abstract path $n$ (Lines \ref{alg:loopstart}-\ref{alg:loopend}).
Finally, the maximum count for any of the program input states $\vec{x}$ is recorded in $h[n]$.
Once all abstract paths are handled, we filter the paths based on the user-defined threshold $\epsilon$ and the distribution of counts in $h$.
Our partitioning function must be \emph{path disjoint}: a single control-flow path cannot map to different abstract paths, i.e., the partitioning function must impose condition on the variables such that it induces a partitioning on the control-flow paths.
Since the size of the formula is proportional to the number of control-flow paths, any partitioning scheme that does not partition the control-flow variables will not be useful.
Additionally, the partitioning function should ideally induce an \emph{asymmetric distribution}: in order for pruning to be effective, certain abstract paths should have low probability counts for all settings of universal symbolic variables.

\paragraph*{Experimental Results.}

\Cref{tab:heuristics} presents the experimental results from running the filtered versions of a subset of the case studies presented in~\Cref{sec:case_studies}. 
\begin{table*}
  \centering
	\caption{Applying path filtering to a subset of the case studies, compared to the baselines in~\Cref{sec:case_studies}. ``Path Reduction (\%)'' is the percentage decrease in the number of paths, ``Min. Error Bound'' is the minimum $\varepsilon$ which allows the target query to be verified, and ``Speedup (\%)'' is the percentage decrease in \textbf{total} time taken for verification.}
	\label{tab:heuristics}
	\begin{tabular}{@{}lrrrrrrr@{}}
		\toprule
    & \multicolumn{3}{c}{Timing (sec.)} &&\multicolumn{1}{c}{\textbf{Path}}&\multicolumn{1}{c}{\textbf{Min. Error}}\\ \cmidrule{2-4}
    \textbf{Case Study} &\textbf{KLEE} & \textbf{Z3} & \textbf{Total} & \textbf{Paths} & \textbf{Reduction (\%)} & \textbf{Bound} & \textbf{Speedup (\%)}\\ \midrule
		Reservoir Sampling & 14 & 90 & \textbf{104} & 63 & 50.4\%& $0.08$ & 19.4\% \\
		Monotone Testing & 5 & 3 & \textbf{8} & 16 & 50.0\% & 0 & 96.0\%\\
		Quicksort & 20 & 182 & \textbf{202} & 71 & 40.8\% & $0.08$ & 45.9\%\\
		Bloom Filter & 16 & 270 & \textbf{286} & 56 & 32.5\% & $0.18$ & 30.2\%\\
		\bottomrule
	\end{tabular}
\end{table*}

}{}
\end{document}